\DocumentMetadata{
  lang=en-US,
  pdfversion=2.0,
  testphase={sec,toc,title,table}
}
\documentclass[12pt]{article}

\usepackage{comment}
\usepackage[T1]{fontenc}
\usepackage{calc}
\usepackage{lmodern}
\usepackage{microtype}
\usepackage[margin=1in]{geometry}
\usepackage{amsmath,amssymb,mathtools,amsthm}
\usepackage{booktabs,array,tabularx,longtable,multirow}
\usepackage{float}
\usepackage{graphicx}
\usepackage{enumitem}
\usepackage{xspace}
\usepackage{aliascnt}
\usepackage[numbers,sort&compress]{natbib}
\usepackage{xcolor}
\usepackage{tikz}
\usetikzlibrary{arrows.meta,positioning,fit,calc,decorations.pathreplacing}
\usepackage{pdflscape}
\usepackage{tagpdf}
\usepackage{hyperref}
\usepackage[nameinlink,capitalise,noabbrev]{cleveref}

\newcommand{\MaxNfourDM}{\textnormal{\textsc{Max-N4DM}}\xspace}

\hypersetup{
  colorlinks=true,
  bookmarksopen=true,
  bookmarksopenlevel=2,
  bookmarksnumbered=true,
  linkcolor=blue!45!black,
  citecolor=blue!45!black,
  urlcolor=blue!45!black,
  pdftitle={Symmetric Numerical Three-Dimensional Matching: Intractability and Inapproximability},
  pdfauthor={Zhi-Long Chen and Nicholas G. Hall},
  pdfkeywords={numerical three-dimensional matching, symmetric item classes, strong NP-completeness, unary encoding, approximation gap, L-reduction, APX-hardness},
  pdfdisplaydoctitle=true
}

\setlist{nosep,leftmargin=2em}
\allowdisplaybreaks
\newtheorem{theorem}{Theorem}
\newaliascnt{lemma}{theorem}
\newtheorem{lemma}[lemma]{Lemma}
\aliascntresetthe{lemma}
\newaliascnt{proposition}{theorem}
\newtheorem{proposition}[proposition]{Proposition}
\aliascntresetthe{proposition}
\newaliascnt{corollary}{theorem}
\newtheorem{corollary}[corollary]{Corollary}
\aliascntresetthe{corollary}
\theoremstyle{definition}
\newaliascnt{definition}{theorem}
\newtheorem{definition}[definition]{Definition}
\aliascntresetthe{definition}
\newaliascnt{problem}{theorem}

\aliascntresetthe{problem}
\newaliascnt{remark}{theorem}
\newtheorem{remark}[remark]{Remark}
\aliascntresetthe{remark}
\newaliascnt{observation}{theorem}

\aliascntresetthe{observation}

\newcommand{\N}{\mathbb N}
\newcommand{\NN}{\mathbb N}
\newcommand{\uplusm}{\mathbin{\uplus}}
\newcommand{\mset}[1]{\{\!\!\{#1\}\!\!\}}
\newcommand{\wt}{\operatorname{wt}}
\newcommand{\degH}{\operatorname{deg}_{\mathcal H}}
\newcommand{\OPT}{\operatorname{OPT}}
\newcommand{\val}{\operatorname{val}}
\newcommand{\sat}{\operatorname{sat}}
\newcommand{\pos}[1]{(#1)_{+}}
\newcommand{\ThreeDM}{\textnormal{\textsc{3DM}}\xspace}
\newcommand{\MaxThreeDM}{\textnormal{\textsc{Max-3DM}}\xspace}
\newcommand{\MaxThreeDMThree}{\textnormal{\textsc{Max-3DM-3}}\xspace}
\newcommand{\NthreeDM}{\textnormal{\textsc{N3DM}}\xspace}
\newcommand{\MaxNthreeDM}{\textnormal{\textsc{Max-N3DM}}\xspace}
\newcommand{\SNthreeDM}{\textnormal{\textsc{SN3DM}}\xspace}
\newcommand{\SymNthreeDM}{\textnormal{\textsc{SN3DM}}\xspace}
\newcommand{\MaxSNthreeDM}{\textnormal{\textsc{Max-SN3DM}}\xspace}
\newcommand{\APX}{\textnormal{APX}\xspace}
\newcommand{\NP}{\textnormal{NP}\xspace}
\newcommand{\Rset}{\mathcal R}
\newcommand{\Pairset}{\mathcal P}
\newcommand{\barrier}{\bot}
\newcommand{\Kvec}{\mathbf K}
\newcommand{\wvec}{\boldsymbol\omega}
\newcommand{\omegavec}{\boldsymbol\omega}

\newcommand{\VersionNumber}{1.0}
\newcommand{\VersionDate}{July 31, 2026}

\numberwithin{equation}{part}
\numberwithin{table}{part}
\numberwithin{figure}{part}
\numberwithin{lemma}{part}
\numberwithin{theorem}{part}
\numberwithin{definition}{part}
\numberwithin{proposition}{part}
\numberwithin{corollary}{part}
\numberwithin{problem}{part}
\numberwithin{remark}{part}
\numberwithin{observation}{part}

\tagpdfsetup{
  activate,
  para/tagging,
  activate/spaces,
  table/tagging=false
}

\begin{document}

\hypersetup{pageanchor=false}

\begin{titlepage}
\vspace*{0.55in}

{\centerline{\LARGE\bfseries Symmetric Numerical Three-Dimensional Matching:}}
\vspace{0.22in}
{\centerline{\LARGE\bfseries Intractability and Inapproximability}}

\vspace{1.75in}
\centerline{{\Large Zhi-Long Chen$^*$}}
\vspace{0.3in}
\centerline{{\Large Nicholas G. Hall$^{**}$}}

\vspace{1.2in}
\noindent {\large $^*$ Smith School of Business, University of Maryland; zlchen@umd.edu\par}
\vspace{0.1in}
\noindent {\large $^{**}$ Fisher College of Business, The Ohio State University; hall.33@osu.edu\par}

\vfill

 \centerline{{\Large Version 1.0: July 31, 2026 \par}}

\end{titlepage}

\baselineskip=20pt
\pagestyle{empty}
\newpage
	
\section*{Results at a Glance}

\bigskip
\bigskip

\Large

\noindent Symmetric Numerical Three-Dimensional Matching has three disjoint labeled classes whose weight multisets are identical and one common target. 

\bigskip

\noindent Part~I proves unary \NP-completeness, and hence strong \NP-completeness. 

\bigskip

\noindent Part~II proves a unary perfect-completeness gap, excludes a PTAS unless P=NP, and gives a self-contained polynomial-time $3$-approximation.  

\bigskip

\noindent Part~III gives a standard L-reduction, proving standard APX-hardness and, together with the $3$-approximation, APX-completeness.

\normalsize
	
\newpage

\baselineskip=20pt

\section*{Abstract}
Symmetric Numerical Three-Dimensional Matching (\SNthreeDM) asks whether three disjoint labeled classes that carry identical weight multisets can be partitioned into class-transversal triples having one common target sum.  Across the three parts, the common conceptual problem is role recovery under marginal symmetry: the three destination classes expose identical numerical catalogues, so the asymmetric source roles must be reconstructed internally from incidence structure.  This tutorial develops three complementary hardness results for that symmetry restriction.  Part~I gives a unary-polynomial reduction from classical \NthreeDM.  It represents the three asymmetric source roles by ports inside one common occurrence set, uses private items and a uniquely forced filler system to reserve one main incidence at every port, restores the three output-class labels by bipartite edge coloring, and packages four logical coordinates into positive integers by a no-carry mixed-radix encoding.  Consequently, \SNthreeDM is strongly \NP-complete.

Part~II studies the maximum-cardinality problem \MaxSNthreeDM.  Strong \NP-hardness alone separates a perfect value $M$ only from at most $M-1$ and therefore does not exclude a polynomial-time approximation scheme.  Starting from Petrank's perfect-completeness gap for bounded three-dimensional matching, two explicit numerical compilers yield a unary perfect-completeness gap for \MaxNthreeDM.  A defect-stability lemma for the symmetric construction proves that a symmetric matching of size $13n-d$ yields a source matching of size at least $n-21d$, where $n$ is the cardinality of the given multisets  and $d$ is a symmetric defect.  Hence there is a constant $\varepsilon_{\mathrm{SN}}>0$ for which it is \NP-hard to distinguish perfect symmetric instances from instances whose optimum is at most $(1-\varepsilon_{\mathrm{SN}})$ times perfect.  Thus \MaxSNthreeDM has no polynomial-time approximation scheme unless $\mathrm P=\mathrm{NP}$.  Every maximal legal triple matching is a $3$-approximation, so the problem belongs to \APX.

Part~III supplies a standard approximation-preserving reduction that Part~II does not claim.  An exact pair compiler and a one-live-port separation map degree-three Maximum Three-Dimensional Matching to unary \MaxSNthreeDM with
\[
  \OPT_{\mathrm{SN}}=\Gamma +\OPT_{3\mathrm{DM}},
\]
where $\Gamma$ is a fixed destination offset and every destination solution is decoded with one-for-one optimum-error transfer.  The resulting L-reduction has the explicit conservative constants $\alpha=764$ and $\beta=1$.  Consequently, \MaxSNthreeDM is APX-hard under a standard L-reduction and, together with the $3$-approximation, is APX-complete under that convention.  Parts~II and III strengthen the approximation analysis in different dimensions.  Part~II gives the broader grouped defect-stability theorem and a direct perfect-completeness gap; Part~III gives the actual-optimum inequalities of a standard L-reduction on the specially separated instances it constructs.  Neither result contains the other.  Detailed ``yes'' and ``no'' instances in all three parts expose the constructions and audit the arithmetic.

\vspace{0.35in}
\noindent\textbf{Keywords:} numerical three-dimensional matching; symmetric item classes; strong \NP-completeness; unary encoding; perfect-completeness gap; L-reduction; APX-hardness; PTAS-reduction; mixed-radix encoding; bipartite edge coloring; one-live-port separation.

\newpage
\pagestyle{plain}
\pagenumbering{arabic}
\hypersetup{pageanchor=true}
\setcounter{tocdepth}{1}
\hypertarget{maincontents}{}
\tableofcontents
\par\bigskip
\noindent\textbf{Direct part links:}
\hyperref[part:exact]{Part I: exact intractability}\quad\textbar\quad
\hyperref[part:gap]{Part II: direct approximation gap}\quad\textbar\quad
\hyperref[part:lred]{Part III: standard L-reduction}.
\clearpage

\section{Purpose, Audience, and Reading Paths}\label{sec:purpose}

This document has three linked mathematical purposes.  Part~I establishes the exact intractability of \SNthreeDM, the restriction of Numerical Three-Dimensional Matching in which the three item classes remain disjoint and labeled but have identical numerical weight multisets.  Part~II proves a unary perfect-completeness approximation gap for the natural maximum-cardinality version.  Part~III proves standard APX-hardness by an explicit L-reduction and gives an exact affine relationship between source and destination optima.

The document is also intended as a pedagogical online tutorial on the design and auditing of hardness reductions.  The exact proof explains how asymmetric source roles can be represented inside three numerically identical output classes.  The gap proof explains why an exact reverse implication is not enough for inapproximability: a nearly perfect output solution must be decoded quantitatively, with every missing triple charged to only a constant number of lost source triples.  The L-reduction explains a still stronger reduction requirement: error must be measured from the actual destination optimum for every feasible destination solution, not merely from a theoretical perfect upper bound.  The presentation therefore separates role admissibility, source arithmetic, locality, incidence counting, recoloring, mixed-radix packing, defect repair, affine offsets, and solution decoding rather than treating the large encoded integers as a black box.

The intended audience includes researchers and graduate students in theoretical computer science, combinatorial optimization, discrete mathematics, scheduling complexity, and related areas.  Familiarity with \NP-completeness and polynomial-time reductions is assumed, but prior experience with numerical matching gadgets, perfect-completeness gaps, or L-reductions is not.  Detailed explanations follow the principal figures and tables, and the worked examples are designed to support both conceptual reading and complete arithmetic verification.

The tutorial has two complementary layers.  The \emph{logical layer} explains why each reduction works and identifies the invariant established at every stage.  The \emph{audit layer} records the finite role certificates, polynomial bounds, complete numerical tables, and example calculations needed for independent verification.  A first reading may follow the logical layer and treat the long tables as certificates; a complete verification can then return to the audit layer without changing the conceptual order of the proof.

The work makes three principal contributions.  Part~I proves unary \NP-completeness, and hence strong \NP-completeness, of \SNthreeDM.  Part~II proves a direct constant perfect-completeness gap for \MaxSNthreeDM, the resulting no-PTAS theorem, and membership in \APX.  Part~III gives an L-reduction from \MaxThreeDMThree with constants $\alpha=764$ and $\beta=1$, thereby proving standard APX-hardness and APX-completeness.  Across the three parts, the document develops a reusable framework based on ports, private items, role filters, two-moment locality, degree budgets, copy-forgetting and recoloring, no-carry encoding, deficit equations, occurrence-level repair, pair compilers, one-live-port separation, and exact affine offsets.

\subsection*{A global map of the tutorial}

\begingroup
\renewcommand{\thefigure}{\arabic{figure}}
\begin{figure}[H]
\centering
\begin{tikzpicture}[
  main/.style={draw,rounded corners,align=center,minimum height=1.12cm,text width=4.25cm,fill=blue!3,font=\small},
  branch/.style={draw,rounded corners,align=center,minimum height=1.35cm,text width=5.05cm,fill=gray!5,font=\small},
  result/.style={draw,rounded corners,align=center,minimum height=1.00cm,text width=5.25cm,fill=blue!3,font=\small},
  arr/.style={-{Latex[length=2.3mm]},thick}
]
\node[main] (p1) at (0,0) {Part~I: exact feasibility\\unary and strong \NP-completeness};
\node[branch] (p2) at (-3.45,-3.0) {Part~II: perfect-bound stability\\direct endpoint gap, no PTAS, and a $3$-approximation};
\node[branch] (p3) at (3.45,-3.0) {Part~III: actual-optimum stability\\standard L-reduction and exact affine error transfer};
\node[result] (all) at (0,-5.7) {Combined classification\\\MaxSNthreeDM is APX-complete under the stated convention};
\draw[arr] (p1.south west) -- node[left,font=\scriptsize,align=right]{stabilize defects from\\the perfect upper bound} (p2.north);
\draw[arr] (p1.south east) -- node[right,font=\scriptsize,align=left]{redesign for\\actual-optimum transfer} (p3.north);
\draw[arr] (p2) -- (all);
\draw[arr] (p3) -- (all);
\end{tikzpicture}
\caption{Logical relationship among the three parts.  The two lower branches strengthen the exact reduction in different ways and are not successive containment steps.}\label{fig:three-part-map}
\end{figure}
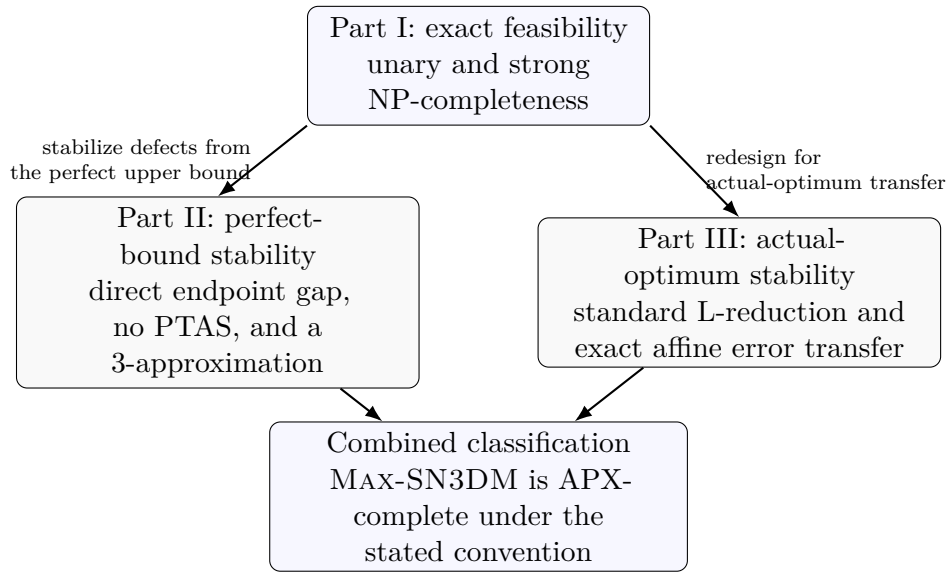
\endgroup

\paragraph{How to read \cref{fig:three-part-map}.}
Part~I is the common exact foundation.  The left branch preserves a constant distance from a theoretical perfect upper bound and supports a direct promise-gap theorem.  The right branch changes the construction so that every feasible destination solution can be compared with the actual destination optimum, as an L-reduction requires.  The arrows into the final box combine the $3$-approximation supplied by Part~II and the standard APX-hardness supplied by Part~III; the two branches should not be read as one theorem subsuming the other.

\subsection*{Proof-dependency matrix}
\addcontentsline{toc}{subsection}{Proof-dependency matrix}

\begingroup
\renewcommand{\thetable}{\arabic{table}}
\scriptsize
\renewcommand{\arraystretch}{1.16}
\begin{longtable}{@{}>{\raggedright\arraybackslash}p{0.14\textwidth}>{\raggedright\arraybackslash}p{0.25\textwidth}>{\raggedright\arraybackslash}p{0.25\textwidth}>{\raggedright\arraybackslash}p{0.25\textwidth}@{}}
\caption{Which mechanisms support exact feasibility, a perfect-bound gap, and actual-optimum error transfer.}\label{tab:proof-dependency}\\
\toprule
Mechanism & Part~I: exact feasibility & Part~II: perfect-bound gap & Part~III: actual-optimum transfer \\
\midrule
\endfirsthead
\caption[]{Which mechanisms support exact feasibility, a perfect-bound gap, and actual-optimum error transfer. (continued)}\\
\toprule
Mechanism & Part~I: exact feasibility & Part~II: perfect-bound gap & Part~III: actual-optimum transfer \\
\midrule
\endhead
\midrule\multicolumn{4}{r}{\emph{continued on next page}}\\
\endfoot
\bottomrule
\endlastfoot
Role recovery & Role codes and ports internalize the source's first, second, and third roles. & The same finite role certificate remains valid when some incidences are missing. & Compiler roles and separated live ports expose exactly the choices needed by the two-stage decoder. \\
No-carry encoding & Four logical coordinates are packed into positive integers without allowing carries. & Both numerical compilers and the symmetric scaffold are scalarized separately; a magnitude ledger preserves unary size. & A six-coordinate pair compiler and a separate four-coordinate symmetric scalarization preserve two exact modules. \\
Locality & The group and group-square coordinates force the private items of every filler into one group. & The same two-moment lock continues to localize fillers in a defective partial solution. & One-live groups and the local filler certificate isolate every charge to one useful port role. \\
Recoloring & A $3$-regular incidence graph decomposes into three perfect matchings that restore the outer labels. & A maximum-degree-three incidence graph is edge-colored to lift a partial hypergraph without cardinality loss. & The same partial recoloring theorem lifts every repaired target hypergraph without changing its value. \\
Defect accounting & Zero defect forces the exact filler multiplicities and leaves one main incidence at every port. & A deficit of $d$ symmetric triples causes at most $21d$ source loss, including occurrence-level repair for repeated values. & Missing fillers bound port overuse; disjoint one-live groups give one-for-one destination-to-source error transfer. \\
Exact pair completion & Not needed: the source is already numerical matching. & Not required: the endpoint compilers preserve perfect completeness but do not create an affine optimum identity. & A fallback contributes one baseline unit, whereas a completed pair contributes two; hence $\OPT_{\mathrm N}=r+\OPT_{3\mathrm{DM}}$. \\
One-live-port separation & Not needed for exact recovery. & Deliberately not imposed, so the theorem applies to the broader grouped image of the original symmetric construction. & Exactly one useful port per group prevents a single local filler deficit from being charged to several source roles. \\
Affine objective identity & The conclusion is a feasibility equivalence. & Error is measured from the perfect bound $M=13n$. & The fixed offset gives $\OPT_{\mathrm{SN}}=36m+r+\OPT_{3\mathrm{DM}}$ and cancels in the second L-inequality. \\
\end{longtable}
\endgroup

\paragraph{How to read \cref{tab:proof-dependency}.}
Read across a row to see how one mechanism changes as the required reverse guarantee becomes stronger.  ``Not needed'' means that the mechanism is unnecessary for that part's stated theorem, not that it is unavailable in principle.  The lower four rows explain the main distinction between Parts~II and III: Part~II controls loss from a perfect upper bound on a broader grouped image, whereas Part~III adds exact pair completion, one-live separation, and an affine offset to compare every solution with the actual optimum.

\clearpage
\begingroup
\renewcommand{\thefigure}{\arabic{figure}}
\begin{figure}[H]
\centering
\begin{tikzpicture}[
  objectbox/.style={draw,rounded corners,align=center,inner sep=7pt,text width=4.1cm,minimum height=1.35cm,font=\small},
  objectarrow/.style={-{Latex[length=2.2mm]},thick},
  labeltext/.style={font=\scriptsize,align=center}
]
\node[objectbox] (source) at (-5.2,0) {bounded three-dimensional matching source\\\scriptsize Parts~II and III begin here};
\node[objectbox] (numerical) at (0,0) {numerical matching instance\\\scriptsize Part~I begins here};
\node[objectbox] (groups) at (5.2,0) {ports and private items in one common occurrence set};
\node[objectbox] (incidence) at (-2.6,-2.8) {degree-three target-sum incidence system};
\node[objectbox] (symmetric) at (2.6,-2.8) {three identical labeled output classes};
\draw[objectarrow] (source) -- (numerical);
\draw[objectarrow] (numerical) -- (groups);
\draw[objectarrow] (groups.south west) -- (incidence.north east);
\draw[objectarrow] (incidence) -- (symmetric);
\end{tikzpicture}
\caption{Object genealogy shared by the three proofs.  The scalar integers package a logical construction whose decisive work is performed by roles, locality, incidence budgets, and the decoder.}\label{fig:object-genealogy}
\end{figure}
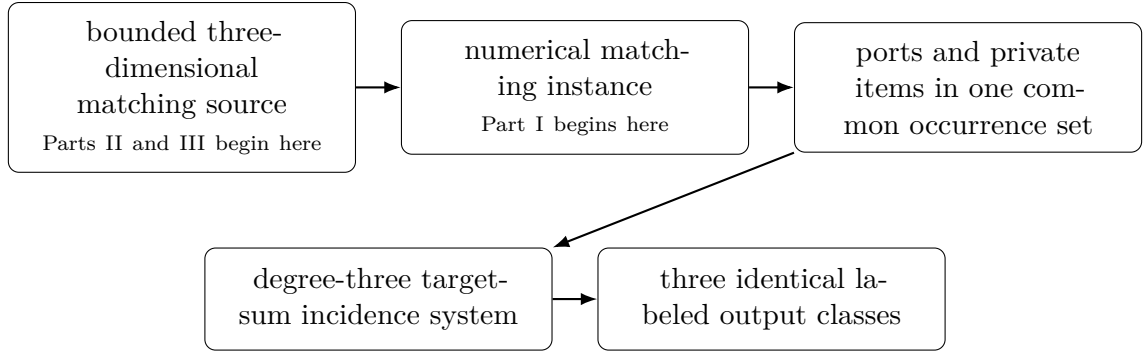
\endgroup

\paragraph{How to read \cref{fig:object-genealogy}.}
The arrows separate tasks that should not be conflated.  The source instance supplies the combinatorial choice.  The numerical layer records that choice in source-sum or pair-completion equations.  The symmetric scaffold hides the asymmetric roles inside one common catalogue copied identically into the three labeled classes.  Finally, the reverse analysis determines whether the result is exact recovery, bounded loss from a perfect upper bound, or error transfer from the actual optimum.

\subsection*{Notation at a glance}
\addcontentsline{toc}{subsection}{Notation at a glance}

The three parts use related but not identical reductions, and several symbols have part-specific roles.  The following table records the principal global quantities, optimum functions, defect parameters, and affine offsets.  Local indices and temporary construction variables remain defined where they are first used.

\begingroup
\renewcommand{\thetable}{\arabic{table}}
\small
\renewcommand{\arraystretch}{1.15}
\begin{longtable}{@{}>{\raggedright\arraybackslash}p{0.11\textwidth}>{\raggedright\arraybackslash}p{0.24\textwidth}>{\raggedright\arraybackslash}p{0.57\textwidth}@{}}
\caption{Principal notation across the three parts.}\label{tab:notation-at-glance}\\
\toprule
Scope & Symbol & Meaning \\
\midrule
\endfirsthead
\caption[]{Principal notation across the three parts. (continued)}\\
\toprule
Scope & Symbol & Meaning \\
\midrule
\endhead
\midrule\multicolumn{3}{r}{\emph{continued on next page}}\\
\endfoot
\bottomrule
\endlastfoot
All parts & $S,S^X,S^Y,S^Z$ & The common underlying occurrence set and its three disjoint labeled output copies.  The three copies have identical weight multisets. \\
All parts & ports and private items & Ports carry source choices; private items close the local filler scaffold and do not occur in main triples. \\
All parts & $F_0,F_1,\ldots,F_9$ & The ten permitted role patterns.  $F_0$ is the main pattern; $F_1,\ldots,F_9$ are filler patterns. \\
Parts~I--II & $A,B,C,t,n$ & The three source occurrence classes, their common target, and the number of source items in each class. \\
Parts~I--II & $\Kvec,K^\star,L$ & The logical vector target, its scalar encoded target, and the mixed-radix base used to prevent carries. \\
Part~II & $M=13n$ & The natural perfect upper bound on the number of symmetric triples in the constructed instance. \\
Part~II & $d$ & The symmetric defect: a partial solution has value $M-d$ and therefore leaves $3d$ output incidences unused. \\
Part~II & $\gamma,\varepsilon_{\mathrm N},\varepsilon_{\mathrm{SN}}$ & Respectively, the bounded matching source-gap constant, the compiled numerical gap, and the final symmetric gap. \\
Part~III & $H=(U,V,W;E),p,k^\star$ & The degree-three source hypergraph, its number of edges, and its maximum matching size. \\
Part~III & $\Pairset,r$ & The set of represented $(u_i,v_j)$ pairs and its cardinality.  Each pair contributes a fixed baseline unit in the exact compiler. \\
Part~III & $\OPT_{3\mathrm{DM}},\OPT_{\mathrm N},$\newline $\OPT_{\mathrm{SN}}$ & The source, intermediate numerical, and final symmetric optimum functions in the two affine identities. \\
Part~III & $m,n=3m$ & The padded numerical class size and the number of separated symmetric groups.  Here $n$ is local to Part~III and is not the Part~I--II source-class size. \\
Part~III & $\Gamma=36m+r;\ \alpha,\beta$ & The fixed destination offset and the two L-reduction constants.  The proof uses $\alpha=764$ and $\beta=1$. \\
\end{longtable}
\endgroup

\paragraph{How to read \cref{tab:notation-at-glance}.}
Begin with the scope column.  A symbol should be interpreted only within the listed part or parts; in particular, $n$ is reused locally in Part~III after the pair compiler has produced padded class size $m$.  The table records quantities that recur across sections.  One-time indices, shifts, and temporary variables remain beside the equations that define them so that this map does not become a second glossary.

\subsection*{Accessibility and table-reading conventions}
\addcontentsline{toc}{subsection}{Accessibility and table-reading conventions}

The distributed PDF is built with document-language and title metadata and a tagged structure tree for headings, the table of contents, paragraphs, and interword spaces.  The diagrams are accompanied by descriptive captions and linear ``How to read'' explanations so that their logical content is available without relying on geometry or color alone.  The tables remain live, selectable, and searchable text.

The largest audit tables combine repeated headers, multi-page continuation logic, mathematical notation, and tightly related columns.  Current automated tagging does not reliably expose every such longtable as a complete cell-level accessibility structure, so this edition does not claim independent PDF/UA certification.  Instead, every dense table is introduced by surrounding proof text and, where useful, followed by an adjacent ``How to read'' paragraph that gives a linear first-pass route through its rows and columns.  The prose route states what the table proves; the table itself remains available for detailed numerical verification.

\subsection*{Terminological convention}
An \emph{item} is a labeled occurrence; different items may have the same numerical weight.  A \emph{source edge} is an edge of a three-dimensional matching instance.  A \emph{numerical triple} is a class-transversal target-sum triple in \NthreeDM or \SNthreeDM.  After the three outer class labels are forgotten, the same selected object is viewed as a three-incidence \emph{hyperedge}.  The adjectives \emph{main} and \emph{filler} describe the object's logical role and may therefore modify either ``triple'' or ``hyperedge'' according to the representation currently being used.  Thus ``edge,'' ``triple,'' and ``hyperedge'' are not global synonyms; the chosen term signals the current representation.

\subsection*{Relationship between Parts II and III}

Parts~II and III answer related but nonidentical questions, and neither theorem should be described as containing the other.  Their main differences are summarized in \cref{tab:partII-partIII}.

\begingroup
\renewcommand{\thetable}{\arabic{table}}
\begin{table}[H]
\centering
\caption{The complementary contributions of Parts II and III.}\label{tab:partII-partIII}
\small
\renewcommand{\arraystretch}{1.18}
\begin{tabularx}{\textwidth}{@{}p{0.19\textwidth}XX@{}}
\toprule
Feature & Part~II: direct gap theorem & Part~III: standard L-reduction \\
\midrule
Source statement & A hard perfect-completeness gap, compiled into unary \MaxNthreeDM & APX-hard degree-three \MaxThreeDMThree \\
Destination family & The broader grouped numerical inputs accepted by the original symmetric construction & Deliberately one-live-port-separated instances produced by the new compiler \\
Reverse guarantee & $13n-d$ symmetric triples yield at least $n-21d$ source triples; error is measured from the perfect upper bound & Every feasible destination solution has a decoded source solution whose error is at most the destination optimum error \\
Characteristic result & An explicit perfect-completeness promise gap and a direct no-PTAS theorem & Exact affine optimum identity, standard APX-hardness, and an explicit PTAS-reduction \\
Contribution profile & Broader grouped image; direct endpoint gap and a transparent quantitative threshold & Separated image family; standard L-reduction and actual-optimum, solution-by-solution transfer \\
\bottomrule
\end{tabularx}
\end{table}
\endgroup

\paragraph{How to read \cref{tab:partII-partIII}.}
First read down each result column to understand one theorem on its own.  Then compare across a row: the source and destination rows describe scope, the reverse-guarantee row states the quantitative invariant, and the final rows state the theorem that should be cited.  The table deliberately does not place Parts~II and III on one strength scale; it shows why the broader endpoint-gap theorem and the more specialized actual-optimum reduction are complementary.

The comparison is dimensional rather than a single ranking.  Part~II does not establish the two affine inequalities of an L-reduction, because its defect is measured from the natural perfect bound rather than from the actual destination optimum.  Part~III changes the construction by separating the source roles so that each local filler deficit is charged at most once; this yields an exact one-for-one decoder, but on the more specialized one-live-port-separated image family.  Conversely, Part~III does not replace Part~II's broader defect-stability theorem or its specifically formulated perfect-completeness promise gap.  A direct gap proof can also give a more transparent, and potentially numerically sharper, explicit inapproximability ratio than the constant inherited through an L-reduction.  Part~II should therefore be cited for the direct endpoint-gap theorem and Part~III for standard APX-hardness and solution-by-solution error transfer.  Neither theorem contains the other.

An L-reduction is a PTAS-transfer mechanism, not by itself a no-PTAS theorem.  Suppose an L-reduction from a maximization problem $A$ to a maximization problem $B$ has constants $\alpha,\beta>0$, and $Y$ is a $(1-\delta)$-approximate solution of the constructed $B$-instance.  Then
\begin{align*}
 \OPT_A-\val_A(g(Y))
 &\le \beta\bigl(\OPT_B-\val_B(Y)\bigr)\\
 &\le \beta\delta\OPT_B
 \le \alpha\beta\delta\OPT_A.
\end{align*}
Hence $\val_A(g(Y))\ge(1-\alpha\beta\delta)\OPT_A$, and choosing $\delta=\varepsilon/(\alpha\beta)$ transfers a PTAS for $B$ to a PTAS for $A$.  The conclusion that Part~III rules out a PTAS therefore uses the approximation hardness of its particular source: bounded Maximum Three-Dimensional Matching is MAX SNP-complete and has no PTAS unless $\mathrm P=\mathrm{NP}$ \citep{Kann1991,ChlebikChlebikova2006}.  With $\alpha=764$ and $\beta=1$, a destination error parameter $\delta=\varepsilon/764$ would yield a source error parameter $\varepsilon$.

Different readers may follow different paths.
\begin{itemize}
  \item Readers interested mainly in the statements and significance of the results may read this front material, the theorem statements near the beginning of each part, and the three concluding sections.
  \item Readers checking the exact reduction should follow Part~I through the role certificate, degree equations, unary-size argument, and its fully worked examples.
  \item Readers interested in the direct inapproximability theorem should read the overview, source-gap transformation, defect-stability lemma, and gap-transfer calculation in Part~II; the one-group examples then provide a compact numerical audit.
  \item Readers interested in standard approximation-preserving reductions should begin with the Part~III proof roadmap, then study the exact pair compiler, one-live-port separation, local filler certificate, affine identity, and the source-to-destination examples.
  \item Readers adapting the methods to another symmetric problem should focus on the shortcuts and checklist below, then compare the zero-defect reverse proof in Part~I, the perfect-bound stability proof in Part~II, and the actual-optimum error transfer in Part~III.
\end{itemize}

The authors' interest in symmetric numerical matching is motivated in part by work on vehicle platooning, where repeated route segments and repeated position-dependent energy effects can produce identical marginal numerical profiles while preserving difficult coordination decisions \citep{ChenHallVehicle2026}.  The present results isolate the abstract matching structure behind that motivation; they do not claim that every application with repeated data has the full hardness construction.

This deposited edition is fixed under the version identifier given below. Corrections and substantive revisions will be issued under new version numbers, with the earlier editions retained for reference. Comments are invited.

\paragraph{2020 Mathematics Subject Classification.}
Primary: 68Q17 (computational difficulty of problems).  Secondary: 68Q15 (complexity classes), 68Q25 (analysis of algorithms and problem complexity), 68W25 (approximation algorithms), 90C27 (combinatorial optimization), 05C70 (factorization and matching).

\section{Introduction}\label{sec:introduction}

Numerical Three-Dimensional Matching (\NthreeDM) is a classical exact-sum matching problem.  An instance has three disjoint labeled classes of equally many item occurrences, a positive integer weight on each occurrence, and one common target.  A feasible solution partitions all occurrences into triples that contain one item from each class and whose weights sum to the target.  The problem is a standard source of strong \NP-completeness reductions because the common target simultaneously enforces three-way assignment and exact arithmetic \citep{GareyJohnson1978,GareyJohnson1979}.

This tutorial studies the restriction in which the three displayed weight multisets are identical.  The restriction is substantial but does not identify the classes: an $X$-occurrence, a $Y$-occurrence, and a $Z$-occurrence remain distinct objects even when they have the same weight.  A direct copy of an ordinary source instance is therefore inadequate.  It creates three copies of every source role in every output class, while the source reduction requires its first, second, and third roles to behave differently.

Part~I resolves the exact problem by moving those asymmetric roles inside one common weighted occurrence set.  Three port types carry the source values, private items regulate local incidence counts, and a finite role-code certificate permits one main pattern and nine filler patterns.  Two locality coordinates prevent private items from different source groups from splicing together.  After the output-class labels are temporarily forgotten, the private-degree equations force two filler incidences at every port and all three incidences at every private item.  One main incidence remains at each port and recovers the source matching.  Bipartite edge coloring then restores the three class labels, while a no-carry mixed-radix encoding converts the vector construction into ordinary positive integers.

Exact intractability is not the end of the analysis.  For the maximum-cardinality version, the decision reduction separates a perfect solution of size $M$ only from a solution of size at most $M-1$.  The relative difference $1/M$ vanishes, so strong \NP-hardness by itself is compatible with a polynomial-time approximation scheme.  Part~II supplies the missing quantitative ingredients.  A hard perfect endpoint for bounded three-dimensional matching is transferred through polynomially bounded numerical encodings, and the exact reverse proof is stabilized by introducing incidence deficits and an explicit repair for repeated-value overuse.  The resulting constant gap gives a direct unary inapproximability theorem and excludes a PTAS unless $\mathrm P=\mathrm{NP}$.

Part~III asks for a different kind of quantitative statement.  A standard L-reduction must compare the decoded source error with error from the \emph{actual} destination optimum for every feasible destination solution.  The Part~II estimate from the theoretical perfect upper bound does not provide that condition.  Part~III therefore introduces an exact pair compiler and one-live-port separation.  The separation ensures that a local filler deficit cannot subsidize excess capacity in several source roles.  This yields the affine identity $\OPT_{\mathrm{SN}}=36m+r+\OPT_{3\mathrm{DM}}$ and a decoder with no multiplicative loss.  The result gives standard APX-hardness while preserving unary size.

Parts~II and III are complementary.  Part~II proves the broader grouped defect-stability theorem and a direct perfect-completeness gap.  Part~III proves the actual-optimum inequalities of a standard L-reduction on a deliberately separated image family.  Neither statement formally contains the other, so both proofs and both families of examples are retained; the precise division of labor is summarized in \cref{tab:partII-partIII}.

The rest of the front material is organized as follows.  \Cref{sec:litreview} reviews the classical problem, exact intractability, and approximability.  \Cref{sec:shortcuts} explains why several tempting reductions or implications are insufficient.  \Cref{sec:reddes} gives a checklist for exact, gap-preserving, and L-reduction-style symmetry-preserving reductions.  Part~I proves strong \NP-completeness and gives full source-to-output examples.  Part~II proves the perfect-completeness gap, the no-PTAS consequence, and a $3$-approximation upper bound.  Part~III gives the standard L-reduction and fully worked ``yes'' and ``no'' instances.  The appendices provide a glossary, acknowledgments, citation information, and the references.

\section{Literature Review}\label{sec:litreview}

\subsection{Numerical Three-Dimensional Matching}

Three-dimensional matching entered the classical catalogue of combinatorial \NP-complete problems through the foundational reduction literature, including Karp's treatment of three-dimensional matching \citep{Karp1972}.  Numerical Three-Dimensional Matching appears in Garey and Johnson as a strongly \NP-complete numerical problem with one common target \citep{GareyJohnson1979}.  Its usefulness as a reduction source comes from combining exact three-way assignment with a uniform arithmetic requirement.

N3DM has been used to transmit common-capacity or synchronization constraints into a variety of structured optimization models.  Examples include two-machine flow shops with delays and unit-time operations \citep{YuHoogeveenLenstra2004}, equal processing times on one machine \citep{WeiYuan2019}, shared multi-processor scheduling \citep{DereniowskiKubiak2017}, periodic due-date scheduling \citep{ChoiMinPark2019}, finite pinwheel scheduling \citep{KanellopoulosEtAl2026}, and optimistic bilevel scheduling on parallel machines \citep{SchauEtAl2026}.  These applications illustrate why N3DM remains useful even when a receiving problem has repeated times, pooled resources, or periodic structure.

Such regularity does not settle the restriction studied here.  Equal processing times, repeated service opportunities, or identical machines in a receiving problem do not require the three input classes of N3DM to carry identical complete weight multisets.  The present symmetry question concerns the marginal numerical data of the matching instance itself, while retaining the disjoint labels and exact-use requirement.

\subsection{Intractability}

Strong \NP-completeness distinguishes hardness that persists when all numerical values are polynomially bounded, or equivalently when the relevant numerical data are written in unary \citep{GareyJohnson1978,GareyJohnson1979}.  This distinction is essential for the present work: the mixed-radix construction must preserve its logical coordinates without hiding exponential magnitude in the encoded integers.

The main structural ingredient beyond numerical encoding is bipartite matching theory.  Hall's theorem supplies perfect matchings in regular bipartite multigraphs \citep{Hall1935}; standard matching theory develops the resulting decompositions \citep{LovaszPlummer1986}; and direct algorithms are available for edge-coloring bipartite multigraphs \citep{Alon2003}.  In Part~I, this machinery restores the three output-class labels after the gadget has been analyzed as an unlabeled degree-three target-sum hypergraph.

The gap closed by Part~I is precise.  Classical N3DM allows its three source classes to carry visibly different numerical roles.  Part~I proves that this visible asymmetry is unnecessary: the source roles can be internalized in one common occurrence set, copied identically into all three labeled classes, and recovered from role, locality, degree, and source-sum constraints.  Thus identical weight multisets do not make exact N3DM tractable.

\subsection{Approximability}

Approximation classes and approximation-preserving reductions were systematized by Papadimitriou and Yannakakis \citep{PapadimitriouYannakakis1991}; a broad reference is Ausiello et al. \citep{AusielloEtAl1999}.  The PCP theorem established the general connection between verification gaps and hardness of approximation \citep{AroraEtAl1998}.  The present proof does not insert a generic PCP construction directly into the numerical gadget.  Instead, it starts from a perfect-completeness gap already located at the endpoint needed by the symmetric construction.

Maximum bounded three-dimensional matching was shown MAX SNP-complete by Kann \citep{Kann1991}.  Petrank proved a hard gap at the perfect-matching endpoint for bounded three-dimensional matching \citep{Petrank1994}.  That endpoint is load-bearing here because both the exact symmetric construction and its quantitative extension preserve perfect completeness.  Related local-improvement work treats three-dimensional matching within set packing and provides approximation context \citep{HurkensSchrijver1989,HazanSafraSchwartz2006}, but it does not resolve the identical-weight, common-target restriction.

Parts~II and III develop two distinct approximation-hardness routes.  Part~II first compiles the bounded perfect-endpoint gap into unary \MaxNthreeDM and then proves stability of the symmetric reverse implication: a small deficit from the perfect destination bound causes only a linearly bounded deficit in the recovered source matching.  Its formal conclusion is a direct perfect-completeness gap and a no-PTAS theorem.  Part~III starts instead from degree-three Maximum Three-Dimensional Matching, whose bounded versions are MAX SNP-complete and remain hard to approximate \citep{Kann1991,ChlebikChlebikova2006}.  An exact numerical compiler and one-live-port separation then establish the two affine inequalities of an L-reduction, giving standard APX-hardness and exact solution-by-solution transfer.  The two routes address different proof obligations; \cref{tab:partII-partIII} records their relationship.

The reductions below combine several established ingredients, but none of those ingredients by itself resolves the symmetry restriction.  Classical N3DM hardness supplies an asymmetric source; edge coloring restores labels only after a suitable incidence structure has been created; endpoint-gap machinery supplies a source separation but not its preservation through the numerical construction; and the abstract definition of an L-reduction supplies the desired inequalities but not the exact compiler needed to prove them.  The new work lies in constructing and auditing the interfaces among these ingredients while forcing all three destination classes to have identical numerical catalogues.

\begingroup
\renewcommand{\thetable}{\arabic{table}}
\small
\renewcommand{\arraystretch}{1.16}
\begin{longtable}{@{}>{\raggedright\arraybackslash}p{0.20\textwidth}>{\raggedright\arraybackslash}p{0.29\textwidth}>{\raggedright\arraybackslash}p{0.43\textwidth}@{}}
\caption{Prior ingredients, their use here, and the remaining step supplied by this tutorial.}\label{tab:prior-ingredients}\\
\toprule
Prior ingredient & Use in this tutorial & New remaining step addressed here \\
\midrule
\endfirsthead
\caption[]{Prior ingredients, their use here, and the remaining step supplied by this tutorial. (continued)}\\
\toprule
Prior ingredient & Use in this tutorial & New remaining step addressed here \\
\midrule
\endhead
\midrule\multicolumn{3}{r}{\emph{continued on next page}}\\
\endfoot
\bottomrule
\endlastfoot
Classical 3DM and N3DM hardness \citep{Karp1972,GareyJohnson1979} & Supplies asymmetric matching and common-target source problems. & Internalize the three asymmetric roles while making the complete numerical catalogue identical in all three output classes. \\
Strong \NP-completeness and unary reductions \citep{GareyJohnson1978,GareyJohnson1979} & Provides the standard criterion for proving that numerical hardness survives unary encoding. & Give explicit polynomial bounds for every shift, radix, encoded value, and output multiplicity in each construction. \\
Hall's theorem and bipartite edge coloring \citep{Hall1935,LovaszPlummer1986,Alon2003} & Restores the forgotten $X,Y,Z$ copy labels after the logical target hypergraph has been constructed. & Build an incidence system whose degrees make recoloring applicable and prove that the reverse decoder survives the forgetting step. \\
Petrank's perfect-completeness gap \citep{Petrank1994} & Supplies a bounded source promise with a constant perfect-endpoint separation. & Preserve the gap through two numerical compilers and a symmetric defect analysis, including repair of repeated numerical values. \\
Bounded Maximum Three-Dimensional Matching hardness \citep{Kann1991,ChlebikChlebikova2006} & Supplies the standard APX-hard source used in Part~III. & Construct an exact pair compiler and a one-live separated symmetric image with a solution-by-solution decoder. \\
L-reductions \citep{PapadimitriouYannakakis1991,AusielloEtAl1999} & Specifies the two optimum and error inequalities required for standard approximation-preserving hardness. & Prove the affine identities, the fixed offset, and the explicit constants $\alpha=764$ and $\beta=1$ rather than inferring them from a perfect upper bound. \\
\end{longtable}
\endgroup

\paragraph{How to read \cref{tab:prior-ingredients}.}
Each row separates three questions: which theorem or technique is inherited, how it is used, and which new obligation remains after that use.  The table is therefore a novelty-boundary map, not a claim that the general ingredients themselves are new.  Reading the third column from top to bottom also gives a compact account of why the symmetry, gap-stability, and L-reduction arguments require distinct proof modules.

\section{Why the Most Natural Shortcuts Fail}\label{sec:shortcuts}

Several tempting arguments are insufficient.  First, simply copying all three source classes into every output class destroys the source roles.  Numerical equality alone does not guarantee that one selected occurrence plays the $A$-role, one the $B$-role, and one the $C$-role.  The construction must encode those roles internally while leaving the three marginal weight lists identical.

Second, one linear locality coordinate does not force three private items to come from a common source group.  A weighted average of two different indices can equal a third.  The additional group-square coordinate supplies a second moment; together the two equations force zero weighted variance and hence a common index for every private filler pattern.

Third, edge coloring is not a cosmetic final step.  Forgetting output-class labels makes the degree analysis possible, but it also suppresses the one-item-from-each-class requirement.  A proper three-edge-coloring of the incidence bipartite multigraph is what assigns the three physical copies back to the incidences without reusing a class at an item or within a triple.

Fourth, strong \NP-hardness of the perfect-matching decision problem does not imply approximation hardness of the maximum-cardinality problem.  The exact reduction distinguishes $M$ from $M-1$, whose ratio tends to one.  A constant relative gap must be supplied by a gap source and preserved through every transformation.

Fifth, the zero-defect reverse proof cannot simply be quoted for nearly perfect solutions.  Missing triples produce nonregular incidence degrees, perturb the filler equations, and may allow repeated numerical values to be requested too often.  The quantitative proof therefore introduces explicit deficits, bounds positive deviations of the filler multiplicities, and deletes only a linearly bounded number of main triples to restore occurrence-level feasibility.  Without that stability analysis, no constant approximation gap follows.

Sixth, neither a direct constant gap nor an L-reduction should be assigned consequences that depend on a missing hypothesis.  A perfect-completeness gap can rule out a PTAS without satisfying the actual-optimum error inequality of an L-reduction.  Conversely, an L-reduction is only a PTAS-transfer mechanism: it rules out a PTAS for the destination only when the chosen source is itself known not to have one.  Part~II and Part~III verify these two routes separately rather than treating their terminology as interchangeable.

\section{A Reduction-Design Checklist}\label{sec:reddes}

The following checklist summarizes the main obligations in an exact, gap-preserving, or L-reduction-style symmetric numerical reduction.
\begin{enumerate}
  \item \textbf{Problem version.} State whether the destination asks for a perfect partition or a maximum-cardinality partial matching, and identify the natural upper bound.
  \item \textbf{Reduction conclusion.} State separately whether the goal is exact hardness, a direct promise gap, a PTAS-reduction, or a standard L-reduction; do not treat these notions as synonyms.
  \item \textbf{Source approximation status.} If a no-PTAS conclusion is transferred, verify that the source itself has no PTAS under the stated complexity assumption.
  \item \textbf{Source roles.} Identify the asymmetric source roles that must survive after the output classes receive identical numerical catalogues.
  \item \textbf{Role representation.} Specify whether roles are carried by codes, ports, multiplicities, positions, targets, or incidence patterns, and prove that unintended roles cannot imitate them.
  \item \textbf{Redundant copies and degree budgets.} Explain how all copies created by symmetry are used and which local constraints reserve the incidences needed for the global decision.
  \item \textbf{Unintended triples.} Give a finite certificate or exhaustive argument for every role pattern that can attain the target.
  \item \textbf{Locality.} Prove that private components cannot splice across source groups; check whether one coordinate is insufficient and whether a second moment is required.
  \item \textbf{Multiplicity locking.} Solve the nonnegative integral degree equations rather than merely displaying the intended filler counts.
  \item \textbf{Repeated values.} Formulate the reverse proof at the occurrence-multiplicity level.  Do not infer equality of indices from equality of weights unless distinctness is an explicit assumption.
  \item \textbf{Class restoration.} If labels are forgotten, identify the theorem or algorithm that restores one item from each labeled class in every triple without loss.
  \item \textbf{Numerical encoding.} Verify positivity, polynomial bounds, shifts, target digits, and a no-carry base so that scalar equality is equivalent to coordinatewise equality.
  \item \textbf{Quantitative stability.} For a direct gap, define defect from the perfect upper bound and prove a constant source-defect bound.  For an L-reduction, measure error from the actual optimum for every destination solution.
  \item \textbf{Charging uniqueness.} Check whether one local deficit can be charged in more than one role.  If it can, strengthen the construction, for example by one-live-port separation, rather than only loosening a global constant.
  \item \textbf{Affine offsets and decoder.} Identify every fixed objective offset, prove the exact optimum identity when claimed, and describe a polynomial decoder rather than only an existential reverse argument.
  \item \textbf{Gap and size ledger.} Track additive losses, expansion factors, approximation constants, and unary magnitudes through the complete reduction chain.
  \item \textbf{Bidirectional examples.} Audit ``yes'' and ``no'' instances.  For a gap proof, display the incidence-deficit ledger; for an L-reduction, display both optimum values, the fixed offset, and the error inequality.
\end{enumerate}

Used systematically, the checklist separates construction from interpretation.  The forward direction must show that intended source solutions can be represented.  The reverse direction must prove that every feasible or nearly feasible destination solution has the required meaning, with exact or quantitatively bounded loss as appropriate.  It also prevents a correct theorem from being described with a stronger reduction label than the proof actually establishes.

\newpage 

\part{Intractability of Symmetric Numerical Three-Dimensional Matching}\label{part:exact}

\setcounter{theorem}{0}
\renewcommand{\thesection}{\arabic{section}}
\renewcommand{\theHsection}{n3dm.sec.\arabic{section}}

\noindent\hyperlink{maincontents}{\small\textit{Return to main contents}}\par\medskip
\section*{Proof navigation for Part I}\addcontentsline{toc}{section}{Proof navigation for Part I}
\textbf{Forward implication.} Begin with a feasible source N3DM matching. For every source triple, insert one main port edge. Around each source group, insert the prescribed local filler edges. The coordinate identities verify that all of these edges attain the vector target, and the degree ledger verifies that every port has two filler incidences and one main incidence while every private item has three filler incidences. The resulting 3-regular target-sum triple system is then recolored to obtain legal triples among the three labeled output classes. The construction is best understood as reserving one of the three incidence slots of every port for the global source decision.

\textbf{Reverse implication.} Begin with any feasible matching of the symmetric output instance and forget its outer class labels. The reverse proof proceeds through four locks. The role coordinate first restricts every edge to one of the intended main or filler patterns. The group and group-square coordinates prevent private items from unrelated groups from mixing. The private-item degree equations then force the exact filler multiplicities, so precisely one incidence of every port remains for a main edge. Finally, the source coordinate proves that those main edges use the source multisets with the required common sum. Edge coloring is used in the opposite conceptual direction to justify that forgetting and restoring labels loses no feasible solutions.

\textbf{What could otherwise go wrong.} Without the role filter, unintended triples could satisfy the arithmetic target. Without the two locality coordinates, private components from different groups could splice together. Without the unique solution of the degree equations, ports might devote too many or too few incidences to main edges. Without the no-carry base choice, an error in one logical coordinate could be concealed by a carry into the next coordinate. The formal proof assigns one of these possible failures to each layer of the construction.

\section{Introduction}

Numerical Three-Dimensional Matching (\NthreeDM) is a classical strongly \NP-complete number problem; it appears as problem SP16 in Garey and Johnson~\cite{GareyJohnson1979}.  An instance has three labeled classes of equally many item occurrences, a positive integer weight for every occurrence, and one target value.  A solution partitions all occurrences into triples that contain one item from each class and have the target sum.

This work studies the following apparently severe restriction.  The classes remain distinct and labeled, but their weight multisets are required to be identical.  In other words, the input consists of three disjoint copies of one weighted occurrence set.  We prove that this symmetry does not make the problem easier under unary encoding.

The main technical point is that literal duplication creates three copies of every role, whereas an ordinary reduction needs different source roles to behave differently.  We resolve this by passing through a degree-constrained triple-system representation.  Forgetting the outer class labels turns a symmetric matching into a target-sum triple system in which every underlying occurrence has degree three.  Conversely, a regular bipartite edge-coloring restores one incidence of each color at every item and at every triple.  A local filler gadget then spends exactly two of the three incidences of every source port, leaving exactly one incidence for the source matching.

The reduction is self-contained apart from the standard unary \NP-completeness of \NthreeDM and the elementary fact that a regular bipartite multigraph decomposes into perfect matchings; see, for example,~\cite{LovaszPlummer1986, Alon2003}.  A restricted form of \NthreeDM is also known to be strongly \NP-hard~\cite{YuHoogeveenLenstra2004}.

\section{Problem definition and main theorem}

Because repeated numerical values are allowed, it is important to distinguish an \emph{item occurrence} from its weight.

\begin{definition}[Numerical Three-Dimensional Matching]
An instance of \NthreeDM is a tuple
\[
  I=(X,Y,Z,w,T),
\]
where $X,Y,Z$ are pairwise disjoint finite sets with $|X|=|Y|=|Z|=m$, the weight function is $w:X\uplus Y\uplus Z\to\NN_{>0}$, and $T\in\NN_{>0}$ is a target.  A feasible matching is a partition $\mathcal M=\{M_1,\ldots,M_m\}$ of $X\uplus Y\uplus Z$ such that, for every $r=1,\ldots,m$,
\[
 |M_r\cap X|=|M_r\cap Y|=|M_r\cap Z|=1
 \quad\text{and}\quad
 \sum_{u\in M_r}w(u)=T.
\]
\end{definition}
 Since the $m$ triples of a feasible matching partition all of $X\uplus Y\uplus Z$ and each sums to $T$, the target is necessarily
\[
 T=\frac{1}{m}\sum_{u\in X\uplus Y\uplus Z} w(u),
\]
the total weight of all elements divided by $m$.

Equivalently, after indexing $X=\{x_1,\ldots,x_m\}$, a solution consists of two permutations $\pi,\sigma$ of $(1, \ldots, m)$, with
\[
  w(x_i)+w(y_{\pi(i)})+w(z_{\sigma(i)})=T
  \qquad (i=1,\ldots,m).
\]
Thus the requirement ``one item from each class'' is part of the problem definition, not something inferred from the numerical values.

\begin{definition}[Identical weight classes]
Index the three classes as $X=\{x_1,\ldots,x_m\}$, $Y=\{y_1,\ldots,y_m\}$, and $Z=\{z_1,\ldots,z_m\}$.  An \NthreeDM{} instance has \emph{identical weight classes} if there exist permutations $\pi,\sigma$ of $(1, \ldots, m)$ such that
\[
 w(x_i)=w(y_{\pi(i)})=w(z_{\sigma(i)})
 \qquad\text{for every }i=1,\ldots,m;
\]
equivalently, the three classes carry the same multiset of weights.  The restriction to such instances is denoted by \SymNthreeDM.  We occasionally write $X=Y=Z$ as shorthand for equality of the three \emph{weight multisets}; the three classes themselves remain disjoint and labeled.
\end{definition}

\begin{definition}[Unary encoding]
In the unary version of either problem, every weight and the target are represented in unary.  Item labels and delimiters use the usual polynomial-size binary or symbolic representation.  A problem is \emph{unary NP-hard} if its unary-encoded restriction is \NP-hard.
\end{definition}

\begin{theorem}[Main theorem]\label{n3dm:thm:main}
Problem~SN3DM is \NP-complete when all numerical data are encoded in unary.  Consequently, Numerical Three-Dimensional Matching remains strongly \NP-complete when its three classes have identical weight multisets.
\end{theorem}

Membership in \NP is immediate: a collection of $m$ triples is a polynomially checkable certificate.  The remainder of Part~I proves unary \NP-hardness.

\section{Proof overview}

We reduce from unary \NthreeDM, which is \NP-complete by the strong \NP-completeness of the classical problem~\cite{GareyJohnson1979}.  Let the source classes be represented by the occurrence multisets
\[
 A=\mset{a_1,\ldots,a_n},\qquad B=\mset{b_1,\ldots,b_n},\qquad C=\mset{c_1,\ldots,c_n},
\]
with common target $t$. Here $\mset{\cdot}$ denotes a multiset: repeated values are permitted and are counted with multiplicity, so equal weights at different indices remain distinct occurrences.

An arbitrary source \NthreeDM instance need not have identical weight classes, and no such assumption is made.  The reduction does not identify the source classes $A,B,C$ with the three output classes.  Instead, the three source roles are represented internally by the port types $P_1,P_2,P_3$ in one common occurrence set $S$, and the output consists of three labeled copies of that entire set.  Consequently, the output classes have identical weight multisets, while the role and source coordinates retain the asymmetry of the original data.  The local filler gadget consumes two of the three incidences of every source port, leaving one incidence for a main $P_1P_2P_3$ triple.  The main triples therefore reproduce the original $A$--$B$--$C$ matching constraints.  The copy-forgetting and recoloring lemma then translates between this degree-three incidence representation and an ordinary matching among the three identical output classes.

For each source index $i$, the reduction creates a group containing three \emph{ports} $P_{1,i},P_{2,i},P_{3,i}$ and ten \emph{private items} $Q_{1,i},\ldots,Q_{10,i}$.  These two kinds of occurrence play opposite roles.  A port is an interface occurrence: it carries a source value ($P_{1,i},P_{2,i},P_{3,i}$ carry $a_i,b_i,c_i$, respectively) and is the only occurrence type permitted to enter a \emph{main} triple $(P_{1,i},P_{2,j},P_{3,k})$, which is what reproduces a source equation $a_i+b_j+c_k=t$.  A private item never appears in a main triple; it exists only inside the group's local \emph{filler} triples, and its sole job is to absorb incidences so as to regulate how many of each port's three incidences are used up locally --- leaving exactly one free for a main triple.  All $13n$ occurrences form one common set $S$, and the output classes are three labeled copies of $S$.  Four vector coordinates serve different purposes:
\begin{enumerate}[label=(\roman*)]
  \item a fixed role coordinate permits only one main role pattern and nine filler role patterns;
  \item a source coordinate makes a main triple encode $a_i+b_j+c_k=t$ and makes each intended filler triple attain the target;
  \item two group coordinates, proportional to $i$ and $i^2$, force the private items of every filler triple to come from one common group.
\end{enumerate}
The local filler equations force two uses of each port in filler triples and three uses of every private item.  Because every underlying occurrence must have total degree three, exactly one incidence of each port remains for main triples.  Those main triples recover a source matching.  A mixed-radix encoding finally replaces the four-dimensional weights by ordinary positive integers.

The following lemma explains why a degree-three triple system is the correct intermediate object.

\begin{lemma}[Copy-forgetting and recoloring]\label{n3dm:lem:copy-forgetting}
Let $S$ be a finite occurrence set with weight function $q:S\to\NN_{>0}$, and let $K\in\NN_{>0}$.  Form a symmetric \NthreeDM{} instance from three labeled copies $S^X,S^Y,S^Z$, assigning all three copies of $s$ the weight $q(s)$ and using target $K$.

This instance is feasible if and only if there is a finite multiset $\mathcal H$ of three-incidence hyperedges on $S$ such that
\begin{enumerate}[label=(\alph*)]
  \item every hyperedge has total $q$-weight $K$; and
  \item every occurrence $s\in S$ has incidence degree $\degH(s)=3$, where $\degH(s)$ is the number of times $s$ shows up across all the hyperedges of $\mathcal H$, counting every appearance separately, including when $s$ shows up more than once in the same hyperedge.
\end{enumerate}
Parallel hyperedges, and repeated incidences within a hyperedge, are permitted in this intermediate representation.
\end{lemma}

\begin{proof}
Suppose first that the symmetric instance has a feasible matching.  In every matching triple, forget the superscripts $X,Y,Z$ and retain only the underlying occurrences in $S$.  Each triple becomes a three-incidence hyperedge of weight $K$.  Since the $X$-, $Y$-, and $Z$-copy of every $s\in S$ is used exactly once, the resulting incidence degree of $s$ is three.

Conversely, suppose that $\mathcal H$ satisfies (a) and (b).  Construct the incidence bipartite multigraph $G$ whose left vertices are the occurrences in $S$ and whose right vertices are the hyperedges of $\mathcal H$, placing \emph{one edge for each incidence-slot}: a hyperedge that fills $k$ of its three slots with the same occurrence $s$ contributes $k$ parallel edges between $s$ and that hyperedge.  With this convention each right vertex (a hyperedge) has degree three, since it has three slots; and each left vertex $s$ has degree $\degH(s)=3$ by condition (b).  Hence the two sides have equal cardinality and $G$ is a $3$-regular bipartite multigraph (a multigraph is \emph{$k$-regular} if every vertex has degree exactly $k$, that is, exactly $k$ edge-ends meet at each vertex, parallel edges counted with multiplicity).

A $k$-regular bipartite multigraph has a perfect matching for the following reason. Let $N(U)$ be the \emph{neighborhood} of a vertex set $U$ (i.e., the set of all vertices adjacent to at least one vertex of $U$). Given any  vertex set $U$ on one side, the $k|U|$ incident edges of $U$ all enter $N(U)$, which, however, are only a subset of   all the $k|N(U)|$ edges incident to the vertices in $N(U)$. This means that  $k|N(U)|\ge k|U|$ so Hall's \cite{Hall1935} inequality $|N(U)|\ge |U|$ holds.  Removing a perfect matching preserves regularity with degree reduced by one.  Applying this argument three times decomposes $E(G)$ (the edge multiset of $G$) into three perfect matchings.  Color them $X,Y,Z$.

At every item vertex, its three incidences now have distinct colors, and the same is true at every hyperedge vertex.  Replace an incidence of color $R\in\{X,Y,Z\}$ at item $s$ by the copy $s^R$.  Each resulting triple contains exactly one item from each labeled class, every copy is used exactly once, and every triple still has weight $K$.  This is a feasible symmetric \NthreeDM{} matching.
\end{proof}

\paragraph{An example.}
Let $S=\{s_1,\ldots,s_5\}$ with weights $q(s_i)=i$ and target $K=9$; the output classes $S^X,S^Y,S^Z$ are the labeled copies.  Consider the multiset of five hyperedges
\[
\begin{aligned}
 e_1&=\mset{s_1,s_3,s_5}, &
 e_2&=\mset{s_1,s_3,s_5},\\
 e_3&=\mset{s_1,s_4,s_4}, &
 e_4&=\mset{s_2,s_2,s_5},\\
 e_5&=\mset{s_2,s_3,s_4}.&&
\end{aligned}
\]
Here $e_1$ and $e_2$ are parallel, $e_3$ reuses $s_4$, and $e_4$ reuses $s_2$.  Each hyperedge has $q$-weight $9$, and each occurrence has total incidence degree three (for instance $s_3$ appears once each in $e_1,e_2,e_5$, and $s_4$ appears twice in $e_3$ and once in $e_5$).  Building $G$ with one edge per incidence-slot gives a $3$-regular bipartite multigraph, and decomposing it into three perfect matchings $X,Y,Z$ yields, as one valid coloring,
\[
\begin{aligned}
 e_1&\mapsto(s_{3,X},s_{5,Y},s_{1,Z}), &
 e_2&\mapsto(s_{5,X},s_{1,Y},s_{3,Z}), &
 e_3&\mapsto(s_{1,X},s_{4,Y},s_{4,Z}),\\
 e_4&\mapsto(s_{2,X},s_{2,Y},s_{5,Z}), &
 e_5&\mapsto(s_{4,X},s_{3,Y},s_{2,Z}). & &
\end{aligned}
\]
Each class then lists every occurrence exactly once --- $S^X$ receives $s_3,s_5,s_1,s_2,s_4$; $S^Y$ receives $s_5,s_1,s_4,s_2,s_3$; $S^Z$ receives $s_1,s_3,s_4,s_5,s_2$ --- and each triple sums to $9$, so this is a feasible symmetric matching produced directly by the proof's procedure.  The decomposition is not unique, so this is one coloring the algorithm may produce, not the only one.  The reused occurrences illustrate the following: the two $s_4$-edges of $e_3$ are parallel edges at one hyperedge vertex, so a proper coloring is forced to give them different colors ($Y$ and $Z$), turning $\mset{s_1,s_4,s_4}$ into the legitimate triple $(s_{1,X},s_{4,Y},s_{4,Z})$; likewise the two $s_2$-edges of $e_4$ receive colors $X$ and $Y$.

\paragraph{Why edge coloring is needed.}
Forgetting the superscripts $X,Y,Z$ simplifies the gadget analysis, but it also suppresses the one-item-from-each-class condition.  Edge coloring is the step that restores exactly that information.  A color on an incidence specifies which physical copy of the underlying item is used in that position.  Because the three color classes are perfect matchings, every item vertex receives one incidence of each color, so each of $s^X,s^Y,s^Z$ is used exactly once; every hyperedge vertex also receives one incidence of each color, so each reconstructed triple contains one item from each labeled class.  Edge coloring does not enforce the source-sum equations or the filler equations.  Its essential role is to prove that every degree-three unlabeled triple system admitted by the gadget can be lifted to a genuine feasible matching in the three symmetric labeled classes.

\begin{figure}[H]
\centering
\begin{tikzpicture}[
  box/.style={draw,rounded corners,align=center,minimum height=1.05cm,text width=2.15cm,fill=blue!3},
  note/.style={draw,rounded corners,align=center,minimum height=.9cm,text width=2.55cm,fill=gray!5},
  arr/.style={-{Latex[length=2.3mm]},thick},
  node distance=6mm,
  font=\small
]
\node[box] (src) {Unary source\\\NthreeDM\\$A,B,C;t$};
\node[box,right=of src] (groups) {For each $i$:\\3 ports $+$\\10 private items};
\node[box,right=of groups] (common) {Common set $S$\\$|S|=13n$};
\node[box,right=of common] (copies) {Labeled copies\\$S^X,S^Y,S^Z$};
\node[box,right=of copies] (out) {Symmetric \NthreeDM\\target $K^\star$};
\node[note,below=8mm of groups] (coords) {role $\mid$ source\\group $\mid$ group-square};
\node[note,below=8mm of out] (factor) {3-regular target-sum\\triple system};
\draw[arr] (src) -- (groups);
\draw[arr] (groups) -- (common);
\draw[arr] (common) -- (copies);
\draw[arr] (copies) -- (out);
\draw[arr] (groups) -- node[right,font=\scriptsize]{four coordinates} (coords);
\draw[<->,thick] (out) -- node[left,align=right,font=\scriptsize]{forget/recolor\\class incidences} (factor);
\end{tikzpicture}
\caption{Architecture of the reduction.  The copy-forgetting lemma translates the one-item-per-class matching condition into degree three at every underlying occurrence.  The four coordinates then force the desired local and main hyperedges.}
\label{n3dm:fig:overview}
\end{figure}
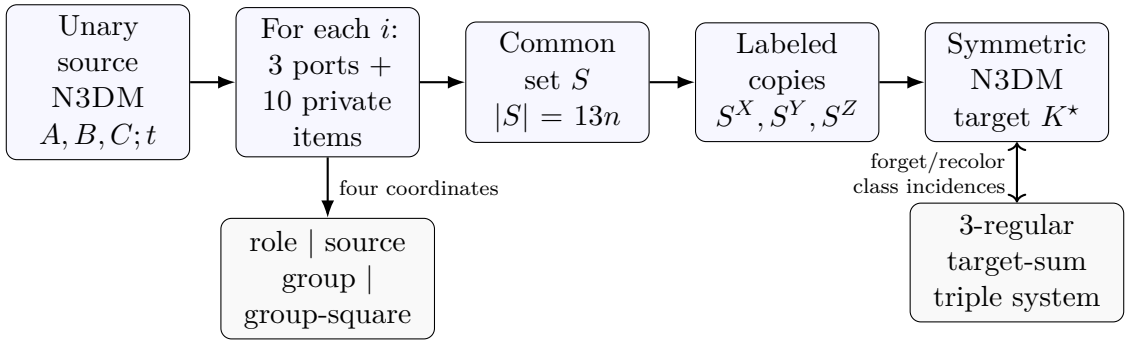

\paragraph{How to read \cref{n3dm:fig:overview}.}
The top row separates the reduction into conceptually different operations.  The first arrow does not copy the source classes directly.  It translates the asymmetric source data into the internal port roles $P_1,P_2,P_3$ and adds private items that will regulate how often those ports may be used.  The next arrow simply gathers all of those roles into one common occurrence set $S$, and the following arrow makes three labeled copies of the entire set.  That last copying operation is what guarantees identical output weight multisets: each output class contains the same catalogue of encoded roles and values.

The lower-left box records the four logical coordinates attached to each occurrence before they are packed into one integer.  The lower-right box is not a second constructed instance; it is an unlabeled way to view a solution of the symmetric instance.  Moving downward along the two-headed arrow forgets whether an incidence came from $S^X$, $S^Y$, or $S^Z$, leaving three incidences at every underlying item.  Moving upward uses edge coloring to assign those incidences back to the three labeled copies consistently.  Thus the figure distinguishes two ideas that are easy to conflate: the top row constructs a symmetric numerical instance, whereas the vertical equivalence supplies a convenient representation for proving that the construction behaves correctly.

\section{The constructed instance}\label{n3dm:sec:construction}

\subsection{Source instance and common occurrence set}

Let
\[
 I=(A,B,C,t)
\]
be a unary \NthreeDM{} instance, where $A=\mset{a_1,\ldots,a_n}$, $B=\mset{b_1,\ldots,b_n}$, and $C=\mset{c_1,\ldots,c_n}$ consist of positive integer occurrences.  Repeated values are allowed.

Given the \NthreeDM{} instance as defined, we now start to construct an instance of \SymNthreeDM. For each $i\in[n]=\{1,\ldots,n\}$, define
\[
 S_i=\{P_{1,i},P_{2,i},P_{3,i},Q_{1,i},\ldots,Q_{10,i}\},
 \qquad
 S=\biguplus_{i=1}^n S_i.
\]
The occurrences $P_{1,i},P_{2,i},P_{3,i}$ are the ports representing $a_i,b_i,c_i$, respectively.  The $Q$-occurrences are private filler items.  The output classes are
\[
 X^\star=S^X,\qquad Y^\star=S^Y,\qquad Z^\star=S^Z,
\]
three disjoint labeled copies of $S$.  They therefore have identical weight multisets by construction.

\subsection{Vector weights}

Let the set of roles be
\[
 \Rset=\{P_1,P_2,P_3,Q_1,\ldots,Q_{10}\}.
\]
For each role $R\in\Rset$, fix a role code $\rho_R$ and a group coefficient $h_R$.  The source-coordinate expression $\eta_{R,i}$ depends on the source values at index $i$.  All data are given in \cref{n3dm:tab:weights}.

\begin{table}[H]
\centering
\caption{Role data for the vector construction.}
\label{n3dm:tab:weights}
\small
\renewcommand{\arraystretch}{1.12}
\begin{tabular}{@{}c r r l@{}}
\toprule
Role $R$ & $\rho_R$ & $h_R$ & $\eta_{R,i}$ \\
\midrule
$P_1$    & 31  & 0  & $3a_i$ \\
$P_2$    & 8   & 0  & $3b_i$ \\
$P_3$    & 261 & 0  & $3c_i$ \\
\midrule
$Q_1$    & 184 & 2  & $-3b_i$ \\
$Q_2$    & 85  & $-2$ & $3t-3a_i+3b_i$ \\
$Q_3$    & 177 & $-2$ & $3t-2a_i+b_i+c_i$ \\
$Q_4$    & 115 & 2  & $2a_i-4b_i-c_i$ \\
$Q_5$    & 23  & 2  & $a_i-2b_i-2c_i$ \\
$Q_6$    & 16  & $-2$ & $3t-a_i+2b_i-c_i$ \\
$Q_7$    & 27  & 1  & $2a_i-b_i-c_i$ \\
$Q_8$    & 89  & $-3$ & $3t-2a_i+4b_i+c_i$ \\
$Q_9$    & 188 & 1  & $a_i-2b_i+c_i$ \\
$Q_{10}$ & 96  & 1  & $0$ \\
\bottomrule
\end{tabular}
\end{table}

\paragraph{How to read Table~\ref{n3dm:tab:weights}.}
Each row is a recipe for every occurrence of one role. The role code
$\rho_R$ is a fixed numerical ``barcode'' for the type of item; it is the
same in every source group and is used only to decide which three role types
may coexist. The source expression $\eta_{R,i}$ is the part that carries the
input values $a_i,b_i,c_i,t$. The coefficient $h_R$ is multiplied by both
$i$ and $i^2$ in the last two coordinates, so it acts as a locality marker
for the private items. The three ports have $h_R=0$ because a main edge must
be free to combine an $A$-port from group $i$, a $B$-port from group $j$,
and a $C$-port from group $k$. The private roles have nonzero coefficients
because unintended mixing among their groups must be detected.

For example, the rows for $P_1,Q_1,Q_2$ are designed to work together.
Their role codes add to $31+184+85=300$; their source expressions cancel
to $3t$; and their group coefficients add to $0+2-2=0$. If all three
occurrences use group $i$, the group-square coefficients cancel as well.
The apparently irregular constants in the table should therefore be viewed
not as meaningful source weights but as a finite system of tags chosen so
that the intended role triples pass all four tests and unintended triples
fail at least one test.

For an occurrence $R_i\in S_i$, define its four-dimensional integer
weight by
\begin{equation}
  \omega(R_i)
  = \bigl(\rho_R,\eta_{R,i},h_R i,h_R i^2\bigr),
  \label{eq:vector-weight}
\end{equation}
and define the vector target
\begin{equation}
  \mathbf{K}=(300,3t,0,0).
  \label{eq:vector-target}
\end{equation}
The four coordinates are called the \emph{role}, \emph{source},
\emph{group}, and \emph{group-square} coordinates.  These vectors are an
intermediate bookkeeping device: the actual integer weight of each
occurrence $R_i$ is a weighted sum of the components of $\omega(R_i)$, with
the weights being powers of a base $L$ defined in
\cref{n3dm:subsec:encoding}; similarly, the actual target is the same
weighted sum of the components of $\mathbf K$.

\paragraph{Role code.}
The role code $\rho_R$ is a fixed, input-independent integer label for
the item type $R$. It does not encode any source value. Instead, the
first-coordinate target $300$ acts as a type filter: the codes are chosen
so that an unordered triple of roles has total role code $300$ exactly
for the ten patterns $F_0,\ldots,F_9$ listed in the next subsection.
Thus the role coordinate determines which combinations of ports and
private items are allowed to form a target-sum triple.

\paragraph{Source-coordinate expression.}
The expression $\eta_{R,i}$ is the second-coordinate value assigned to
role $R$ in source group $i$. Unlike the role code, it depends on the
input through $a_i,b_i,c_i$, and $t$. For the three ports it is $3a_i$,
$3b_i$, and $3c_i$, so a main triple
$(P_{1,i},P_{2,j},P_{3,k})$ has source-coordinate sum
$3(a_i+b_j+c_k)$ and reaches the target $3t$ exactly when
$a_i+b_j+c_k=t$. The expressions for the private items are chosen so
that their terms cancel in every intended local filler pattern, making
its source-coordinate sum identically $3t$. The source coordinate
therefore carries the numerical content of the source instance after
the role coordinate has selected an admissible triple type.

\subsection{Permitted role patterns}

The intended role patterns are
\begin{align*}
 F_0&=(P_1,P_2,P_3),\\
 F_1&=(P_1,Q_1,Q_2), & F_2&=(P_2,Q_3,Q_4), & F_3&=(P_3,Q_5,Q_6),\\
 F_4&=(Q_1,Q_7,Q_8), & F_5&=(Q_2,Q_7,Q_9), & F_6&=(Q_3,Q_7,Q_{10}),\\
 F_7&=(Q_4,Q_8,Q_{10}), & F_8&=(Q_5,Q_8,Q_9), & F_9&=(Q_6,Q_9,Q_{10}).
\end{align*}


We may write $F_r(\cdot,\cdot,\cdot)$ for the occurrence triple obtained by
assigning the listed group indices to the three roles of $F_r$, in the order
in which the roles are displayed above.  In general the three indices may
differ; in particular, a \emph{main} occurrence triple is written as
\[
  F_0(i,j,k)=(P_{1,i},P_{2,j},P_{3,k}).
\]
When all three roles come from one common group $i$, we abbreviate
$F_r(i,i,i)$ to $F_r(i)$.  A \emph{local} filler triple takes all three
roles from one group and is therefore written $F_r(i)$ for $r\ge 1$.

\begin{lemma}[Role filter]\label{n3dm:lem:role-filter}
An unordered multiset of three roles has role-code sum $300$ if and only if it is one of $F_0,F_1,\ldots,F_9$.
\end{lemma}

\begin{proof}
Direct addition shows that every listed pattern has sum $300$.  Conversely, there are only $\binom{13+3-1}{3}=455$ unordered triples with repetition.  The complementary-pair certificate in \cref{app:role-codes} lists, for each possible first role, every role pair that completes it to $300$; it yields exactly the ten patterns above and no others.
\end{proof}

\subsection{Target identities}

A \emph{target identity} is an algebraic equality showing that the coordinates of an intended filler pattern add to the target for every permissible choice of the source data.  It is called an identity because it is built into the formulas in \cref{n3dm:tab:weights}; it is not an additional condition that the source instance must happen to satisfy.  For a local filler triple, the role and source coordinates must equal $300$ and $3t$, while the group and group-square coordinates must equal zero.  By contrast, a main $P_1P_2P_3$ triple reaches the source target only when the selected source values satisfy the original equation $a_i+b_j+c_k=t$.

Every local filler triple has vector sum $\Kvec$.  For example,
\begin{align*}
 \eta_{P_1,i}+\eta_{Q_1,i}+\eta_{Q_2,i}
 &=3a_i-3b_i+(3t-3a_i+3b_i)=3t,\\
 \eta_{Q_1,i}+\eta_{Q_7,i}+\eta_{Q_8,i}
 &=-3b_i+(2a_i-b_i-c_i)\\
 &\qquad +(3t-2a_i+4b_i+c_i)=3t.
\end{align*}
The remaining source-coordinate identities, together with the group coefficients, are summarized in \cref{n3dm:tab:patterns}.  Since all occurrences in $F_r(i, i, i)$ have the same index $i$, a zero sum of the $h_R$ coefficients makes both group coordinates vanish.

\begin{table}[H]
\centering
\caption{Coordinate identities for the filler patterns.  In every row the role-code sum is $300$, the source-coordinate sum is $3t$, and the listed group coefficients sum to zero.}
\label{n3dm:tab:patterns}
\small
\renewcommand{\arraystretch}{1.1}
\begin{tabular}{@{}c l c@{}}
\toprule
Pattern & Roles & Group coefficients \\
\midrule
$F_1$ & $P_1,Q_1,Q_2$ & $0,2,-2$ \\
$F_2$ & $P_2,Q_3,Q_4$ & $0,-2,2$ \\
$F_3$ & $P_3,Q_5,Q_6$ & $0,2,-2$ \\
$F_4$ & $Q_1,Q_7,Q_8$ & $2,1,-3$ \\
$F_5$ & $Q_2,Q_7,Q_9$ & $-2,1,1$ \\
$F_6$ & $Q_3,Q_7,Q_{10}$ & $-2,1,1$ \\
$F_7$ & $Q_4,Q_8,Q_{10}$ & $2,-3,1$ \\
$F_8$ & $Q_5,Q_8,Q_9$ & $2,-3,1$ \\
$F_9$ & $Q_6,Q_9,Q_{10}$ & $-2,1,1$ \\
\bottomrule
\end{tabular}
\end{table}

\paragraph{How to read \cref{n3dm:tab:patterns}.}
Each row is a blueprint for one allowed filler edge.  The first three rows attach one port to two private items; these are the edges that will consume two of the three available incidences of each port.  The remaining six rows connect only private items and close the local degree-balancing network.  The final column is a quick locality check.  Because the listed coefficients add to zero, taking all private roles from the same group $i$ makes both the group-coordinate sum $i\sum h_R$ and the group-square sum $i^2\sum h_R$ equal zero.

The table is best read as a verification that every intended filler pattern succeeds, not as the proof that no other pattern succeeds.  The role-filter lemma supplies that converse by showing that only these nine filler patterns and the one main pattern have role-code sum $300$.  The source-expression identities then verify the arithmetic target, and the two group coordinates prevent private roles from unrelated groups from being spliced together.  In this way the four coordinates divide the work: role codes select the pattern, source expressions check its numerical meaning, and group data enforce locality.

For a main triple, the port coefficients $h_{P_1},h_{P_2},h_{P_3}$ are all zero, and therefore
\begin{equation}\label{n3dm:eq:main-equivalence}
 \wvec(P_{1,i})+\wvec(P_{2,j})+\wvec(P_{3,k})=\Kvec
 \quad\Longleftrightarrow\quad
 a_i+b_j+c_k=t.
\end{equation}

\subsection{Encoding the vector weights by positive integers}\label{n3dm:subsec:encoding}

Recall that $S=\biguplus_{i=1}^n S_i$ is the common occurrence set, with $|S|=13n$; these  occurrences have four-dimensional weights $\omega(R_i)$, which we now encode as ordinary positive integers.  We index the four coordinates by $r=0,1,2,3$, corresponding to the role, source, group, and group-square coordinates in that order.

Some source and group coordinates in Equation~(\ref{eq:vector-weight}) may be negative, so we first shift each coordinate to be nonnegative.  For each coordinate $r$, choose a shift
\[
 D_r\ge \max\bigl\{0,-\min_{u\in S}\wvec_r(u)\bigr\},
\]
and define the \emph{shifted weight} and shifted target in coordinate $r$ by
\[
 \overline\wvec_r(u)=\wvec_r(u)+D_r,
 \qquad
 \overline K_r=K_r+3D_r
 \qquad (r=0, 1,2,3).
\]
Uniform shifting preserves triple equalities because every candidate contains exactly three items.  The role coordinate needs no shift, so $D_0=0$; the other shifts depend on the instance.  For example, in a small instance with target $t=9$ and groups $(a_i,b_i,c_i)=(2,3,4)$ and $(1,5,3)$, the coordinate minima over all occurrences are $(8,-21,-6,-12)$, giving $D=(0,21,6,12)$.  In general one may take $D_r=\max\{0,-\min_{u}\wvec_r(u)\}$.

Choose an integer base $L$ satisfying
\begin{equation}\label{n3dm:eq:base-choice}
 L>
 \max_{0\le r\le 3}
 \left\{\overline K_r,
 3\max_{u\in S}\overline\wvec_r(u)\right\}.
\end{equation}
Any integer above this bound works; in the small instance above the bound evaluates to $783$, so $L=1000$ is a convenient choice.  Define the ordinary item weight and target by
\begin{equation}\label{n3dm:eq:mixed-radix}
 W(u)=\sum_{r=0}^3 \overline\wvec_r(u)\,L^{\,r},
 \qquad
 K^\star=\sum_{r=0}^3 \overline K_r\,L^{\,r}.
\end{equation}
Every $W(u)$ is positive because its role digit is positive.

\paragraph{The item weights explicitly.}
Writing the shifted weight of each coordinate out, the integer weight of an occurrence of role $R$ in group $i$ is
\[
 W(R_i)=\overline\wvec_0(R_i)+\overline\wvec_1(R_i)\,L+\overline\wvec_2(R_i)\,L^2+\overline\wvec_3(R_i)\,L^3,
\]
a weighted sum of powers of $L$.  Substituting the role data of \cref{n3dm:tab:weights} and the shifts $D=(0,D_1,D_2,D_3)$ gives, role by role,
\begin{align*}
 W(P_{1,i})&=31+(3a_i+D_1)L+D_2L^2+D_3L^3,\\
 W(P_{2,i})&=8+(3b_i+D_1)L+D_2L^2+D_3L^3,\\
 W(P_{3,i})&=261+(3c_i+D_1)L+D_2L^2+D_3L^3,\\
 W(Q_{1,i})&=184+(-3b_i+D_1)L+(2i+D_2)L^2+(2i^2+D_3)L^3,\\
 W(Q_{2,i})&=85+(3t-3a_i+3b_i+D_1)L+(-2i+D_2)L^2+(-2i^2+D_3)L^3,\\
 W(Q_{3,i})&=177+(3t-2a_i+b_i+c_i+D_1)L+(-2i+D_2)L^2+(-2i^2+D_3)L^3,\\
 W(Q_{4,i})&=115+(2a_i-4b_i-c_i+D_1)L+(2i+D_2)L^2+(2i^2+D_3)L^3,\\
 W(Q_{5,i})&=23+(a_i-2b_i-2c_i+D_1)L+(2i+D_2)L^2+(2i^2+D_3)L^3,\\
 W(Q_{6,i})&=16+(3t-a_i+2b_i-c_i+D_1)L+(-2i+D_2)L^2+(-2i^2+D_3)L^3,\\
 W(Q_{7,i})&=27+(2a_i-b_i-c_i+D_1)L+(i+D_2)L^2+(i^2+D_3)L^3,\\
 W(Q_{8,i})&=89+(3t-2a_i+4b_i+c_i+D_1)L+(-3i+D_2)L^2+(-3i^2+D_3)L^3,\\
 W(Q_{9,i})&=188+(a_i-2b_i+c_i+D_1)L+(i+D_2)L^2+(i^2+D_3)L^3,\\
 W(Q_{10,i})&=96+D_1L+(i+D_2)L^2+(i^2+D_3)L^3,
\end{align*}
and the target is $K^\star=300+(3t+3D_1)L+3D_2L^2+3D_3L^3$.  These are the concrete positive integers on which \SymNthreeDM{} is finally posed.

\paragraph{What a carry means.}
In base-$L$ arithmetic, a \emph{carry} occurs when the sum in one digit position reaches at least $L$: one full block of size $L$ is removed from that position and transferred as one unit to the next higher position.  A carry would be dangerous here because an excess in one logical coordinate could then alter the digit representing the next coordinate and potentially conceal a mismatch.  The inequality in \cref{n3dm:eq:base-choice} makes every sum of three item digits strictly smaller than $L$.  Hence each coordinate can be compared independently, exactly as if the four coordinates were still written separately.

\begin{lemma}[No-carry encoding]\label{n3dm:lem:no-carry}
For any three occurrences $u,v,z\in S$,
\[
 W(u)+W(v)+W(z)=K^\star
 \quad\Longleftrightarrow\quad
 \wvec(u)+\wvec(v)+\wvec(z)=\Kvec.
\]
\end{lemma}

\begin{proof}
By \cref{n3dm:eq:base-choice}, every digit of the sum of three encoded item weights lies in $\{0,\ldots,L-1\}$, as does every target digit.  Hence no carry occurs in any coordinate, and equality of the encoded integers is equivalent to coordinatewise equality of the shifted vectors.  Subtracting $3D_r$ in coordinates $r=1,2,3$ gives the original vector equality.
\end{proof}

The output instance is now the ordinary \NthreeDM{} instance
\[
 I^\star=(X^\star,Y^\star,Z^\star,W,K^\star),
\]
where all three classes carry the same multiset $\mset{W(u):u\in S}$.

\begin{figure}[H]
\centering
\begin{tikzpicture}[
  slot/.style={draw,rounded corners=1pt,minimum width=10mm,minimum height=7.5mm,fill=gray!7},
  item/.style={anchor=east,font=\small},
  font=\small
]
\node[font=\small\bfseries] at (1.05,2.0) {three incidence slots};
\node[item] at (-0.7,1.25) {$P_{1,i}$};
\node[item] at (-0.7,0.45) {$P_{2,i}$};
\node[item] at (-0.7,-0.35) {$P_{3,i}$};
\node[item] at (-0.7,-1.15) {$Q_{1,i},\ldots,Q_{10,i}$};
\foreach \y in {1.25,0.45,-0.35}{
  \node[slot] at (0,\y) {L};
  \node[slot] at (1.05,\y) {L};
  \node[slot,fill=blue!4] at (2.10,\y) {M};
}
\node[slot] at (0,-1.15) {L};
\node[slot] at (1.05,-1.15) {L};
\node[slot] at (2.10,-1.15) {L};
\node[anchor=west,align=left,font=\small] at (3.0,0.65)
  {$L$: incidence in a local filler edge\\[2pt]
   $M$: incidence in a main edge};
\end{tikzpicture}
\caption{Degree budget of one source group.  The private-degree equations force the pattern $L,L,M$ at every port and $L,L,L$ at every private item.  Hence each port has exactly one incidence available for a main triple.}
\label{n3dm:fig:degree-budget}
\end{figure}

\paragraph{How to read \cref{n3dm:fig:degree-budget}.}
Each row represents one underlying occurrence after the outer labels $X,Y,Z$ have been forgotten, and the three boxes represent its three available incidence slots.  The letters $L$ and $M$ describe functions, not predetermined class labels: $L$ marks a use in a local filler edge, while $M$ marks a use in a main edge that encodes one source triple.  Edge coloring later decides which of the three physical copies supplies each slot.

For a port, the desired pattern is $L,L,M$.  Two local uses make the surrounding private gadget close correctly, while the single remaining main use allows that source occurrence to participate once in the recovered matching.  A private item has pattern $L,L,L$, so all of its incidence capacity is absorbed inside the local gadget and it can never leak into a main triple.  The picture is therefore a visual summary of what the degree equations will prove; it is not an extra assumption imposed on a solution.

\section{Correctness of the reduction}

\subsection{The ``if'' case}\label{n3dm:sec:if}

Assume that the source instance has a feasible matching.  Thus there are index triples
\[
 \mathcal M=\{(i_r,j_r,k_r):r=1,\ldots,n\}
\]
such that each index in every source class occurs exactly once and
\begin{equation}\label{n3dm:eq:source-solution}
 a_{i_r}+b_{j_r}+c_{k_r}=t
 \qquad (r=1,\ldots,n).
\end{equation}
We construct a 3-regular target-sum triple system $\mathcal H$ on $S$.

For every group $i$, include two parallel copies of each of $F_1(i)$, $F_2(i)$, and $F_3(i)$, and one copy of each of $F_4(i),\ldots,F_9(i)$.  For every source matching triple $(i_r,j_r,k_r)$, include the main hyperedge
\[
 F_0(i_r,j_r,k_r)=\bigl(P_{1,i_r},P_{2,j_r},P_{3,k_r}\bigr).
\]
Every filler edge has vector sum $\Kvec$ by \cref{n3dm:tab:patterns}, and every main edge has vector sum $\Kvec$ by \cref{n3dm:eq:source-solution,n3dm:eq:main-equivalence}.  By \cref{n3dm:lem:no-carry}, every edge has ordinary weight $K^\star$.

It remains to check degrees.  Each port occurs twice in its corresponding local filler edge and once in the unique main edge selected by the source solution.  Hence every port has degree three.  For the private items, \cref{n3dm:tab:forward-private-degrees} lists, for each $Q_{\ell,i}$, the filler triples that contain it and the number of copies of each; the incidences always total three.

\begin{table}[H]
\centering
\caption{Incidence ledger of the forward construction: each private item $Q_{\ell,i}$, the local filler triples containing it (with the number of copies), and its total incidence degree.}
\label{n3dm:tab:forward-private-degrees}
\small
\renewcommand{\arraystretch}{1.15}
\begin{tabular}{@{}c l c@{}}
\toprule
Private item & Filler triples containing it (copies) & Total incidences \\
\midrule
$Q_{1,i}$ & $F_1$ (two copies), $F_4$ & $3$ \\
$Q_{2,i}$ & $F_1$ (two copies), $F_5$ & $3$ \\
$Q_{3,i}$ & $F_2$ (two copies), $F_6$ & $3$ \\
$Q_{4,i}$ & $F_2$ (two copies), $F_7$ & $3$ \\
$Q_{5,i}$ & $F_3$ (two copies), $F_8$ & $3$ \\
$Q_{6,i}$ & $F_3$ (two copies), $F_9$ & $3$ \\
$Q_{7,i}$ & $F_4$, $F_5$, $F_6$ & $3$ \\
$Q_{8,i}$ & $F_4$, $F_7$, $F_8$ & $3$ \\
$Q_{9,i}$ & $F_5$, $F_8$, $F_9$ & $3$ \\
$Q_{10,i}$ & $F_6$, $F_7$, $F_9$ & $3$ \\
\bottomrule
\end{tabular}
\end{table}

 \paragraph{How to read \cref{n3dm:tab:forward-private-degrees}.}
The table is an incidence ledger.  For instance, the row for $Q_{1,i}$ records that it appears in each of the two copies of $F_1$ and in the one copy of $F_4$, giving total degree three; the row for $Q_{7,i}$ records that it appears once in each of $F_4,F_5,F_6$, again degree three.  Where a private item appears in ``two copies'' of a pattern, these are two distinct hyperedge occurrences (parallel edges), not two copies of the item inside one hyperedge.  Reading every row this way confirms that the forward construction uses each private occurrence exactly three times.

The multiset $\mathcal H$ therefore satisfies the conditions of \cref{n3dm:lem:copy-forgetting}.  Recoloring its incidence graph produces a feasible matching among $X^\star,Y^\star,Z^\star$.  Therefore $I^\star$ is a ``yes'' instance.

\subsection{The ``only if'' case}\label{n3dm:sec:only-if}

Assume that the constructed symmetric instance $I^\star$ is feasible.  By \cref{n3dm:lem:copy-forgetting}, after forgetting the outer class labels we obtain a multiset $\mathcal H$ of target-sum triples on $S$ such that every occurrence of $S$ has degree three.  By \cref{n3dm:lem:no-carry}, every hyperedge has vector sum $\Kvec$.

\subsubsection{Every edge has an intended role pattern}

The first coordinate of every edge sums to $300$.  The role filter, \cref{n3dm:lem:role-filter}, therefore implies that every edge has one of the role patterns $F_0,F_1,\ldots,F_9$.  An $F_0$-edge is called a \emph{main edge}; all other edges are \emph{filler edges}.  In particular, no private $Q$-item can occur in a main edge, and each port type $P_r$ can occur only in $F_0$ or its corresponding filler pattern $F_r$.

\subsubsection{The group coordinates synchronize private items}

We first isolate the elementary two-moment argument used by the gadget.
\begin{lemma}[Two moments force a common group]\label{n3dm:lem:two-moments}
Let $\alpha,\beta,\gamma$ be nonzero real numbers satisfying $\alpha+\beta+\gamma=0$.  If indices $i,j,k$ satisfy
\[
 \alpha i+\beta j+\gamma k=0,
 \qquad
 \alpha i^2+\beta j^2+\gamma k^2=0,
\]
then $i=j=k$.
\end{lemma}

\begin{proof}
Since the coefficients are nonzero and sum to zero, two of them share one sign and the third has the opposite sign.  After permuting the terms and multiplying all coefficients by $-1$ if necessary, assume $\alpha,\beta>0$ and $\gamma=-(\alpha+\beta)$.    The first equation gives $k=(\alpha i+\beta j)/(\alpha+\beta)$.  Substituting into the second equation gives
\[
 0=\alpha i^2+\beta j^2-(\alpha+\beta)k^2
   =\frac{\alpha\beta}{\alpha+\beta}(i-j)^2.
\]
Hence $i=j$, and the first equation then yields $k=i$.
\end{proof}

An $F_1$-edge means a hyperedge whose three role types are $P_1,Q_1,Q_2$; before the group coordinates are analyzed, its three occurrences may in principle carry different group indices.  In such an edge, the \emph{private occurrences} are the two $Q$-type occurrences.  They are called private because they are internal components of the filler gadget, unlike the port occurrence $P_1$, which is the interface through which the group can participate in a main source-encoding edge.

Consider an $F_1$-edge with private occurrences $Q_{1,i}$ and $Q_{2,j}$.  Its group coordinate is $2i-2j$, so equality to zero implies $i=j$.  The same argument applies to the private pairs in $F_2$ and $F_3$. For each of $F_4,\ldots,F_9$, the three nonzero group coefficients are one of $(2,1,-3)$, $(-2,1,1)$, $(2,-3,1)$.  In each case, the coefficients are nonzero, sum to zero, and have two entries of one sign and one of the opposite sign, so the two-moment lemma forces all three private occurrences of such an edge to have one common group index.

Consequently, every filler edge has all of its private $Q$-items in a uniquely determined group.  In $F_1,F_2,F_3$, the port index is not yet forced to equal that group index; this will be controlled by the source coordinate.

\subsubsection{Private degrees force the filler multiplicities}

A \emph{filler multiplicity} is the number of times a specified filler-edge pattern occurs in the unlabeled triple system, counting parallel hyperedge occurrences separately.  It is a count of edges, not an item weight and not a count of physical $X$-, $Y$-, or $Z$-copies.  Determining these multiplicities tells us how much of each port's degree budget is consumed locally.

Fix a group $i$ and a filler pattern $F_r$ with $r\in\{1,\ldots,9\}$.
Look through the triple system $\mathcal H$ and select the edges that
(i) have pattern $F_r$, and (ii) have every one of their private items
in group $i$.  Then $x_{r,i}$ is simply the number of such edges:
\[
  x_{r,i}
  \;=\;
  \bigl|\{\,e\in\mathcal H \;:\;
  e \text{ has pattern } F_r,\ \text{and all private items of } e
  \text{ lie in group } i\,\}\bigr|.
\]
In other words, among all the filler edges of type $F_r$ present in the
solution, $x_{r,i}$ counts how many are ``localized'' at group $i$ in the
sense that all their private items come from that group.  Three remarks
make this unambiguous.

\begin{itemize}
\item It counts \emph{edges}, not incidences: a single $F_r$-edge
      contributes $1$ to $x_{r,i}$, no matter how many private items it
      contains (two for $F_1,F_2,F_3$, three for $F_4,\ldots,F_9$).
\item The condition asks that \emph{all} of an edge's private items lie
      in group $i$, not merely one of them.  This is meaningful because
      the previous subsection showed that the group coordinates force
      every filler edge's private items to share a single common group
      index, so each filler edge is localized at exactly one group.
\item For $F_1,F_2,F_3$ the condition constrains only the two private
      items; the port of such an edge may belong to a different group,
      and its value is pinned separately in the next subsection.
\end{itemize}

 Before writing the equations, it helps to have a concrete picture of
what $x_{r,i}$ measures.  Consider the intended solution --- the one the
forward construction produces.  There, each group $i$ contributes two
copies of each of the port-attaching patterns $F_1,F_2,F_3$ and one copy
of each of the private-only patterns $F_4,\ldots,F_9$; that is,
\[
  x_{1,i}=x_{2,i}=x_{3,i}=2,
  \qquad
  x_{4,i}=x_{5,i}=\cdots=x_{9,i}=1 .
\]
For instance, $x_{7,i}=1$ records that a \emph{single} $F_7$-edge is
localized at group $i$.  That one edge, $F_7=(Q_4,Q_8,Q_{10})$, meets
three private items, contributing one incidence to each of
$Q_{4,i},Q_{8,i},Q_{10,i}$ --- which is why $x_{7,i}$ appears in the
degree equation of each of those three items below.  These are the
values we expect; the equations that follow prove they are the
\emph{only} nonnegative-integer values consistent with the degree
constraints.
Each private occurrence has total degree three, so
\begin{align}
 x_{1,i}+x_{4,i}&=3, & x_{1,i}+x_{5,i}&=3,\label{n3dm:eq:q12}\\
 x_{2,i}+x_{6,i}&=3, & x_{2,i}+x_{7,i}&=3,\label{n3dm:eq:q34}\\
 x_{3,i}+x_{8,i}&=3, & x_{3,i}+x_{9,i}&=3,\label{n3dm:eq:q56}\\
 x_{4,i}+x_{5,i}+x_{6,i}&=3,\label{n3dm:eq:q7}\\
 x_{4,i}+x_{7,i}+x_{8,i}&=3,\label{n3dm:eq:q8}\\
 x_{5,i}+x_{8,i}+x_{9,i}&=3,\label{n3dm:eq:q9}\\
 x_{6,i}+x_{7,i}+x_{9,i}&=3.\label{n3dm:eq:q10}
\end{align}

\paragraph{Collective meaning of \eqref{n3dm:eq:q12}--\eqref{n3dm:eq:q10}.}
These are the reverse-direction versions of the incidence ledger in \cref{n3dm:tab:forward-private-degrees}.  Each equality is the degree-balance equation for one private occurrence $Q_{\ell,i}$.  Its left-hand side adds the multiplicities of all filler patterns that contain that role, and its right-hand side is $3$ because every underlying occurrence has total degree three.  For instance, $Q_{1,i}$ can appear only in $F_1$ and $F_4$, which gives $x_{1,i}+x_{4,i}=3$; $Q_{7,i}$ can appear only in $F_4,F_5,F_6$, which gives $x_{4,i}+x_{5,i}+x_{6,i}=3$.

The first three numbered lines pair the balances for $Q_1,Q_2$, then $Q_3,Q_4$, then $Q_5,Q_6$.  The final four lines give the balances for $Q_7,Q_8,Q_9,Q_{10}$.  Because the role filter excludes every other use of a private item, these seven numbered lines contain all ten private-degree constraints and no hidden terms.  Notice also that no source value appears: this is a purely combinatorial locking system.  Its unique nonnegative integral solution forces two copies of each port-attaching pattern $F_1,F_2,F_3$ and one copy of each private-only pattern $F_4,\ldots,F_9$.

The paired equations imply
\[
 x_{4,i}=x_{5,i}=:a,\qquad x_{6,i}=x_{7,i}=:b,\qquad x_{8,i}=x_{9,i}=:c.
\]

Equations \eqref{n3dm:eq:q7}--\eqref{n3dm:eq:q10} become
\[
 2a+b=3,
 \qquad a+b+c=3,
 \qquad a+2c=3,
 \qquad 2b+c=3.
\]
Subtracting the second equation from the first gives $a=c$, and subtracting the second from the third gives $c=b$.  Hence $a=b=c$, and $a+b+c=3$ yields
\[
 a=b=c=1.
\]
Returning to \eqref{n3dm:eq:q12}--\eqref{n3dm:eq:q56}, we obtain the forced multiplicities
\begin{equation}\label{eq:forced-multiplicities}
 x_{1,i}=x_{2,i}=x_{3,i}=2,
 \qquad
 x_{4,i}=\cdots=x_{9,i}=1.
\end{equation}
Thus every group consumes exactly two filler incidences of each port type and all three incidences of every private item, precisely as suggested by \cref{n3dm:fig:degree-budget}.

\subsubsection{The source coordinate identifies the port values}

Consider an $F_1$-edge whose private items come from group $i$.  Their source-coordinate sum is
\begin{equation}\label{eq:f1-private-sum}
 \eta_{Q_1,i}+\eta_{Q_2,i}
 =-3b_i+(3t-3a_i+3b_i)
 =3t-3a_i.
\end{equation}
If the port in this edge is $P_{1,j}$, the total source coordinate equals $3t$ if and only if
\[
 3a_j+(3t-3a_i)=3t,
\]
that is, if and only if $a_j=a_i$.

Similarly,
\begin{align}
 \eta_{Q_3,i}+\eta_{Q_4,i}&=3t-3b_i,\label{eq:f2-private-sum}\\
 \eta_{Q_5,i}+\eta_{Q_6,i}&=3t-3c_i.\label{eq:f3-private-sum}
\end{align}
Therefore an $F_2$-edge associated with group $i$ can use only a port $P_{2,j}$ with $b_j=b_i$, and an $F_3$-edge associated with group $i$ can use only a port $P_{3,j}$ with $c_j=c_i$.

Equal-valued source occurrences may be exchanged among groups, which is harmless.  What matters is that the multiplicities of each numerical value are preserved.

\subsubsection{The main edges recover the source matching}

Everything so far has concerned filler edges.  This final step shows that what the fillers \emph{leave behind} is exactly a source matching.

\emph{Setup.}  Recall that $A=\mset{a_1,\ldots,a_n}$ is the multiset of source $A$-values, and that the $P_1$-ports carry precisely these values (similarly $P_2\leftrightarrow B$ and $P_3\leftrightarrow C$).  Fix any source value $\alpha$, and let
\[
 r_A(\alpha)=|\{i\in[n]:a_i=\alpha\}|
\]
be the number of groups whose $A$-value equals $\alpha$.  Because the previous subsection pinned only a port's \emph{value} and not its group, we count per value rather than per group.

\emph{Incidences available.}  There are $r_A(\alpha)$ ports $P_1$ of value $\alpha$, each of degree three, so they supply $3r_A(\alpha)$ incidences of value $\alpha$ in total.

\emph{Incidences consumed by fillers.}  A $P_1$-port can appear only in an $F_1$-edge or in the main pattern $F_0$ (role filter).  Each value-$\alpha$ group contributes exactly two $F_1$-edges --- the forced multiplicity $x_{1,i}=2$ of \cref{eq:forced-multiplicities} --- and by \cref{eq:f1-private-sum} each of those edges uses a $P_1$-port of value $\alpha$.  So filler edges consume $2r_A(\alpha)$ incidences of value $\alpha$.

\emph{Incidences left for main edges.}  Thus, 
\[
 3r_A(\alpha)-2r_A(\alpha)=r_A(\alpha)
\]
incidences of value-$\alpha$ $P_1$-ports remain for main edges.  Since this holds for every value $\alpha$, the multiset of $A$-values in the $P_1$ positions of the main edges is exactly $A$.  The identical count with $(F_2,P_2)$ and $(F_3,P_3)$ shows the $P_2$ positions contain exactly $B$ and the $P_3$ positions exactly $C$.

\emph{Reading off the matching.}  In particular there are exactly $n$ main edges, each of the form
\[
 F_0(i,j,k)=(P_{1,i},P_{2,j},P_{3,k}),
\]
whose source-coordinate equation $3a_i+3b_j+3c_k=3t$ gives $a_i+b_j+c_k=t$.  The $n$ main edges therefore use the source multisets $A,B,C$ exactly once each and meet the common target, so they form a feasible matching of the source \NthreeDM{} instance.  If a value repeats, its equal-valued labeled occurrences may be assigned to the corresponding main-edge slots in any way; the per-value counts above guarantee the numbers agree.

This completes the ``Only if'' direction.

\section{Unary size and completion of the theorem}\label{n3dm:sec:size}

Let
\[
 U=\max\bigl(\{t\}\cup\{a_i,b_i,c_i:i\in[n]\}\bigr).
\]
The role coordinate is bounded by an absolute constant.  Every source-coordinate expression in \cref{n3dm:tab:weights} has magnitude $O(U)$.  Since $|h_R|\le 3$ and $i\le n$, the group and group-square coordinates have magnitudes $O(n)$ and $O(n^2)$, respectively.  The shifts $D_r$ and the base $L$ in \cref{n3dm:eq:base-choice} can therefore be chosen with
\[
 D_r=O(U+n^2),
 \qquad
 L=O(U+n^2).
\]
Every encoded item weight and the encoded target are smaller than $L^4$, and hence are bounded by
\[
 O\bigl((U+n^2)^4\bigr).
\]

The common occurrence set has $13n$ elements, and the ordinary output lists three copies, for $39n$ item occurrences in total.  Its unary length is therefore
\[
 O\bigl(n(U+n^2)^4\bigr),
\]
which is polynomial in the unary length of the source instance because both $n$ and $U$ are bounded by that length.  All shifts, base choices, and encoded weights can be computed in polynomial time.

Sections~\ref{n3dm:sec:if} and~\ref{n3dm:sec:only-if} prove
\[
 I\text{ is feasible}
 \quad\Longleftrightarrow\quad
 I^\star\text{ is feasible}.
\]
Thus unary \NthreeDM{} reduces in polynomial time to unary \SymNthreeDM.  The latter is in \NP, so it is unary \NP-complete.  This proves \cref{n3dm:thm:main}.  \qed

\section{Additional comments}\label{sec:comments}

\begin{remark}[The meaning of $X=Y=Z$]
The theorem concerns the standard three-class problem.  The symbols $X,Y,Z$ denote disjoint labeled classes, and feasibility always requires exactly one item from each class.  Equality means only that the three classes carry the same multiset of numerical weights.  The copy-forgetting lemma does not discard this requirement: it translates it into degree three and then restores the three labels by edge coloring.
\end{remark}

\begin{remark}[Repeated values]
No distinctness assumption is used.  The reverse direction works with multiplicity functions such as $r_A(\alpha)$, so equal-valued occurrences may be exchanged without changing the recovered matching.
\end{remark}

\begin{remark}[Parallel hyperedges]
The forward construction uses two parallel copies of $F_1(i)$, $F_2(i)$, and $F_3(i)$.  Parallelism exists only in the intermediate triple system.  The incidence-coloring step assigns the parallel edge occurrences different combinations of labeled copies as needed, producing an ordinary partition of the three output classes.
\end{remark}

\begin{remark}[Fixed finite certificates]
The role codes are fixed constants independent of the input.  Their sole purpose is to realize the ten permitted role patterns as exactly the triples with role sum $300$.  Section~\ref{app:role-codes} gives an explicit finite certificate.  A short independent verification program may also accompany the source distribution; it is not needed for the asymptotic argument.
\end{remark}

\begin{remark}[Scope]
A different problem in which one partitions a single unlabeled multiset into arbitrary triples is not the standard \NthreeDM{} problem and is not covered by this theorem.  Such a variant would require the numerical construction itself to encode the missing class labels.
\end{remark}

\section{Complementary-pair certificate for the role codes}\label{app:role-codes}

For a role $R$, \cref{n3dm:tab:complements} lists every unordered role pair $\{R',R''\}$ satisfying
\[
 \rho_R+\rho_{R'}+\rho_{R''}=300.
\]
Checking the displayed additions is immediate.  Since every unordered triple has at least one role that can be selected as $R$, the table is an exhaustive certificate for \cref{n3dm:lem:role-filter}.  The same allowed triple appears in three rows, once for each of its roles.

\begin{table}[H]
\centering
\caption{All complementary role pairs for target role sum $300$.}
\label{n3dm:tab:complements}
\small
\renewcommand{\arraystretch}{1.12}
\begin{tabularx}{\textwidth}{@{}c X@{}}
\toprule
$R$ & All pairs $\{R',R''\}$ with $\rho_R+\rho_{R'}+\rho_{R''}=300$ \\
\midrule
$P_1$ & $\{P_2,P_3\}$; $\{Q_1,Q_2\}$ \\
$P_2$ & $\{P_1,P_3\}$; $\{Q_3,Q_4\}$ \\
$P_3$ & $\{P_1,P_2\}$; $\{Q_5,Q_6\}$ \\
$Q_1$ & $\{P_1,Q_2\}$; $\{Q_7,Q_8\}$ \\
$Q_2$ & $\{P_1,Q_1\}$; $\{Q_7,Q_9\}$ \\
$Q_3$ & $\{P_2,Q_4\}$; $\{Q_7,Q_{10}\}$ \\
$Q_4$ & $\{P_2,Q_3\}$; $\{Q_8,Q_{10}\}$ \\
$Q_5$ & $\{P_3,Q_6\}$; $\{Q_8,Q_9\}$ \\
$Q_6$ & $\{P_3,Q_5\}$; $\{Q_9,Q_{10}\}$ \\
$Q_7$ & $\{Q_1,Q_8\}$; $\{Q_2,Q_9\}$; $\{Q_3,Q_{10}\}$ \\
$Q_8$ & $\{Q_1,Q_7\}$; $\{Q_4,Q_{10}\}$; $\{Q_5,Q_9\}$ \\
$Q_9$ & $\{Q_2,Q_7\}$; $\{Q_5,Q_8\}$; $\{Q_6,Q_{10}\}$ \\
$Q_{10}$ & $\{Q_3,Q_7\}$; $\{Q_4,Q_8\}$; $\{Q_6,Q_9\}$ \\
\bottomrule
\end{tabularx}
\end{table}

\paragraph{How to read \cref{n3dm:tab:complements}.}
Fix a role in the first column and subtract its code from the target $300$.  The second column lists every unordered pair of role codes that supplies the remaining amount.  For example, the $P_1$ row says that after choosing code $31$, the only completing pairs are $P_2,P_3$, whose codes sum to $269$, and $Q_1,Q_2$, whose codes also sum to $269$.  These produce exactly the main pattern $F_0$ and filler pattern $F_1$ involving $P_1$.

The same allowed role triple appears in three different rows because any one of its roles can be selected as the first role.  That redundancy is useful: it makes the certificate symmetric and easy to audit.  Exhaustiveness follows because any unordered triple with total code $300$ must, after one of its roles is fixed, appear as a complementary pair in that role's row.  Hence the absence of any additional pair in the table certifies that no unintended role triple can pass the role coordinate.  Since the codes and this certificate are fixed independently of the source instance, this check contributes only constant-size logic to the reduction.

\section{Fully worked examples}

\subsection{Guide to the examples}

\begin{center}
\setlength{\fboxsep}{9pt}
\fbox{%
\begin{minipage}{0.92\textwidth}
\textbf{How to read the full Part I examples}

\smallskip
The examples support two useful reading paths.  Both begin with the same
conceptual question: how can three asymmetric source roles survive after the
construction replaces the source classes by three identical copies of one
common occurrence set?

\medskip
\textbf{Conceptual path.}
Read the source matching and the definition of the groups $S_i$; then follow
one representative port weight, one representative filler identity, the
construction of the main edges, and the reverse degree argument.  At this
level, focus on four mechanisms:
\begin{enumerate}[label=(\roman*),leftmargin=2em,itemsep=0.15em,topsep=0.25em]
  \item $P_1,P_2,P_3$ preserve the three source roles inside the common set
        $S$;
  \item the filler patterns consume exactly two incidences of each port,
        leaving one incidence for a main edge;
  \item the role coordinate enforces pattern type, the source coordinate
      enforces the required source sum, and the group and group-square
      coordinates jointly enforce locality; and

  \item copy-forgetting and bipartite edge coloring translate between the
        degree-three triple system and a legal matching in
        $X^\star,Y^\star,Z^\star$.
\end{enumerate}
The long weight tables may be treated initially as certificates for these
claims rather than read entry by entry.

\medskip
\textbf{Full verification path.}
Check every mechanical layer in order: the $4$ source groups; the $52$
underlying occurrences in $S$; the identical $52$-weight catalogues in each
of the three output classes; every shift, base digit, encoded weight, and the
target $K^\star$; the nine filler identities in every group; the $48$ filler
hyperedge occurrences and $4$ main hyperedges; degree three at every
underlying occurrence; and all $52$ target equalities after recoloring.  Then
follow the reverse equations to verify that the filler multiplicities are
forced and that the four remaining main edges recover a source matching.  The
``no'' instance should be read as a stress test of the same logic: the local
construction remains valid, but the required system of four disjoint main
edges cannot be completed.

\medskip
\textbf{Completion checkpoint.}
The forward direction is complete only after all $156$ labeled output
occurrences have been used exactly once in $52$ legal triples of weight
$K^\star$.  The reverse direction is complete only after an arbitrary feasible
complete output matching has been shown to yield a feasible matching of the
original source instance.
\end{minipage}}
\end{center}

We provide fully worked examples of a ``yes'' instance and a ``no'' instance. The first starts from the classical N3DM ``yes'' instance
\[
A=\mset{3,4,5,6},\qquad B=\mset{13,11,9,7},\qquad C=\mset{20,21,22,23},
\]
and the second starts from the ``no'' instance obtained by changing the middle class to
\(B=\mset{13,12,8,7}\).  Both examples follow the same construction, mixed-radix encoding, filler-degree analysis, copy-forgetting argument, and recoloring logic used in Part I of the tutorial.  All source coordinates, shifts, encoded item weights, filler checks, main-edge checks, and reverse-direction arguments are given explicitly.

\subsection{Construction used in both examples}
The source classes have four labeled occurrences each.  For each source index \(i\in\{1,2,3,4\}\), create
\[
S_i=\{P_{1,i},P_{2,i},P_{3,i},Q_{1,i},\ldots,Q_{10,i}\},
\qquad S=\biguplus_{i=1}^4 S_i.
\]
Thus \(|S_i|=13\) and \(|S|=52\).  The constructed SN3DM classes are the three disjoint labeled copies
\[
X^\star=S^X,\qquad Y^\star=S^Y,\qquad Z^\star=S^Z.
\]
Each output class therefore contains 52 occurrences, the full output contains 156 labeled occurrences, and the three output weight multisets are identical.

For an occurrence of role \(R\) in source group \(i\), the four-dimensional weight is
\[
\omega(R_i)=\bigl(\rho_R,\eta_{R,i},h_R i,h_R i^2\bigr),
\]
with vector target
\[
K=(300,3t,0,0).
\]
The fixed role data are reproduced in Table~\ref{tab:roledata} so that the examples are self-contained.
\begin{table}[htbp]
\centering
\small
\caption{Role data used in both examples.}\label{tab:roledata}
\begin{tabular}{c r r l}
\toprule
role $R$ & $\rho_R$ & $h_R$ & $\eta_{R,i}$ \\
\midrule
$P_{1}$ & 31 & 0 & $3a_i$ \\
$P_{2}$ & 8 & 0 & $3b_i$ \\
$P_{3}$ & 261 & 0 & $3c_i$ \\
$Q_{1}$ & 184 & 2 & $-3b_i$ \\
$Q_{2}$ & 85 & -2 & $3t-3a_i+3b_i$ \\
$Q_{3}$ & 177 & -2 & $3t-2a_i+b_i+c_i$ \\
$Q_{4}$ & 115 & 2 & $2a_i-4b_i-c_i$ \\
$Q_{5}$ & 23 & 2 & $a_i-2b_i-2c_i$ \\
$Q_{6}$ & 16 & -2 & $3t-a_i+2b_i-c_i$ \\
$Q_{7}$ & 27 & 1 & $2a_i-b_i-c_i$ \\
$Q_{8}$ & 89 & -3 & $3t-2a_i+4b_i+c_i$ \\
$Q_{9}$ & 188 & 1 & $a_i-2b_i+c_i$ \\
$Q_{10}$ & 96 & 1 & $0$ \\
\bottomrule
\end{tabular}
\end{table}

The permitted role patterns are
\begin{align*}
F_0&=(P_1,P_2,P_3),\\
F_1&=(P_1,Q_1,Q_2),&F_2&=(P_2,Q_3,Q_4),&F_3&=(P_3,Q_5,Q_6),\\
F_4&=(Q_1,Q_7,Q_8),&F_5&=(Q_2,Q_7,Q_9),&F_6&=(Q_3,Q_7,Q_{10}),\\
F_7&=(Q_4,Q_8,Q_{10}),&F_8&=(Q_5,Q_8,Q_9),&F_9&=(Q_6,Q_9,Q_{10}).
\end{align*}
The role codes admit exactly these ten patterns at role sum 300.  A local filler occurrence \(F_r(i)\), \(r\geq 1\), takes all its roles from group \(i\).  A main edge may combine three different source indices and is written
\[
F_0(i,j,k)=(P_{1,i},P_{2,j},P_{3,k}).
\]
After coordinate shifts \(D_0,D_1,D_2,D_3\) and a no-carry base \(L\), the ordinary integer weight is
\[
W(u)=\sum_{r=0}^3 \bar\omega_r(u)L^r,
\qquad
K^\star=\sum_{r=0}^3 \bar K_rL^r.
\]
A vector target equality and an integer target equality are then equivalent coordinate by coordinate.  Finally, a 3-regular target-sum triple system on \(S\) is converted to a legal matching in \(X^\star,Y^\star,Z^\star\) by decomposing its incidence bipartite graph into three perfect matchings and using them as the colors \(X,Y,Z\).

\subsection{A source-to-output ``yes'' instance}
\subsubsection{The classical source instance}
Take
\[
A=\mset{3,4,5,6},\qquad
B=\mset{13,11,9,7},\qquad
C=\mset{20,21,22,23}.
\]
Index each class in the displayed order.  The class sums are
\[
\sum_i a_i=18,\qquad \sum_i b_i=40,\qquad \sum_i c_i=86.
\]
A solution has four triples, so the common target is forced by total weight:
\[
4t=18+40+86=144,\qquad t=36.
\]
The diagonal index triples form a feasible source matching.
\begin{table}[htbp]
\centering
\small
\caption{A feasible source matching for the ``yes'' instance.}\label{tab:yes-source-match}
\begin{tabular}{c c c c c c}
\toprule
$r$ & index triple $(i,j,k)$ & $a_i$ & $b_j$ & $c_k$ & sum \\
\midrule
1 & $(1,1,1)$ & 3 & 13 & 20 & 36 \\
2 & $(2,2,2)$ & 4 & 11 & 21 & 36 \\
3 & $(3,3,3)$ & 5 & 9 & 22 & 36 \\
4 & $(4,4,4)$ & 6 & 7 & 23 & 36 \\
\bottomrule
\end{tabular}
\end{table}

In \cref{tab:yes-source-match}, every row sums to 36, and every index \(1,2,3,4\) appears exactly once in each source position.  Notice that the three source classes are not numerically identical.  Their asymmetry is represented internally by the port roles \(P_1,P_2,P_3\), whereas each output class contains a copy of every role.

\subsubsection{The occurrence sets and identical output classes}
For each \(i\in\{1,2,3,4\}\), the set \(S_i\) contains three ports and ten private items.  Hence
\[
|S|=4\cdot 13=52,
\qquad |X^\star|=|Y^\star|=|Z^\star|=52,
\qquad |X^\star\uplus Y^\star\uplus Z^\star|=156.
\]
The construction assigns one integer weight \(W(u)\) to every underlying occurrence \(u\in S\), then gives the same weight to \(u^X,u^Y,u^Z\).  Therefore the three displayed weight catalogues are identical even though \(A,B,C\) are different.

\subsubsection{All specialized source-coordinate values}
Here \(3t=108\).  The port coordinates are \(3a_i,3b_i,3c_i\), and the private coordinates are obtained by substituting the displayed \((a_i,b_i,c_i)\) into Table~\ref{tab:roledata}.

\begin{table}[htbp]
\centering
\small
\caption{Port source-coordinate values for the ``yes'' instance.}\label{tab:yes-ports}
\begin{tabular}{c r r r r r r}
\toprule
$i$ & $a_i$ & $b_i$ & $c_i$ & $\eta_{P_1,i}$ & $\eta_{P_2,i}$ & $\eta_{P_3,i}$ \\
\midrule
1 & 3 & 13 & 20 & 9 & 39 & 60 \\
2 & 4 & 11 & 21 & 12 & 33 & 63 \\
3 & 5 & 9 & 22 & 15 & 27 & 66 \\
4 & 6 & 7 & 23 & 18 & 21 & 69 \\
\bottomrule
\end{tabular}
\end{table}
\begin{table}[h!]
\centering
\scriptsize
\setlength{\tabcolsep}{3.5pt}
\caption{All private source-coordinate values for the ``yes'' instance.}\label{tab:yes-private}
\begin{tabular}{c c r r r r r r r r r r}
\toprule
$i$ & $(a_i,b_i,c_i)$ & $Q_1$ & $Q_2$ & $Q_3$ & $Q_4$ & $Q_5$ & $Q_6$ & $Q_7$ & $Q_8$ & $Q_9$ & $Q_{10}$ \\
\midrule
1 & $(3,13,20)$ & -39 & 138 & 135 & -66 & -63 & 111 & -27 & 174 & -3 & 0 \\
2 & $(4,11,21)$ & -33 & 129 & 132 & -57 & -60 & 105 & -24 & 165 & 3 & 0 \\
3 & $(5,9,22)$ & -27 & 120 & 129 & -48 & -57 & 99 & -21 & 156 & 9 & 0 \\
4 & $(6,7,23)$ & -21 & 111 & 126 & -39 & -54 & 93 & -18 & 147 & 15 & 0 \\
\bottomrule
\end{tabular}
\end{table}
For example, in group 1,
\[
\eta_{Q_2,1}=108-3(3)+3(13)=138,
\qquad
\eta_{Q_4,1}=2(3)-4(13)-20=-66,
\]
and
\[
\eta_{Q_8,1}=108-2(3)+4(13)+20=174.
\]
The irregular-looking values are deliberate cancellation terms: every intended local filler pattern has source-coordinate sum 108.

\subsubsection{Mixed-radix shifts, base, target, and every item weight}
The most negative source coordinate is \(-66\), attained by \(Q_{4,1}\).  The most negative group coordinate is \(-12\), and the most negative group-square coordinate is \(-48\), both attained by \(Q_{8,4}\).  We therefore take the minimal shifts
\[
D_0=0,\qquad D_1=66,\qquad D_2=12,\qquad D_3=48.
\]
The shifted target digits are
\[
\bar K=(300,108+3\cdot 66,3\cdot 12,3\cdot 48)=(300,306,36,144).
\]
The largest shifted item digits are
\[
261,\qquad 240,\qquad 20,\qquad 80.
\]
Consequently,
\[
3\cdot261=783,\quad 3\cdot240=720,\quad 3\cdot20=60,\quad 3\cdot80=240,
\]
and every target digit is at most 306.  The smallest admissible integer base is therefore
\[
L=784>783.
\]
Its required powers are
\[
L^2=614{,}656,\qquad L^3=481{,}890{,}304.
\]
The encoded target is
\[
\begin{aligned}
K^\star
&=300+306(784)+36(614{,}656)+144(481{,}890{,}304)\\
&=69{,}414{,}571{,}596.
\end{aligned}
\]
For example,
\[
\begin{aligned}
W(P_{1,1})
&=31+75(784)+12(614{,}656)+48(481{,}890{,}304)\\
&=23{,}138{,}169{,}295,
\end{aligned}
\]
whereas \(Q_{8,4}\) has shifted digits \((89,213,0,0)\), so
\[
W(Q_{8,4})=89+213(784)=167{,}081.
\]
The small latter weight is not an error: both high-order locality digits become zero after shifting.
\begingroup
\scriptsize
\renewcommand{\arraystretch}{0.78}
\setlength{\tabcolsep}{4pt}
\begin{longtable}{c c c c c r}
\caption{All 52 vector weights, shifted digit vectors, and encoded positive integer weights for the ``yes'' instance.}\label{tab:yes-weights}\\
\toprule
$i$ & $(a_i,b_i,c_i)$ & item $u$ & $\omega(u)$ & $\bar\omega(u)$ & $W(u)$ \\
\midrule
\endfirsthead
\caption[]{All 52 vector weights, shifted digit vectors, and encoded positive integer weights for the ``yes'' instance. (continued)}\\
\toprule
$i$ & $(a_i,b_i,c_i)$ & item $u$ & $\omega(u)$ & $\bar\omega(u)$ & $W(u)$ \\
\midrule
\endhead
\midrule\multicolumn{6}{r}{\emph{continued on next page}}\\
\endfoot
\bottomrule
\endlastfoot
1 & $(3,13,20)$ & $P_{1,1}$ & $(31,9,0,0)$ & $(31,75,12,48)$ & {23{,}138{,}169{,}295} \\
1 & $(3,13,20)$ & $P_{2,1}$ & $(8,39,0,0)$ & $(8,105,12,48)$ & {23{,}138{,}192{,}792} \\
1 & $(3,13,20)$ & $P_{3,1}$ & $(261,60,0,0)$ & $(261,126,12,48)$ & {23{,}138{,}209{,}509} \\
1 & $(3,13,20)$ & $Q_{1,1}$ & $(184,-39,2,2)$ & $(184,27,14,50)$ & {24{,}103{,}141{,}736} \\
1 & $(3,13,20)$ & $Q_{2,1}$ & $(85,138,-2,-2)$ & $(85,204,10,46)$ & {22{,}173{,}260{,}565} \\
1 & $(3,13,20)$ & $Q_{3,1}$ & $(177,135,-2,-2)$ & $(177,201,10,46)$ & {22{,}173{,}258{,}305} \\
1 & $(3,13,20)$ & $Q_{4,1}$ & $(115,-66,2,2)$ & $(115,0,14,50)$ & {24{,}103{,}120{,}499} \\
1 & $(3,13,20)$ & $Q_{5,1}$ & $(23,-63,2,2)$ & $(23,3,14,50)$ & {24{,}103{,}122{,}759} \\
1 & $(3,13,20)$ & $Q_{6,1}$ & $(16,111,-2,-2)$ & $(16,177,10,46)$ & {22{,}173{,}239{,}328} \\
1 & $(3,13,20)$ & $Q_{7,1}$ & $(27,-27,1,1)$ & $(27,39,13,49)$ & {23{,}620{,}646{,}027} \\
1 & $(3,13,20)$ & $Q_{8,1}$ & $(89,174,-3,-3)$ & $(89,240,9,45)$ & {21{,}690{,}783{,}833} \\
1 & $(3,13,20)$ & $Q_{9,1}$ & $(188,-3,1,1)$ & $(188,63,13,49)$ & {23{,}620{,}665{,}004} \\
1 & $(3,13,20)$ & $Q_{10,1}$ & $(96,0,1,1)$ & $(96,66,13,49)$ & {23{,}620{,}667{,}264} \\
2 & $(4,11,21)$ & $P_{1,2}$ & $(31,12,0,0)$ & $(31,78,12,48)$ & {23{,}138{,}171{,}647} \\
2 & $(4,11,21)$ & $P_{2,2}$ & $(8,33,0,0)$ & $(8,99,12,48)$ & {23{,}138{,}188{,}088} \\
2 & $(4,11,21)$ & $P_{3,2}$ & $(261,63,0,0)$ & $(261,129,12,48)$ & {23{,}138{,}211{,}861} \\
2 & $(4,11,21)$ & $Q_{1,2}$ & $(184,-33,4,8)$ & $(184,33,16,56)$ & {26{,}995{,}717{,}576} \\
2 & $(4,11,21)$ & $Q_{2,2}$ & $(85,129,-4,-8)$ & $(85,195,8,40)$ & {19{,}280{,}682{,}373} \\
2 & $(4,11,21)$ & $Q_{3,2}$ & $(177,132,-4,-8)$ & $(177,198,8,40)$ & {19{,}280{,}684{,}817} \\
2 & $(4,11,21)$ & $Q_{4,2}$ & $(115,-57,4,8)$ & $(115,9,16,56)$ & {26{,}995{,}698{,}691} \\
2 & $(4,11,21)$ & $Q_{5,2}$ & $(23,-60,4,8)$ & $(23,6,16,56)$ & {26{,}995{,}696{,}247} \\
2 & $(4,11,21)$ & $Q_{6,2}$ & $(16,105,-4,-8)$ & $(16,171,8,40)$ & {19{,}280{,}663{,}488} \\
2 & $(4,11,21)$ & $Q_{7,2}$ & $(27,-24,2,4)$ & $(27,42,14,52)$ & {25{,}066{,}933{,}947} \\
2 & $(4,11,21)$ & $Q_{8,2}$ & $(89,165,-6,-12)$ & $(89,231,6,36)$ & {17{,}351{,}920{,}073} \\
2 & $(4,11,21)$ & $Q_{9,2}$ & $(188,3,2,4)$ & $(188,69,14,52)$ & {25{,}066{,}955{,}276} \\
2 & $(4,11,21)$ & $Q_{10,2}$ & $(96,0,2,4)$ & $(96,66,14,52)$ & {25{,}066{,}952{,}832} \\
3 & $(5,9,22)$ & $P_{1,3}$ & $(31,15,0,0)$ & $(31,81,12,48)$ & {23{,}138{,}173{,}999} \\
3 & $(5,9,22)$ & $P_{2,3}$ & $(8,27,0,0)$ & $(8,93,12,48)$ & {23{,}138{,}183{,}384} \\
3 & $(5,9,22)$ & $P_{3,3}$ & $(261,66,0,0)$ & $(261,132,12,48)$ & {23{,}138{,}214{,}213} \\
3 & $(5,9,22)$ & $Q_{1,3}$ & $(184,-27,6,18)$ & $(184,39,18,66)$ & {31{,}815{,}854{,}632} \\
3 & $(5,9,22)$ & $Q_{2,3}$ & $(85,120,-6,-18)$ & $(85,186,6,30)$ & {14{,}460{,}542{,}965} \\
3 & $(5,9,22)$ & $Q_{3,3}$ & $(177,129,-6,-18)$ & $(177,195,6,30)$ & {14{,}460{,}550{,}113} \\
3 & $(5,9,22)$ & $Q_{4,3}$ & $(115,-48,6,18)$ & $(115,18,18,66)$ & {31{,}815{,}838{,}099} \\
3 & $(5,9,22)$ & $Q_{5,3}$ & $(23,-57,6,18)$ & $(23,9,18,66)$ & {31{,}815{,}830{,}951} \\
3 & $(5,9,22)$ & $Q_{6,3}$ & $(16,99,-6,-18)$ & $(16,165,6,30)$ & {14{,}460{,}526{,}432} \\
3 & $(5,9,22)$ & $Q_{7,3}$ & $(27,-21,3,9)$ & $(27,45,15,57)$ & {27{,}477{,}002{,}475} \\
3 & $(5,9,22)$ & $Q_{8,3}$ & $(89,156,-9,-27)$ & $(89,222,3,21)$ & {10{,}121{,}714{,}489} \\
3 & $(5,9,22)$ & $Q_{9,3}$ & $(188,9,3,9)$ & $(188,75,15,57)$ & {27{,}477{,}026{,}156} \\
3 & $(5,9,22)$ & $Q_{10,3}$ & $(96,0,3,9)$ & $(96,66,15,57)$ & {27{,}477{,}019{,}008} \\
4 & $(6,7,23)$ & $P_{1,4}$ & $(31,18,0,0)$ & $(31,84,12,48)$ & {23{,}138{,}176{,}351} \\
4 & $(6,7,23)$ & $P_{2,4}$ & $(8,21,0,0)$ & $(8,87,12,48)$ & {23{,}138{,}178{,}680} \\
4 & $(6,7,23)$ & $P_{3,4}$ & $(261,69,0,0)$ & $(261,135,12,48)$ & {23{,}138{,}216{,}565} \\
4 & $(6,7,23)$ & $Q_{1,4}$ & $(184,-21,8,32)$ & $(184,45,20,80)$ & {38{,}563{,}552{,}904} \\
4 & $(6,7,23)$ & $Q_{2,4}$ & $(85,111,-8,-32)$ & $(85,177,4,16)$ & {7{,}712{,}842{,}341} \\
4 & $(6,7,23)$ & $Q_{3,4}$ & $(177,126,-8,-32)$ & $(177,192,4,16)$ & {7{,}712{,}854{,}193} \\
4 & $(6,7,23)$ & $Q_{4,4}$ & $(115,-39,8,32)$ & $(115,27,20,80)$ & {38{,}563{,}538{,}723} \\
4 & $(6,7,23)$ & $Q_{5,4}$ & $(23,-54,8,32)$ & $(23,12,20,80)$ & {38{,}563{,}526{,}871} \\
4 & $(6,7,23)$ & $Q_{6,4}$ & $(16,93,-8,-32)$ & $(16,159,4,16)$ & {7{,}712{,}828{,}160} \\
4 & $(6,7,23)$ & $Q_{7,4}$ & $(27,-18,4,16)$ & $(27,48,16,64)$ & {30{,}850{,}851{,}611} \\
4 & $(6,7,23)$ & $Q_{8,4}$ & $(89,147,-12,-48)$ & $(89,213,0,0)$ & {167{,}081} \\
4 & $(6,7,23)$ & $Q_{9,4}$ & $(188,15,4,16)$ & $(188,81,16,64)$ & {30{,}850{,}877{,}644} \\
4 & $(6,7,23)$ & $Q_{10,4}$ & $(96,0,4,16)$ & $(96,66,16,64)$ & {30{,}850{,}865{,}792} \\
\end{longtable}
\endgroup

\paragraph{How to read \cref{tab:yes-weights}.}
Read each row from left to right: the first two columns identify the source group, the next column names the occurrence, the signed vector records its logical role, the shifted vector gives its nonnegative radix digits, and the last column is the resulting scalar weight.  For a first audit, verify one port and one private item in each group and compare the remaining rows with the same role formulas.  The proof uses the digitwise equalities and the no-carry bound; the decimal size of a weight has no separate logical meaning.

The base choice gives the complete no-carry audit.  In every digit position, the sum of any three item digits is at most 783, strictly below \(L=784\), and every target digit is also below \(L\).  Therefore
\[
W(u)+W(v)+W(z)=K^\star
\quad\Longleftrightarrow\quad
\omega(u)+\omega(v)+\omega(z)=(300,108,0,0).
\]

\normalsize
\baselineskip=20pt

\subsubsection{All filler patterns and their numerical checks}
The nine target identities are shown in Table~\ref{tab:yes-identities}.  In each row, \(a,b,c\) denote the three source values belonging to one common group.  Since the three roles in a local pattern have one index \(i\), a zero sum of the \(h\)-coefficients simultaneously gives group-coordinate sum \(i\sum h_R=0\) and group-square sum \(i^2\sum h_R=0\).
\begin{table}[h!]
\centering
\scriptsize
\setlength{\tabcolsep}{5pt}
\caption{Local filler identities for target source coordinate $3t=108$.}\label{tab:yes-identities}
\begin{tabularx}{\textwidth}{@{}c c >{\raggedright\arraybackslash}X c@{}}
\toprule
pattern & role-code calculation & source-coordinate calculation & $h$-sum \\
\midrule
$F_1$ & $31+184+85=300$ & $3a-3b+(108-3a+3b)=108$ & $0+2-2=0$ \\
$F_2$ & $8+177+115=300$ & $3b+(108-2a+b+c)+(2a-4b-c)=108$ & $0-2+2=0$ \\
$F_3$ & $261+23+16=300$ & $3c+(a-2b-2c)+(108-a+2b-c)=108$ & $0+2-2=0$ \\
$F_4$ & $184+27+89=300$ & $-3b+(2a-b-c)+(108-2a+4b+c)=108$ & $2+1-3=0$ \\
$F_5$ & $85+27+188=300$ & $(108-3a+3b)+(2a-b-c)+(a-2b+c)=108$ & $-2+1+1=0$ \\
$F_6$ & $177+27+96=300$ & $(108-2a+b+c)+(2a-b-c)+0=108$ & $-2+1+1=0$ \\
$F_7$ & $115+89+96=300$ & $(2a-4b-c)+(108-2a+4b+c)+0=108$ & $2-3+1=0$ \\
$F_8$ & $23+89+188=300$ & $(a-2b-2c)+(108-2a+4b+c)+(a-2b+c)=108$ & $2-3+1=0$ \\
$F_9$ & $16+188+96=300$ & $(108-a+2b-c)+(a-2b+c)+0=108$ & $-2+1+1=0$ \\
\bottomrule
\end{tabularx}
\end{table}
\scriptsize
\setlength{\tabcolsep}{4pt}
\begin{longtable}{c c l p{7cm} r}
\caption{Every one of the 36 local filler pattern instances for the ``yes'' example.}\label{tab:yes-localinstances}\\
\toprule
$i$ & pattern & occurrences & source-coordinate calculation & encoded sum \\
\midrule
\endfirsthead
\caption[]{Every one of the 36 local filler pattern instances for the ``yes'' example. (continued)}\\
\toprule
$i$ & pattern & occurrences & source-coordinate calculation & encoded sum \\
\midrule
\endhead
\midrule\multicolumn{5}{r}{\emph{continued on next page}}\\
\endfoot
\bottomrule
\endlastfoot
1 & $F_{1}(1)$ & $P_{1,1}$, $Q_{1,1}$, $Q_{2,1}$ & $9 + (-39) + 138=108$ & {69{,}414{,}571{,}596} \\
1 & $F_{2}(1)$ & $P_{2,1}$, $Q_{3,1}$, $Q_{4,1}$ & $39 + 135 + (-66)=108$ & {69{,}414{,}571{,}596} \\
1 & $F_{3}(1)$ & $P_{3,1}$, $Q_{5,1}$, $Q_{6,1}$ & $60 + (-63) + 111=108$ & {69{,}414{,}571{,}596} \\
1 & $F_{4}(1)$ & $Q_{1,1}$, $Q_{7,1}$, $Q_{8,1}$ & $(-39) + (-27) + 174=108$ & {69{,}414{,}571{,}596} \\
1 & $F_{5}(1)$ & $Q_{2,1}$, $Q_{7,1}$, $Q_{9,1}$ & $138 + (-27) + (-3)=108$ & {69{,}414{,}571{,}596} \\
1 & $F_{6}(1)$ & $Q_{3,1}$, $Q_{7,1}$, $Q_{10,1}$ & $135 + (-27) + 0=108$ & {69{,}414{,}571{,}596} \\
1 & $F_{7}(1)$ & $Q_{4,1}$, $Q_{8,1}$, $Q_{10,1}$ & $(-66) + 174 + 0=108$ & {69{,}414{,}571{,}596} \\
1 & $F_{8}(1)$ & $Q_{5,1}$, $Q_{8,1}$, $Q_{9,1}$ & $(-63) + 174 + (-3)=108$ & {69{,}414{,}571{,}596} \\
1 & $F_{9}(1)$ & $Q_{6,1}$, $Q_{9,1}$, $Q_{10,1}$ & $111 + (-3) + 0=108$ & {69{,}414{,}571{,}596} \\
2 & $F_{1}(2)$ & $P_{1,2}$, $Q_{1,2}$, $Q_{2,2}$ & $12 + (-33) + 129=108$ & {69{,}414{,}571{,}596} \\
2 & $F_{2}(2)$ & $P_{2,2}$, $Q_{3,2}$, $Q_{4,2}$ & $33 + 132 + (-57)=108$ & {69{,}414{,}571{,}596} \\
2 & $F_{3}(2)$ & $P_{3,2}$, $Q_{5,2}$, $Q_{6,2}$ & $63 + (-60) + 105=108$ & {69{,}414{,}571{,}596} \\
2 & $F_{4}(2)$ & $Q_{1,2}$, $Q_{7,2}$, $Q_{8,2}$ & $(-33) + (-24) + 165=108$ & {69{,}414{,}571{,}596} \\
2 & $F_{5}(2)$ & $Q_{2,2}$, $Q_{7,2}$, $Q_{9,2}$ & $129 + (-24) + 3=108$ & {69{,}414{,}571{,}596} \\
2 & $F_{6}(2)$ & $Q_{3,2}$, $Q_{7,2}$, $Q_{10,2}$ & $132 + (-24) + 0=108$ & {69{,}414{,}571{,}596} \\
2 & $F_{7}(2)$ & $Q_{4,2}$, $Q_{8,2}$, $Q_{10,2}$ & $(-57) + 165 + 0=108$ & {69{,}414{,}571{,}596} \\
2 & $F_{8}(2)$ & $Q_{5,2}$, $Q_{8,2}$, $Q_{9,2}$ & $(-60) + 165 + 3=108$ & {69{,}414{,}571{,}596} \\
2 & $F_{9}(2)$ & $Q_{6,2}$, $Q_{9,2}$, $Q_{10,2}$ & $105 + 3 + 0=108$ & {69{,}414{,}571{,}596} \\
3 & $F_{1}(3)$ & $P_{1,3}$, $Q_{1,3}$, $Q_{2,3}$ & $15 + (-27) + 120=108$ & {69{,}414{,}571{,}596} \\
3 & $F_{2}(3)$ & $P_{2,3}$, $Q_{3,3}$, $Q_{4,3}$ & $27 + 129 + (-48)=108$ & {69{,}414{,}571{,}596} \\
3 & $F_{3}(3)$ & $P_{3,3}$, $Q_{5,3}$, $Q_{6,3}$ & $66 + (-57) + 99=108$ & {69{,}414{,}571{,}596} \\
3 & $F_{4}(3)$ & $Q_{1,3}$, $Q_{7,3}$, $Q_{8,3}$ & $(-27) + (-21) + 156=108$ & {69{,}414{,}571{,}596} \\
3 & $F_{5}(3)$ & $Q_{2,3}$, $Q_{7,3}$, $Q_{9,3}$ & $120 + (-21) + 9=108$ & {69{,}414{,}571{,}596} \\
3 & $F_{6}(3)$ & $Q_{3,3}$, $Q_{7,3}$, $Q_{10,3}$ & $129 + (-21) + 0=108$ & {69{,}414{,}571{,}596} \\
3 & $F_{7}(3)$ & $Q_{4,3}$, $Q_{8,3}$, $Q_{10,3}$ & $(-48) + 156 + 0=108$ & {69{,}414{,}571{,}596} \\
3 & $F_{8}(3)$ & $Q_{5,3}$, $Q_{8,3}$, $Q_{9,3}$ & $(-57) + 156 + 9=108$ & {69{,}414{,}571{,}596} \\
3 & $F_{9}(3)$ & $Q_{6,3}$, $Q_{9,3}$, $Q_{10,3}$ & $99 + 9 + 0=108$ & {69{,}414{,}571{,}596} \\
4 & $F_{1}(4)$ & $P_{1,4}$, $Q_{1,4}$, $Q_{2,4}$ & $18 + (-21) + 111=108$ & {69{,}414{,}571{,}596} \\
4 & $F_{2}(4)$ & $P_{2,4}$, $Q_{3,4}$, $Q_{4,4}$ & $21 + 126 + (-39)=108$ & {69{,}414{,}571{,}596} \\
4 & $F_{3}(4)$ & $P_{3,4}$, $Q_{5,4}$, $Q_{6,4}$ & $69 + (-54) + 93=108$ & {69{,}414{,}571{,}596} \\
4 & $F_{4}(4)$ & $Q_{1,4}$, $Q_{7,4}$, $Q_{8,4}$ & $(-21) + (-18) + 147=108$ & {69{,}414{,}571{,}596} \\
4 & $F_{5}(4)$ & $Q_{2,4}$, $Q_{7,4}$, $Q_{9,4}$ & $111 + (-18) + 15=108$ & {69{,}414{,}571{,}596} \\
4 & $F_{6}(4)$ & $Q_{3,4}$, $Q_{7,4}$, $Q_{10,4}$ & $126 + (-18) + 0=108$ & {69{,}414{,}571{,}596} \\
4 & $F_{7}(4)$ & $Q_{4,4}$, $Q_{8,4}$, $Q_{10,4}$ & $(-39) + 147 + 0=108$ & {69{,}414{,}571{,}596} \\
4 & $F_{8}(4)$ & $Q_{5,4}$, $Q_{8,4}$, $Q_{9,4}$ & $(-54) + 147 + 15=108$ & {69{,}414{,}571{,}596} \\
4 & $F_{9}(4)$ & $Q_{6,4}$, $Q_{9,4}$, $Q_{10,4}$ & $93 + 15 + 0=108$ & {69{,}414{,}571{,}596} \\
\end{longtable}
\normalsize

\paragraph{How to read \cref{tab:yes-identities,tab:yes-localinstances}.}
The shorter table gives the nine symbolic identities once; the long table then substitutes the four concrete source groups into those templates.  Read one pattern, such as $F_1$, first in symbolic form and then down its four numerical instances.  Repeating that comparison for the other patterns verifies legality.  These tables do not by themselves prove that no other triples are possible or that the displayed multiplicities are forced; those conclusions come from the role certificate, the locality equations, and the degree ledger.

\baselineskip=20pt

For the forward construction, each group contributes the multiset
\[
\mathcal F_i=
2F_1(i)\uplus2F_2(i)\uplus2F_3(i)\uplus F_4(i)\uplus\cdots\uplus F_9(i).
\]
Thus each group contributes 12 filler-edge occurrences, and the four groups contribute 48 filler-edge occurrences in total.  Table~\ref{tab:yes-localinstances} lists each of the 36 distinct local pattern instances once; \(F_1,F_2,F_3\) are each taken twice in the actual hyperedge multiset.

\subsubsection{The ``if'' direction in the example}
Add the four main edges corresponding to the source matching in Table~\ref{tab:yes-source-match}.
\begin{table}[h!]
\centering
\small
\caption{The four main edges and their encoded sums for the ``yes'' example.}\label{tab:yes-mainedges}
\begin{tabular}{c l c r}
\toprule
edge & occurrences & source values & encoded sum \\
\midrule
$M_1$ & $P_{1,1}$, $P_{2,1}$, $P_{3,1}$ & $3+13+20=36$ & {69{,}414{,}571{,}596} \\
$M_2$ & $P_{1,2}$, $P_{2,2}$, $P_{3,2}$ & $4+11+21=36$ & {69{,}414{,}571{,}596} \\
$M_3$ & $P_{1,3}$, $P_{2,3}$, $P_{3,3}$ & $5+9+22=36$ & {69{,}414{,}571{,}596} \\
$M_4$ & $P_{1,4}$, $P_{2,4}$, $P_{3,4}$ & $6+7+23=36$ & {69{,}414{,}571{,}596} \\
\bottomrule
\end{tabular}
\end{table}

The resulting unlabeled hyperedge multiset has \(48+4=52\) hyperedges.  Each port occurs twice in its corresponding local filler pattern and once in its unique main edge.  The private incidence ledger in every group is
\[
\begin{array}{ll@{\qquad\qquad}ll}
Q_1: & F_1\ (\text{two copies}),\ F_4 & Q_2: & F_1\ (\text{two copies}),\ F_5\\
Q_3: & F_2\ (\text{two copies}),\ F_6 & Q_4: & F_2\ (\text{two copies}),\ F_7\\
Q_5: & F_3\ (\text{two copies}),\ F_8 & Q_6: & F_3\ (\text{two copies}),\ F_9\\
Q_7: & F_4,\ F_5,\ F_6 & Q_8: & F_4,\ F_7,\ F_8\\
Q_9: & F_5,\ F_8,\ F_9 & Q_{10}: & F_6,\ F_7,\ F_9
\end{array}
\]

\normalsize
\baselineskip=20pt

Every one of the 52 underlying occurrences therefore has degree three.  The incidence graph is a 3-regular bipartite graph with 52 item vertices, 52 hyperedge vertices, and 156 incidence edges.  Decomposing its edges into three perfect matchings yields the explicit recoloring in Tables~\ref{tab:yes-recolor-filler} and~\ref{tab:yes-recolor-main}.  Each row contains one occurrence from each labeled output class and has encoded sum \(K^\star\).
\begingroup
\scriptsize
\renewcommand{\arraystretch}{0.78}
\setlength{\tabcolsep}{5pt}
\begin{longtable}{l c c c r}
\caption{The 48 recolored filler triples in one explicit feasible output matching for the ``yes'' instance.}\label{tab:yes-recolor-filler}\\
\toprule
unlabeled edge & $X^\star$ occurrence & $Y^\star$ occurrence & $Z^\star$ occurrence & sum \\
\midrule
\endfirsthead
\caption[]{The 48 recolored filler triples in one explicit feasible output matching for the ``yes'' instance. (continued)}\\
\toprule
unlabeled edge & $X^\star$ occurrence & $Y^\star$ occurrence & $Z^\star$ occurrence & sum \\
\midrule
\endhead
\midrule\multicolumn{5}{r}{\emph{continued on next page}}\\
\endfoot
\bottomrule
\endlastfoot
$F_{1}^{(a)}(1)$ & $P^{X}_{1,1}$ & $Q^{Y}_{1,1}$ & $Q^{Z}_{2,1}$ & {69{,}414{,}571{,}596} \\
$F_{1}^{(b)}(1)$ & $Q^{X}_{1,1}$ & $Q^{Y}_{2,1}$ & $P^{Z}_{1,1}$ & {69{,}414{,}571{,}596} \\
$F_{2}^{(a)}(1)$ & $Q^{X}_{4,1}$ & $P^{Y}_{2,1}$ & $Q^{Z}_{3,1}$ & {69{,}414{,}571{,}596} \\
$F_{2}^{(b)}(1)$ & $Q^{X}_{3,1}$ & $Q^{Y}_{4,1}$ & $P^{Z}_{2,1}$ & {69{,}414{,}571{,}596} \\
$F_{3}^{(a)}(1)$ & $Q^{X}_{6,1}$ & $P^{Y}_{3,1}$ & $Q^{Z}_{5,1}$ & {69{,}414{,}571{,}596} \\
$F_{3}^{(b)}(1)$ & $P^{X}_{3,1}$ & $Q^{Y}_{5,1}$ & $Q^{Z}_{6,1}$ & {69{,}414{,}571{,}596} \\
$F_{4}(1)$ & $Q^{X}_{7,1}$ & $Q^{Y}_{8,1}$ & $Q^{Z}_{1,1}$ & {69{,}414{,}571{,}596} \\
$F_{5}(1)$ & $Q^{X}_{2,1}$ & $Q^{Y}_{7,1}$ & $Q^{Z}_{9,1}$ & {69{,}414{,}571{,}596} \\
$F_{6}(1)$ & $Q^{X}_{10,1}$ & $Q^{Y}_{3,1}$ & $Q^{Z}_{7,1}$ & {69{,}414{,}571{,}596} \\
$F_{7}(1)$ & $Q^{X}_{8,1}$ & $Q^{Y}_{10,1}$ & $Q^{Z}_{4,1}$ & {69{,}414{,}571{,}596} \\
$F_{8}(1)$ & $Q^{X}_{5,1}$ & $Q^{Y}_{9,1}$ & $Q^{Z}_{8,1}$ & {69{,}414{,}571{,}596} \\
$F_{9}(1)$ & $Q^{X}_{9,1}$ & $Q^{Y}_{6,1}$ & $Q^{Z}_{10,1}$ & {69{,}414{,}571{,}596} \\
$F_{1}^{(a)}(2)$ & $Q^{X}_{2,2}$ & $P^{Y}_{1,2}$ & $Q^{Z}_{1,2}$ & {69{,}414{,}571{,}596} \\
$F_{1}^{(b)}(2)$ & $Q^{X}_{1,2}$ & $Q^{Y}_{2,2}$ & $P^{Z}_{1,2}$ & {69{,}414{,}571{,}596} \\
$F_{2}^{(a)}(2)$ & $P^{X}_{2,2}$ & $Q^{Y}_{4,2}$ & $Q^{Z}_{3,2}$ & {69{,}414{,}571{,}596} \\
$F_{2}^{(b)}(2)$ & $Q^{X}_{3,2}$ & $P^{Y}_{2,2}$ & $Q^{Z}_{4,2}$ & {69{,}414{,}571{,}596} \\
$F_{3}^{(a)}(2)$ & $P^{X}_{3,2}$ & $Q^{Y}_{5,2}$ & $Q^{Z}_{6,2}$ & {69{,}414{,}571{,}596} \\
$F_{3}^{(b)}(2)$ & $Q^{X}_{5,2}$ & $Q^{Y}_{6,2}$ & $P^{Z}_{3,2}$ & {69{,}414{,}571{,}596} \\
$F_{4}(2)$ & $Q^{X}_{7,2}$ & $Q^{Y}_{1,2}$ & $Q^{Z}_{8,2}$ & {69{,}414{,}571{,}596} \\
$F_{5}(2)$ & $Q^{X}_{9,2}$ & $Q^{Y}_{7,2}$ & $Q^{Z}_{2,2}$ & {69{,}414{,}571{,}596} \\
$F_{6}(2)$ & $Q^{X}_{10,2}$ & $Q^{Y}_{3,2}$ & $Q^{Z}_{7,2}$ & {69{,}414{,}571{,}596} \\
$F_{7}(2)$ & $Q^{X}_{4,2}$ & $Q^{Y}_{8,2}$ & $Q^{Z}_{10,2}$ & {69{,}414{,}571{,}596} \\
$F_{8}(2)$ & $Q^{X}_{8,2}$ & $Q^{Y}_{9,2}$ & $Q^{Z}_{5,2}$ & {69{,}414{,}571{,}596} \\
$F_{9}(2)$ & $Q^{X}_{6,2}$ & $Q^{Y}_{10,2}$ & $Q^{Z}_{9,2}$ & {69{,}414{,}571{,}596} \\
$F_{1}^{(a)}(3)$ & $Q^{X}_{1,3}$ & $Q^{Y}_{2,3}$ & $P^{Z}_{1,3}$ & {69{,}414{,}571{,}596} \\
$F_{1}^{(b)}(3)$ & $P^{X}_{1,3}$ & $Q^{Y}_{1,3}$ & $Q^{Z}_{2,3}$ & {69{,}414{,}571{,}596} \\
$F_{2}^{(a)}(3)$ & $P^{X}_{2,3}$ & $Q^{Y}_{4,3}$ & $Q^{Z}_{3,3}$ & {69{,}414{,}571{,}596} \\
$F_{2}^{(b)}(3)$ & $Q^{X}_{3,3}$ & $P^{Y}_{2,3}$ & $Q^{Z}_{4,3}$ & {69{,}414{,}571{,}596} \\
$F_{3}^{(a)}(3)$ & $Q^{X}_{5,3}$ & $P^{Y}_{3,3}$ & $Q^{Z}_{6,3}$ & {69{,}414{,}571{,}596} \\
$F_{3}^{(b)}(3)$ & $Q^{X}_{6,3}$ & $Q^{Y}_{5,3}$ & $P^{Z}_{3,3}$ & {69{,}414{,}571{,}596} \\
$F_{4}(3)$ & $Q^{X}_{7,3}$ & $Q^{Y}_{8,3}$ & $Q^{Z}_{1,3}$ & {69{,}414{,}571{,}596} \\
$F_{5}(3)$ & $Q^{X}_{2,3}$ & $Q^{Y}_{7,3}$ & $Q^{Z}_{9,3}$ & {69{,}414{,}571{,}596} \\
$F_{6}(3)$ & $Q^{X}_{10,3}$ & $Q^{Y}_{3,3}$ & $Q^{Z}_{7,3}$ & {69{,}414{,}571{,}596} \\
$F_{7}(3)$ & $Q^{X}_{4,3}$ & $Q^{Y}_{10,3}$ & $Q^{Z}_{8,3}$ & {69{,}414{,}571{,}596} \\
$F_{8}(3)$ & $Q^{X}_{8,3}$ & $Q^{Y}_{9,3}$ & $Q^{Z}_{5,3}$ & {69{,}414{,}571{,}596} \\
$F_{9}(3)$ & $Q^{X}_{9,3}$ & $Q^{Y}_{6,3}$ & $Q^{Z}_{10,3}$ & {69{,}414{,}571{,}596} \\
$F_{1}^{(a)}(4)$ & $P^{X}_{1,4}$ & $Q^{Y}_{1,4}$ & $Q^{Z}_{2,4}$ & {69{,}414{,}571{,}596} \\
$F_{1}^{(b)}(4)$ & $Q^{X}_{1,4}$ & $Q^{Y}_{2,4}$ & $P^{Z}_{1,4}$ & {69{,}414{,}571{,}596} \\
$F_{2}^{(a)}(4)$ & $Q^{X}_{3,4}$ & $P^{Y}_{2,4}$ & $Q^{Z}_{4,4}$ & {69{,}414{,}571{,}596} \\
$F_{2}^{(b)}(4)$ & $Q^{X}_{4,4}$ & $Q^{Y}_{3,4}$ & $P^{Z}_{2,4}$ & {69{,}414{,}571{,}596} \\
$F_{3}^{(a)}(4)$ & $P^{X}_{3,4}$ & $Q^{Y}_{6,4}$ & $Q^{Z}_{5,4}$ & {69{,}414{,}571{,}596} \\
$F_{3}^{(b)}(4)$ & $Q^{X}_{5,4}$ & $P^{Y}_{3,4}$ & $Q^{Z}_{6,4}$ & {69{,}414{,}571{,}596} \\
$F_{4}(4)$ & $Q^{X}_{7,4}$ & $Q^{Y}_{8,4}$ & $Q^{Z}_{1,4}$ & {69{,}414{,}571{,}596} \\
$F_{5}(4)$ & $Q^{X}_{2,4}$ & $Q^{Y}_{9,4}$ & $Q^{Z}_{7,4}$ & {69{,}414{,}571{,}596} \\
$F_{6}(4)$ & $Q^{X}_{10,4}$ & $Q^{Y}_{7,4}$ & $Q^{Z}_{3,4}$ & {69{,}414{,}571{,}596} \\
$F_{7}(4)$ & $Q^{X}_{8,4}$ & $Q^{Y}_{4,4}$ & $Q^{Z}_{10,4}$ & {69{,}414{,}571{,}596} \\
$F_{8}(4)$ & $Q^{X}_{9,4}$ & $Q^{Y}_{5,4}$ & $Q^{Z}_{8,4}$ & {69{,}414{,}571{,}596} \\
$F_{9}(4)$ & $Q^{X}_{6,4}$ & $Q^{Y}_{10,4}$ & $Q^{Z}_{9,4}$ & {69{,}414{,}571{,}596} \\
\end{longtable}
\endgroup
\begin{table}[htbp]
\centering
\small
\caption{The 4 recolored main triples in one explicit feasible output matching for the ``yes'' instance.}\label{tab:yes-recolor-main}
\begin{tabular}{l c c c r}
\toprule
unlabeled edge & $X^\star$ occurrence & $Y^\star$ occurrence & $Z^\star$ occurrence & sum \\
\midrule
$M1=F_0(1,1,1)$ & $P^{X}_{2,1}$ & $P^{Y}_{1,1}$ & $P^{Z}_{3,1}$ & {69{,}414{,}571{,}596} \\
$M2=F_0(2,2,2)$ & $P^{X}_{1,2}$ & $P^{Y}_{3,2}$ & $P^{Z}_{2,2}$ & {69{,}414{,}571{,}596} \\
$M3=F_0(3,3,3)$ & $P^{X}_{3,3}$ & $P^{Y}_{1,3}$ & $P^{Z}_{2,3}$ & {69{,}414{,}571{,}596} \\
$M4=F_0(4,4,4)$ & $P^{X}_{2,4}$ & $P^{Y}_{1,4}$ & $P^{Z}_{3,4}$ & {69{,}414{,}571{,}596} \\
\bottomrule
\end{tabular}
\end{table}

Reading the three colored columns across the two tables gives a complete occurrence audit: for every \(u\in S\), the copies \(u^X,u^Y,u^Z\) each occur exactly once.  Thus the tables use all 156 labeled output occurrences in 52 legal triples.  This completes the forward implication concretely.

\subsubsection{The ``only if'' direction in the example}
Begin instead with any feasible matching of the constructed integer instance and forget the outer labels.  The no-carry bound decodes every integer equality into the four coordinate equalities.  The role coordinate restricts every edge to one of \(F_0,\ldots,F_9\).

For \(F_1,F_2,F_3\), the group coordinate first forces the two private indices to coincide.  The source-coordinate equations then identify the attached port value.  For example, if the private items of an \(F_1\)-edge come from group \(i\), their source sum is
\[
\eta_{Q_1,i}+\eta_{Q_2,i}=108-3a_i.
\]
Attaching \(P_{1,j}\) reaches 108 exactly when \(a_j=a_i\).  Because the four \(A\)-values are distinct, this gives \(j=i\).  The same reasoning applies to \(F_2\) using the distinct \(B\)-values and to \(F_3\) using the distinct \(C\)-values.

For \(F_4,\ldots,F_9\), the group and group-square equations force all three private indices to agree.  For example, coefficient pattern \((2,1,-3)\) gives
\[
2i+j-3k=0,
\qquad
2i^2+j^2-3k^2=0,
\]
which implies
\[
2(i-k)^2+(j-k)^2=0,
\]
and hence \(i=j=k\).

Fix a group and let \(x_r\) be the multiplicity of \(F_r\).  The ten private-degree equations are
\begin{align*}
x_1+x_4&=3,&x_1+x_5&=3,\\
x_2+x_6&=3,&x_2+x_7&=3,\\
x_3+x_8&=3,&x_3+x_9&=3,\\
x_4+x_5+x_6&=3,&x_4+x_7+x_8&=3,\\
x_5+x_8+x_9&=3,&x_6+x_7+x_9&=3.
\end{align*}
Their unique nonnegative integral solution is
\[
x_1=x_2=x_3=2,
\qquad
x_4=x_5=x_6=x_7=x_8=x_9=1.
\]
Thus exactly two incidences of every port are consumed locally and the third incidence must lie in a main \(F_0\)-edge.  There are exactly four main edges, they use every \(P_1\)-port, every \(P_2\)-port, and every \(P_3\)-port once, and their source-coordinate equation is
\[
3a_i+3b_j+3c_k=108
\quad\Longleftrightarrow\quad
a_i+b_j+c_k=36.
\]
Forgetting the gadget roles therefore recovers a feasible four-triple matching of the source instance.  Applied to the explicit matching above, this reverse process recovers the four diagonal source triples shown in Table~\ref{tab:yes-recolor-main}.

\subsubsection{Complete numerical audit}
\begin{center}
\begin{tabular}{l r}
\toprule
object & count\\
\midrule
source occurrences & 12\\
groups \(S_i\) & 4\\
underlying occurrences in \(S\) & 52\\
occurrences in each of \(X^\star,Y^\star,Z^\star\) & 52\\
total labeled output occurrences & 156\\
filler hyperedge occurrences & 48\\
main hyperedge occurrences & 4\\
output matching triples & 52\\
\bottomrule
\end{tabular}
\end{center}

\subsection{A source-to-output ``no'' instance}
\subsubsection{The classical source instance and its forced target}
Now take
\[
A=\mset{3,4,5,6},\qquad
B=\mset{13,12,8,7},\qquad
C=\mset{20,21,22,23}.
\]
The class sums remain 18, 40, and 86, so the common target is again forced:
\[
4t=18+40+86=144,
\qquad t=36.
\]
The complete list of ordered index triples satisfying \(a_i+b_j+c_k=36\) is given in Table~\ref{tab:no-possible-main}.
\begin{table}[htbp]
\centering
\small
\caption{All index triples whose source values sum to 36 in the ``no'' source instance.}\label{tab:no-possible-main}
\begin{tabular}{c c c c c c}
\toprule
$i$ & $j$ & $k$ & $a_i$ & $b_j$ & $c_k$ \\
\midrule
1 & 1 & 1 & 3 & 13 & 20 \\
1 & 2 & 2 & 3 & 12 & 21 \\
2 & 2 & 1 & 4 & 12 & 20 \\
3 & 3 & 4 & 5 & 8 & 23 \\
4 & 3 & 3 & 6 & 8 & 22 \\
4 & 4 & 4 & 6 & 7 & 23 \\
\bottomrule
\end{tabular}
\end{table}

This source instance has no feasible matching.  The occurrence \(a_2=4\) appears in only the triple \((2,2,1)\), and \(a_3=5\) appears in only the triple \((3,3,4)\).  Any perfect matching would therefore have to select both of these triples.  But after selecting \((2,2,1)\), the two possible triples containing \(a_1=3\) are both blocked: \((1,1,1)\) reuses \(c_1\), while \((1,2,2)\) reuses \(b_2\).  Hence \(a_1\) cannot be matched, and no four-triple source matching exists.

\subsubsection{The occurrence sets and identical output classes}
The occurrence construction is unchanged.  There are four groups \(S_i\), 52 underlying occurrences in \(S\), 52 occurrences in each output class, and 156 labeled output occurrences in total.  Each output class is a complete copy of the same encoded set \(S\), so the output weight multisets are identical even though the source classes are asymmetric.

\subsubsection{All specialized source-coordinate values}
Again \(3t=108\).  Only the middle-class values at indices 2 and 3 differ from the ``yes'' instance, and the corresponding private source coordinates change as shown below.
\begin{table}[htbp]
\centering
\small
\caption{Port source-coordinate values for the ``no'' instance.}\label{tab:no-ports}
\begin{tabular}{c r r r r r r}
\toprule
$i$ & $a_i$ & $b_i$ & $c_i$ & $\eta_{P_1,i}$ & $\eta_{P_2,i}$ & $\eta_{P_3,i}$ \\
\midrule
1 & 3 & 13 & 20 & 9 & 39 & 60 \\
2 & 4 & 12 & 21 & 12 & 36 & 63 \\
3 & 5 & 8 & 22 & 15 & 24 & 66 \\
4 & 6 & 7 & 23 & 18 & 21 & 69 \\
\bottomrule
\end{tabular}
\end{table}
\begin{table}[h!]
\centering
\scriptsize
\setlength{\tabcolsep}{3.5pt}
\caption{All private source-coordinate values for the ``no'' instance.}\label{tab:no-private}
\begin{tabular}{c c r r r r r r r r r r}
\toprule
$i$ & $(a_i,b_i,c_i)$ & $Q_1$ & $Q_2$ & $Q_3$ & $Q_4$ & $Q_5$ & $Q_6$ & $Q_7$ & $Q_8$ & $Q_9$ & $Q_{10}$ \\
\midrule
1 & $(3,13,20)$ & -39 & 138 & 135 & -66 & -63 & 111 & -27 & 174 & -3 & 0 \\
2 & $(4,12,21)$ & -36 & 132 & 133 & -61 & -62 & 107 & -25 & 169 & 1 & 0 \\
3 & $(5,8,22)$ & -24 & 117 & 128 & -44 & -55 & 97 & -20 & 152 & 11 & 0 \\
4 & $(6,7,23)$ & -21 & 111 & 126 & -39 & -54 & 93 & -18 & 147 & 15 & 0 \\
\bottomrule
\end{tabular}
\end{table}
For example, in group 2, where \((a_2,b_2,c_2)=(4,12,21)\),
\[
\eta_{Q_2,2}=108-3(4)+3(12)=132,
\quad
\eta_{Q_4,2}=2(4)-4(12)-21=-61,
\]
and
\[
\eta_{Q_8,2}=108-2(4)+4(12)+21=169.
\]
As in the ``yes'' instance, these values cancel within every intended local filler pattern.

\subsubsection{Shifts, base, target, and all 52 encoded weights}
The extreme coordinates are unchanged: the minimum source coordinate is \(-66\), the minimum group coordinate is \(-12\), and the minimum group-square coordinate is \(-48\).  Thus
\[
(D_0,D_1,D_2,D_3)=(0,66,12,48),
\qquad
\bar K=(300,306,36,144).
\]
The largest shifted digits are again \((261,240,20,80)\), so the same smallest admissible base works:
\[
L=784.
\]
Consequently, the encoded target is also unchanged:
\[
K^\star=69{,}414{,}571{,}596.
\]
The complete item catalogue follows.
\begingroup
\scriptsize
\renewcommand{\arraystretch}{0.78}
\setlength{\tabcolsep}{4pt}
\begin{longtable}{c c c c c r}
\caption{All 52 vector weights, shifted digit vectors, and encoded positive integer weights for the ``no'' instance.}\label{tab:no-weights}\\
\toprule
$i$ & $(a_i,b_i,c_i)$ & item $u$ & $\omega(u)$ & $\bar\omega(u)$ & $W(u)$ \\
\midrule
\endfirsthead
\caption[]{All 52 vector weights, shifted digit vectors, and encoded positive integer weights for the ``no'' instance. (continued)}\\
\toprule
$i$ & $(a_i,b_i,c_i)$ & item $u$ & $\omega(u)$ & $\bar\omega(u)$ & $W(u)$ \\
\midrule
\endhead
\midrule\multicolumn{6}{r}{\emph{continued on next page}}\\
\endfoot
\bottomrule
\endlastfoot
1 & $(3,13,20)$ & $P_{1,1}$ & $(31,9,0,0)$ & $(31,75,12,48)$ & {23{,}138{,}169{,}295} \\
1 & $(3,13,20)$ & $P_{2,1}$ & $(8,39,0,0)$ & $(8,105,12,48)$ & {23{,}138{,}192{,}792} \\
1 & $(3,13,20)$ & $P_{3,1}$ & $(261,60,0,0)$ & $(261,126,12,48)$ & {23{,}138{,}209{,}509} \\
1 & $(3,13,20)$ & $Q_{1,1}$ & $(184,-39,2,2)$ & $(184,27,14,50)$ & {24{,}103{,}141{,}736} \\
1 & $(3,13,20)$ & $Q_{2,1}$ & $(85,138,-2,-2)$ & $(85,204,10,46)$ & {22{,}173{,}260{,}565} \\
1 & $(3,13,20)$ & $Q_{3,1}$ & $(177,135,-2,-2)$ & $(177,201,10,46)$ & {22{,}173{,}258{,}305} \\
1 & $(3,13,20)$ & $Q_{4,1}$ & $(115,-66,2,2)$ & $(115,0,14,50)$ & {24{,}103{,}120{,}499} \\
1 & $(3,13,20)$ & $Q_{5,1}$ & $(23,-63,2,2)$ & $(23,3,14,50)$ & {24{,}103{,}122{,}759} \\
1 & $(3,13,20)$ & $Q_{6,1}$ & $(16,111,-2,-2)$ & $(16,177,10,46)$ & {22{,}173{,}239{,}328} \\
1 & $(3,13,20)$ & $Q_{7,1}$ & $(27,-27,1,1)$ & $(27,39,13,49)$ & {23{,}620{,}646{,}027} \\
1 & $(3,13,20)$ & $Q_{8,1}$ & $(89,174,-3,-3)$ & $(89,240,9,45)$ & {21{,}690{,}783{,}833} \\
1 & $(3,13,20)$ & $Q_{9,1}$ & $(188,-3,1,1)$ & $(188,63,13,49)$ & {23{,}620{,}665{,}004} \\
1 & $(3,13,20)$ & $Q_{10,1}$ & $(96,0,1,1)$ & $(96,66,13,49)$ & {23{,}620{,}667{,}264} \\
2 & $(4,12,21)$ & $P_{1,2}$ & $(31,12,0,0)$ & $(31,78,12,48)$ & {23{,}138{,}171{,}647} \\
2 & $(4,12,21)$ & $P_{2,2}$ & $(8,36,0,0)$ & $(8,102,12,48)$ & {23{,}138{,}190{,}440} \\
2 & $(4,12,21)$ & $P_{3,2}$ & $(261,63,0,0)$ & $(261,129,12,48)$ & {23{,}138{,}211{,}861} \\
2 & $(4,12,21)$ & $Q_{1,2}$ & $(184,-36,4,8)$ & $(184,30,16,56)$ & {26{,}995{,}715{,}224} \\
2 & $(4,12,21)$ & $Q_{2,2}$ & $(85,132,-4,-8)$ & $(85,198,8,40)$ & {19{,}280{,}684{,}725} \\
2 & $(4,12,21)$ & $Q_{3,2}$ & $(177,133,-4,-8)$ & $(177,199,8,40)$ & {19{,}280{,}685{,}601} \\
2 & $(4,12,21)$ & $Q_{4,2}$ & $(115,-61,4,8)$ & $(115,5,16,56)$ & {26{,}995{,}695{,}555} \\
2 & $(4,12,21)$ & $Q_{5,2}$ & $(23,-62,4,8)$ & $(23,4,16,56)$ & {26{,}995{,}694{,}679} \\
2 & $(4,12,21)$ & $Q_{6,2}$ & $(16,107,-4,-8)$ & $(16,173,8,40)$ & {19{,}280{,}665{,}056} \\
2 & $(4,12,21)$ & $Q_{7,2}$ & $(27,-25,2,4)$ & $(27,41,14,52)$ & {25{,}066{,}933{,}163} \\
2 & $(4,12,21)$ & $Q_{8,2}$ & $(89,169,-6,-12)$ & $(89,235,6,36)$ & {17{,}351{,}923{,}209} \\
2 & $(4,12,21)$ & $Q_{9,2}$ & $(188,1,2,4)$ & $(188,67,14,52)$ & {25{,}066{,}953{,}708} \\
2 & $(4,12,21)$ & $Q_{10,2}$ & $(96,0,2,4)$ & $(96,66,14,52)$ & {25{,}066{,}952{,}832} \\
3 & $(5,8,22)$ & $P_{1,3}$ & $(31,15,0,0)$ & $(31,81,12,48)$ & {23{,}138{,}173{,}999} \\
3 & $(5,8,22)$ & $P_{2,3}$ & $(8,24,0,0)$ & $(8,90,12,48)$ & {23{,}138{,}181{,}032} \\
3 & $(5,8,22)$ & $P_{3,3}$ & $(261,66,0,0)$ & $(261,132,12,48)$ & {23{,}138{,}214{,}213} \\
3 & $(5,8,22)$ & $Q_{1,3}$ & $(184,-24,6,18)$ & $(184,42,18,66)$ & {31{,}815{,}856{,}984} \\
3 & $(5,8,22)$ & $Q_{2,3}$ & $(85,117,-6,-18)$ & $(85,183,6,30)$ & {14{,}460{,}540{,}613} \\
3 & $(5,8,22)$ & $Q_{3,3}$ & $(177,128,-6,-18)$ & $(177,194,6,30)$ & {14{,}460{,}549{,}329} \\
3 & $(5,8,22)$ & $Q_{4,3}$ & $(115,-44,6,18)$ & $(115,22,18,66)$ & {31{,}815{,}841{,}235} \\
3 & $(5,8,22)$ & $Q_{5,3}$ & $(23,-55,6,18)$ & $(23,11,18,66)$ & {31{,}815{,}832{,}519} \\
3 & $(5,8,22)$ & $Q_{6,3}$ & $(16,97,-6,-18)$ & $(16,163,6,30)$ & {14{,}460{,}524{,}864} \\
3 & $(5,8,22)$ & $Q_{7,3}$ & $(27,-20,3,9)$ & $(27,46,15,57)$ & {27{,}477{,}003{,}259} \\
3 & $(5,8,22)$ & $Q_{8,3}$ & $(89,152,-9,-27)$ & $(89,218,3,21)$ & {10{,}121{,}711{,}353} \\
3 & $(5,8,22)$ & $Q_{9,3}$ & $(188,11,3,9)$ & $(188,77,15,57)$ & {27{,}477{,}027{,}724} \\
3 & $(5,8,22)$ & $Q_{10,3}$ & $(96,0,3,9)$ & $(96,66,15,57)$ & {27{,}477{,}019{,}008} \\
4 & $(6,7,23)$ & $P_{1,4}$ & $(31,18,0,0)$ & $(31,84,12,48)$ & {23{,}138{,}176{,}351} \\
4 & $(6,7,23)$ & $P_{2,4}$ & $(8,21,0,0)$ & $(8,87,12,48)$ & {23{,}138{,}178{,}680} \\
4 & $(6,7,23)$ & $P_{3,4}$ & $(261,69,0,0)$ & $(261,135,12,48)$ & {23{,}138{,}216{,}565} \\
4 & $(6,7,23)$ & $Q_{1,4}$ & $(184,-21,8,32)$ & $(184,45,20,80)$ & {38{,}563{,}552{,}904} \\
4 & $(6,7,23)$ & $Q_{2,4}$ & $(85,111,-8,-32)$ & $(85,177,4,16)$ & {7{,}712{,}842{,}341} \\
4 & $(6,7,23)$ & $Q_{3,4}$ & $(177,126,-8,-32)$ & $(177,192,4,16)$ & {7{,}712{,}854{,}193} \\
4 & $(6,7,23)$ & $Q_{4,4}$ & $(115,-39,8,32)$ & $(115,27,20,80)$ & {38{,}563{,}538{,}723} \\
4 & $(6,7,23)$ & $Q_{5,4}$ & $(23,-54,8,32)$ & $(23,12,20,80)$ & {38{,}563{,}526{,}871} \\
4 & $(6,7,23)$ & $Q_{6,4}$ & $(16,93,-8,-32)$ & $(16,159,4,16)$ & {7{,}712{,}828{,}160} \\
4 & $(6,7,23)$ & $Q_{7,4}$ & $(27,-18,4,16)$ & $(27,48,16,64)$ & {30{,}850{,}851{,}611} \\
4 & $(6,7,23)$ & $Q_{8,4}$ & $(89,147,-12,-48)$ & $(89,213,0,0)$ & {167{,}081} \\
4 & $(6,7,23)$ & $Q_{9,4}$ & $(188,15,4,16)$ & $(188,81,16,64)$ & {30{,}850{,}877{,}644} \\
4 & $(6,7,23)$ & $Q_{10,4}$ & $(96,0,4,16)$ & $(96,66,16,64)$ & {30{,}850{,}865{,}792} \\
\end{longtable}
\endgroup
The no-carry audit is identical: every sum of three item digits is at most 783, strictly below the base 784.  Hence an integer triple reaches \(K^\star\) if and only if its four unshifted coordinates sum to \((300,108,0,0)\).

\subsubsection{Local filler identities}
The local identities depend algebraically on \((a_i,b_i,c_i)\) but not on any source matching.  They therefore remain valid in every group of the ``no'' instance.
\begin{table}[h!]
\centering
\scriptsize
\setlength{\tabcolsep}{5pt}
\caption{Local filler identities for target source coordinate $3t=108$.}\label{tab:no-identities}
\begin{tabularx}{\textwidth}{@{}c c >{\raggedright\arraybackslash}X c@{}}
\toprule
pattern & role-code calculation & source-coordinate calculation & $h$-sum \\
\midrule
$F_1$ & $31+184+85=300$ & $3a-3b+(108-3a+3b)=108$ & $0+2-2=0$ \\
$F_2$ & $8+177+115=300$ & $3b+(108-2a+b+c)+(2a-4b-c)=108$ & $0-2+2=0$ \\
$F_3$ & $261+23+16=300$ & $3c+(a-2b-2c)+(108-a+2b-c)=108$ & $0+2-2=0$ \\
$F_4$ & $184+27+89=300$ & $-3b+(2a-b-c)+(108-2a+4b+c)=108$ & $2+1-3=0$ \\
$F_5$ & $85+27+188=300$ & $(108-3a+3b)+(2a-b-c)+(a-2b+c)=108$ & $-2+1+1=0$ \\
$F_6$ & $177+27+96=300$ & $(108-2a+b+c)+(2a-b-c)+0=108$ & $-2+1+1=0$ \\
$F_7$ & $115+89+96=300$ & $(2a-4b-c)+(108-2a+4b+c)+0=108$ & $2-3+1=0$ \\
$F_8$ & $23+89+188=300$ & $(a-2b-2c)+(108-2a+4b+c)+(a-2b+c)=108$ & $2-3+1=0$ \\
$F_9$ & $16+188+96=300$ & $(108-a+2b-c)+(a-2b+c)+0=108$ & $-2+1+1=0$ \\
\bottomrule
\end{tabularx}
\end{table}
For every group \(i\), the forward filler multiset would again be
\[
2F_1(i)\uplus2F_2(i)\uplus2F_3(i)\uplus F_4(i)\uplus\cdots\uplus F_9(i).
\]
All 48 local filler hyperedges are valid target-sum edges.  They give degree two to every port and degree three to every private item.  Therefore the local part of the construction works perfectly; any failure must occur in the four main edges that encode the source matching.

\subsubsection{Why the ``if'' construction cannot be completed}
A main edge has the form
\[
F_0(i,j,k)=(P_{1,i},P_{2,j},P_{3,k})
\]
and reaches the target exactly when
\[
3a_i+3b_j+3c_k=108
\quad\Longleftrightarrow\quad
a_i+b_j+c_k=36.
\]
Thus the only possible main edges are precisely the six index triples in Table~\ref{tab:no-possible-main}.  The port \(P_{1,2}\) can occur in only \(F_0(2,2,1)\), and \(P_{1,3}\) can occur in only \(F_0(3,3,4)\).  Selecting those two forced edges consumes \(P_{2,2}\) and \(P_{3,1}\).  The port \(P_{1,1}\) can then use neither of its two possible main edges: \(F_0(1,1,1)\) would reuse \(P_{3,1}\), while \(F_0(1,2,2)\) would reuse \(P_{2,2}\).  Hence the required one-main-incidence-per-port condition cannot be completed.

This is the precise failure of the forward recipe.  It does not occur in the filler layer; it is exactly the original source-matching obstruction preserved by the main source coordinate.

\subsubsection{The ``only if'' direction proves impossibility of an output matching}
Assume for contradiction that the constructed SN3DM instance has a feasible output matching.  Forgetting the outer labels yields a 3-regular target-sum triple system on \(S\).  The no-carry property decodes every integer equality coordinatewise, and the role coordinate restricts every edge to one of \(F_0,\ldots,F_9\).

For \(F_1,F_2,F_3\), the group coordinate synchronizes the two private indices, and the source coordinate forces the attached port to represent the same source value.  Since the values within each of \(A,B,C\) are distinct, the port index is the group index.  For \(F_4,\ldots,F_9\), the group and group-square coordinates force all three private indices to coincide.

The ten private-degree equations have the same unique solution as in the ``yes'' instance:
\[
x_1=x_2=x_3=2,
\qquad
x_4=x_5=x_6=x_7=x_8=x_9=1.
\]
Therefore every port has exactly two filler incidences and exactly one remaining incidence in a main \(F_0\)-edge.  The main edges must use every \(P_1\)-port, every \(P_2\)-port, and every \(P_3\)-port exactly once, and each must satisfy \(a_i+b_j+c_k=36\).  They would consequently form a feasible source matching.  The forced-edge argument above proves that no such source matching exists.  This contradiction shows that the symmetric output instance has no feasible matching.

\subsubsection{Complete numerical audit}
\begin{center}
\begin{tabular}{l r}
\toprule
object & count\\
\midrule
source occurrences & 12\\
groups \(S_i\) & 4\\
underlying occurrences in \(S\) & 52\\
occurrences in each of \(X^\star,Y^\star,Z^\star\) & 52\\
total labeled output occurrences & 156\\
valid local filler hyperedge occurrences & 48\\
required main hyperedge occurrences & 4\\
numerically possible main-edge types & 6\\
perfectly disjoint four-edge main system & 0\\
complete feasible output matchings & 0\\
\bottomrule
\end{tabular}
\end{center}

\section{Lessons learned from Part I}

When every triple shares one common target as in problem SN3DM, bookkeeping cannot be delegated to distinct target labels. It must be internalized in the constructed instance. Here, role filters determine which types may meet, locality coordinates prevent cross-group splicing, degree equations reserve exactly one main incidence per port, and recoloring restores the labeled-class requirement. The general lesson is that each gadget layer should have one stated logical responsibility and that the reverse proof should compose those layers in a visible dependency chain.

\section{Conclusion}

\subsection*{Summary of results}
Part~I establishes that Symmetric Numerical Three-Dimensional Matching (\SNthreeDM) is \NP-complete when its numerical data are encoded in unary and, consequently, that classical Numerical Three-Dimensional Matching remains strongly \NP-complete when its three disjoint labeled classes have identical weight multisets.  The reduction starts from an arbitrary unary \NthreeDM instance; the three source classes need not be identical.  Their asymmetric roles are represented internally by the port types $P_1,P_2,P_3$ in one common occurrence set, of which the output contains three labeled copies.

The proof temporarily forgets those outer labels and works with a $3$-regular target-sum triple system.  Role, source, group, and group-square coordinates admit only the intended main and filler patterns, preserve the source target equation, and prevent cross-group splicing.  The private-item degree equations force the filler multiplicities, use two incidences of every port, and leave exactly one incidence of every port for a main edge; the main edges then recover a source matching.  Bipartite edge coloring restores the three output-class labels, and a no-carry mixed-radix encoding converts the coordinate construction into positive unary-polynomial integers.  The worked ``yes'' and ``no'' instances illustrate both directions of this equivalence.

The exact theorem is also the structural foundation for Part~II.  It identifies the perfect filler multiplicities and the one-main-incidence-per-port invariant.  Approximation hardness requires those equalities to remain informative after some triples are omitted.  Part~II therefore treats the reverse proof here as its zero-defect case and replaces exact degree equations by quantitative deficit equations.

\subsection*{Modeling and computational implications}
The Part~I result shows that neither a common target nor identical marginal weight distributions guarantee an easy exact-matching problem.  Three facilities, teams, shifts, or resource pools may contain the same multiset of capacities, durations, or requirements, and every selected triple may be required to attain the same total, yet coordinating one occurrence from each labeled pool while using every occurrence exactly once remains strongly \NP-hard.  The difficulty lies in the joint combinatorial assignment, not in visible differences among the three numerical lists.

For structured applications, the proof suggests exposing assignment roles explicitly, writing degree or incidence constraints, strengthening locality restrictions, and adding symmetry-breaking constraints among numerically identical occurrences.  Mixed-integer programming, constraint programming, exact-cover methods, and hypergraph formulations should therefore be assessed partly by how well they control equivalent copies and propagate the filler and main-edge structure.

\subsection*{Future research possibilities}
Natural exact-algorithm questions remain for bounded numbers of distinct weights, bounded weight multiplicities, restricted numerical ranges, interval or ordering structure, restricted triple eligibility, and bounded-width compatibility hypergraphs.  It would also be useful to simplify the reduction by decreasing the number of private items, filler patterns, role codes, or logical coordinates, or to determine whether the group-square coordinate can be replaced by a different locality mechanism.

The copy-forgetting and recoloring lemma suggests extensions to more than three identical classes and to other exact-sum problems whose labels can be restored from a regular incidence structure.  Parameterized and exact algorithms may be effective when only a small number of occurrences depart from the canonical filler scaffold.  Part~II settles the general no-PTAS question under the stated gap convention, but leaves open the approximation behavior of such restricted or parameterized families.

\clearpage
\part{Inapproximability of Maximum Symmetric Numerical Three-Dimensional Matching}\label{part:gap}
\setcounter{theorem}{0}
\renewcommand{\thesection}{\arabic{section}}
\renewcommand{\theHsection}{apx.sec.\arabic{section}}

\noindent\hyperlink{maincontents}{\small\textit{Return to main contents}}\par\medskip
\section*{Proof navigation for Part II}
\addcontentsline{toc}{section}{Proof navigation for Part II}

\textbf{Forward implication.}
Begin with a perfect matching in the bounded three-dimensional source instance.  The first numerical encoding replaces every source edge by one actual or dummy target quadruple.  The second compiler splits each target quadruple into two numerical triples and uses filler triples for unused pair indices.  The symmetric construction then inserts the twelve local filler triples around each source group and one main triple for each selected source triple.  Copy recoloring converts the resulting degree-at-most-three target-sum hypergraph into legal triples among the three labeled, weight-identical output classes.  A useful way to follow the forward direction is to check, at each stage, that a perfect source solution reaches the natural upper bound of the destination problem.

\textbf{Quantitative reverse implication.}
Begin with a symmetric matching that omits $d$ triples from perfection.  Forgetting the outer class labels creates exactly $3d$ missing incidences.  The role and locality coordinates still classify every surviving triple, while the private-item degree equations now contain explicit deficit terms.  Those equations bound the departure of the filler multiplicities from their perfect values.  Any remaining overuse of repeated source values is repaired by deleting only linearly many main triples.  The resulting inequality
\[
  \OPT_{\mathrm N}(I)\ge n-21d
\]
is the bridge from the exact proof in Part~I to a constant approximation gap.

\textbf{What could otherwise go wrong.}
Strong \NP-hardness alone distinguishes $M$ from $M-1$ and therefore supplies only a vanishing relative gap.  A nonquantitative reverse proof would likewise say nothing about solutions that are nearly perfect.  At the construction level, unintended role triples, cross-group splicing, arithmetic carry, or uncontrolled reuse of repeated values could each absorb many defects.  The proof assigns a separate mechanism to each danger: a role filter, two locality moments, a no-carry encoding, deficit equations, and a final overuse-repair step.  Readers interested mainly in the new idea may read \cref{apx:sec:proof-overview,apx:sec:stability,apx:sec:gap-transfer} first and then use the worked examples in \cref{apx:sec:worked-examples} as a numerical guide to the full construction.

\section{Introduction}

Part~I proves that Numerical Three-Dimensional Matching remains strongly \NP-complete when the three disjoint labeled classes have identical weight multisets and all numerical data are encoded in unary.  Part~II studies the natural maximum-cardinality version and asks whether that exact hardness can be strengthened to a constant approximation gap.

That decision result does not by itself imply any fixed-ratio inapproximability result.  Let $M$ be the number of items in each class, and let the optimization objective be the maximum number of pairwise item-disjoint target-sum triples.  The decision reduction separates only
\[
  \OPT=M
  \qquad\text{from}\qquad
  \OPT\le M-1.
\]
The ratio $(M-1)/M$ tends to one.  The existence of such a reduction is therefore compatible with the existence of a polynomial-time approximation scheme.

The missing ingredient is quantitative stability of the reverse implication.  In the perfect reverse proof, every underlying occurrence has incidence degree exactly three, and ten local degree equations force the intended filler multiplicities.  In a partial solution, some incidences are absent.  We replace the exact degree equations by deficit equations and prove that omitting $d$ symmetric triples can destroy at most $21d$ triples in the recovered source matching.

A second issue is the source problem.  Ordinary bounded three-dimensional matching has a perfect-completeness gap, but the symmetric construction starts from numerical three-dimensional matching.  We therefore include an explicit gap-preserving numerical transformation.  The transformation passes through numerical four-dimensional matching and keeps all numbers polynomially bounded.

The theorem below is the precise result established in this part.  The constant is intentionally coarse; only its positivity is needed.

\begin{theorem}[Unary perfect-completeness gap for \MaxSNthreeDM]\label{apx:thm:main-gap}
There exists a constant $\varepsilon_{\mathrm{SN}}>0$ such that the following promise problem is \NP-hard.  The input is a \MaxSNthreeDM instance with $M$ items in each of its three disjoint labeled classes, identical weight multisets in the three classes, positive integer weights, and one positive integer target.  Distinguish between
\begin{align*}
  \textnormal{\rmfamily(YES)}&\qquad \OPT_{\mathrm{SN}}=M,\\
  \textnormal{\rmfamily(NO)}&\qquad \OPT_{\mathrm{SN}}\le(1-\varepsilon_{\mathrm{SN}})M.
\end{align*}
The hardness holds when every numerical value is bounded by a polynomial in the combinatorial input size and hence when all numerical data are encoded in unary.
\end{theorem}

\begin{corollary}[Fixed-ratio inapproximability]\label{apx:cor:no-ptas}
Unless $\mathrm P=\mathrm{NP}$, \MaxSNthreeDM has no polynomial-time approximation scheme.  In particular, there is a constant $\alpha>1$ for which no polynomial-time $\alpha$-approximation exists unless $\mathrm P=\mathrm{NP}$.
\end{corollary}

The unambiguous formal conclusion of this part is \cref{apx:thm:main-gap} and its no-PTAS consequence.  The displayed transformation is not claimed here to satisfy the two affine optimum-error inequalities of an L-reduction.  Part~III supplies a separate standard L-reduction after changing the construction by one-live-port separation.  The distinction is substantive rather than terminological and is revisited in \cref{apx:rem:reduction-notion}.

Part~II is a quantitative extension of Part~I.  It repeats the definitions and the portion of the symmetric construction needed for approximation analysis while retaining the same port, private-item, role-code, locality, copy-forgetting, and recoloring vocabulary.  The exact reverse proof in Part~I is the zero-defect case of the argument developed here.

\paragraph{Suggested reading paths.}
For the conceptual result, read the main theorem, \cref{apx:sec:proof-overview}, the defect-stability lemma, and the gap-transfer proof.  For construction verification, read the two numerical compilers and then the complementary-pair certificate in Part~I.  For a concrete first encounter, read the one-group ``yes'' and ``no'' instances before returning to the general deficit equations.  The examples deliberately separate the fixed twelve-triple filler scaffold from the one main triple that carries the source decision.

\section{Optimization problems and gap terminology}\label{apx:sec:optimization}

Repeated numerical values are allowed throughout.  An item is an occurrence, not merely its numerical value.

\begin{definition}[\MaxNthreeDM]\label{apx:def:max-n3dm}
An instance is a tuple $I=(A,B,C,w,t)$, where $A,B,C$ are pairwise disjoint sets of cardinality $n$, $w:A\uplusm B\uplusm C\to\N_{>0}$, and $t\in\N_{>0}$.  A feasible solution is a collection of pairwise item-disjoint triples $(a,b,c)\in A\times B\times C$ satisfying
\[
 w(a)+w(b)+w(c)=t.
\]
The objective is to maximize the number of selected triples.  Its optimum is denoted $\OPT_{\mathrm N}(I)$.
\end{definition}
 
\begin{definition}[\MaxNfourDM]\label{apx:def:max-n4dm}
An instance of \MaxNfourDM{} (Maximum Numerical Four-Dimensional
Matching) consists of four pairwise disjoint item sets
$\mathcal A,\mathcal B,\mathcal C,\mathcal D$ of common cardinality $p$,
a positive integer weight for each item, and a positive integer target
$T$.  A feasible solution, called a \emph{matching}, is a collection of
pairwise item-disjoint quadruples in
$\mathcal A\times\mathcal B\times\mathcal C\times\mathcal D$ whose
weights sum to $T$.  The objective is to maximize the number of selected
quadruples.  A matching of size $p$ is \emph{perfect}: it uses every
item of every class exactly once.
\end{definition}

\begin{definition}[\MaxSNthreeDM]\label{apx:def:max-sn3dm}
An instance of \MaxSNthreeDM is an instance of \MaxNthreeDM together with a weighted occurrence set $S$ and bijections from $S$ to the three classes under which corresponding occurrences have equal weights.  Equivalently, the classes are disjoint labeled copies $S^X,S^Y,S^Z$ of one weighted occurrence set $S$.  The objective is the maximum number $\OPT_{\mathrm{SN}}$ of pairwise item-disjoint class-transversal triples having the common target sum.
\end{definition}

\begin{definition}[Perfect-completeness gap]\label{apx:def:perfect-gap}
Let $\Pi$ be a maximization problem and let $N(I)$ be a polynomial-time computable upper bound on the optimum.  A perfect-completeness gap of width $\varepsilon>0$ is an \NP-hard promise distinction between
\[
 \OPT_\Pi(I)=N(I)
 \quad\text{and}\quad
 \OPT_\Pi(I)\le(1-\varepsilon)N(I).
\]
\end{definition}

Such a gap immediately excludes a PTAS.  The substantive work is to establish a constant width rather than the vanishing gap $N(I)$ versus $N(I)-1$.

\subsection{Technical vocabulary for the quantitative proof}\label{apx:subsec:technical-vocabulary}

\begin{definition}[Defect]\label{apx:def:defect}
If a maximization instance has a natural perfect upper bound $B$ and a feasible solution has value $B-d$, then $d$ is its \emph{defect}.  A defect statement is useful for approximation only when it bounds the destination defect and source defect by constants independent of the instance size.
\end{definition}

\begin{definition}[Actual and dummy items]\label{apx:def:actual-dummy}
In the first numerical compiler, one chosen incidence of each bounded-\ThreeDM vertex is represented by an \emph{actual} item.  Its remaining incidences are represented by \emph{dummy} items.  An actual target quadruple records a selected source edge; a dummy target quadruple disposes of an unselected edge without consuming another actual occurrence of its endpoint.
\end{definition}

\begin{definition}[Port, private item, main triple, and filler triple]\label{apx:def:ports-fillers}
In the symmetric construction, $P_1,P_2,P_3$ are \emph{ports} carrying the three source roles.  The occurrences $Q_1,\ldots,Q_{10}$ are \emph{private items} used only to regulate local incidence counts.  A triple of ports is a \emph{main triple} and represents one selected source triple.  Every other permitted role pattern is a \emph{filler triple}; filler triples use two of the three available incidences of every port and all three incidences of every private item in a perfect solution.
\end{definition}

\begin{definition}[Role filter, locality coordinates, and no carry]\label{apx:def:coordinate-roles}
The \emph{role coordinate} admits only a finite catalogue of role patterns.  The \emph{group} and \emph{group-square} coordinates force private items in a filler triple to come from one source group.  The \emph{source coordinate} enforces the original target equation.  A mixed-radix base has \emph{no carry} when a sum of three encoded item weights reaches the scalar target if and only if all logical coordinate sums reach their vector targets separately.
\end{definition}

\paragraph{How to use these definitions.}
The proof repeatedly separates classification from counting.  Role and locality coordinates first determine what a selected triple can mean.  Incidence equations then count how many triples of each permitted meaning can occur.  Finally, the source coordinate interprets the surviving main triples.  Reading the proof in this order prevents the large encoded integers from obscuring the combinatorial invariant.

\section{Proof overview}\label{apx:sec:proof-overview}

The proof is a chain of three gap-preserving transformations followed by one quantitative stability argument.  The first two transformations are included because the approximation source is a bounded hypergraph problem, whereas the symmetric construction expects a numerical three-class instance.  The last transformation is the exact construction from Part~I, strengthened so that partial solutions can be decoded.

\begin{figure}[H]
\centering
\begin{tikzpicture}[
  box/.style={draw,rounded corners,align=center,minimum height=1.08cm,text width=2.35cm,fill=blue!3,font=\small},
  note/.style={draw,rounded corners,align=center,minimum height=.80cm,text width=3.15cm,fill=gray!6,font=\scriptsize},
  arr/.style={-{Latex[length=2.2mm]},thick},
  node distance=5.2mm
]
\node[box] (petrank) {bounded \ThreeDM\\perfect-endpoint gap};
\node[box,right=of petrank] (equal) {equal-part\\bounded \ThreeDM};
\node[box,right=of equal] (four) {numerical\\four-dimensional matching};
\node[box,right=of four] (nthree) {unary\\\MaxNthreeDM};
\node[box,right=of nthree] (sym) {symmetric\\\MaxSNthreeDM};
\draw[arr] (petrank) -- (equal);
\draw[arr] (equal) -- (four);
\draw[arr] (four) -- (nthree);
\draw[arr] (nthree) -- (sym);
\node[note,below=8mm of four] (defects) {the two numerical compilers preserve a defect $d$ one-for-one};
\node[note,below=8mm of sym] (stability) {the symmetric reverse step turns defect $d$ into at most $21d$ lost source triples};
\draw[arr] (four.south) -- (defects.north);
\draw[arr] (sym.south) -- (stability.north);
\end{tikzpicture}
\caption{Architecture of the approximation-hardness proof.  Perfect completeness moves from left to right.  Quantitative soundness moves from right to left by tracking the number of missing triples.}
\label{apx:fig:gap-architecture}
\end{figure}
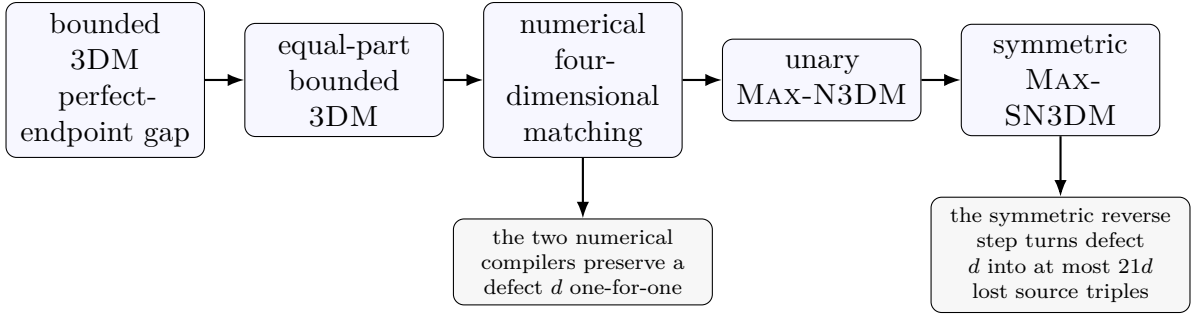

\paragraph{How to read \cref{apx:fig:gap-architecture}.}
The top row describes constructions, not algorithms for solving the problems.  In the forward direction, perfection is preserved at every arrow.  In the reverse direction, begin at the far right and attach a defect to the observed solution.  The middle two compilers lose at most one recovered source object for each missing destination object.  Only the final symmetric step has a larger constant, because one missing triple creates three incidence deficits and can disturb both filler multiplicities and repeated-value accounting.  The lower boxes identify exactly where the constant in the final gap is spent.

\begin{table}[H]
\centering
\caption{Defect and size ledger for the reduction chain.}
\label{apx:tab:gap-ledger}
\small
\renewcommand{\arraystretch}{1.15}
\begin{tabularx}{\textwidth}{@{}p{.21\textwidth}p{.25\textwidth}p{.29\textwidth}X@{}}
\toprule
Stage & Perfect upper bound & Reverse statement from a solution of defect $d$ & Size relation used in the gap\\
\midrule
Equal-part bounded \ThreeDM & $q$ edges & the Petrank NO case permits at most $(1-\gamma)q$ & maximum degree $3$\\
Numerical four-dimensional matching & $p$ quadruples & $p-d$ quadruples recover at least $q-d$ source edges & $q\le p\le3q$\\
Compiled \MaxNthreeDM & $N=p+r$ triples & $N-d$ triples recover at least $p-d$ quadruples & $p\le r\le3p$, hence $N\le12q$\\
Symmetric \MaxSNthreeDM & $M=13n$ triples & $M-d$ triples recover at least $n-21d$ source triples & $M=13n$\\
\bottomrule
\end{tabularx}
\end{table}

\paragraph{How to read \cref{apx:tab:gap-ledger}.}
The third column is the important one.  It records an additive loss statement, not an approximation ratio.  The fourth column converts that additive statement into a relative gap by comparing the destination upper bound with the source size.  Multiplying the two size losses gives the displayed conservative constant $\varepsilon_{\mathrm{SN}}=\gamma/3{,}276$.  No numerical magnitude enters this ratio; polynomial bounds on the encoded numbers are checked separately to retain unary hardness.

\subsection{The two directions at a glance}
In the YES direction, a perfect bounded matching produces a perfect numerical four-dimensional matching, then a perfect three-dimensional numerical matching, and finally a perfect symmetric matching.  In the NO direction, assume the symmetric matching is too close to perfect.  The defect-stability lemma recovers a numerical matching that is too close to perfect; the two compilers then recover a bounded matching that violates Petrank's source gap.  This contradiction yields a constant fraction of missing symmetric triples.

\subsection{Why the proof is longer than the strong-hardness proof}
The exact proof needs to analyze only degree three at every underlying occurrence.  Approximation hardness requires a statement for every degree pattern at most three.  The deficit variables make those deviations explicit, and the repair argument handles repeated numerical values that could otherwise cause the decoded main edges to reuse a source value too often.  These are quantitative additions to the Part~I reverse implication, not changes to the basic role-and-filler construction.

\section{An equal-part bounded-\ThreeDM gap}\label{apx:sec:3dm-gap}

Petrank's theorem is formulated for a quality version of bounded three-dimensional matching.  We first state that result accurately, and then derive the equal-part maximum-matching gap used by the numerical construction.

Let $G=(X,Y,Z;E)$ be a three-partite hypergraph.  For a subfamily $F\subseteq E$, define
\[
 \sat_G(F)=\bigl|\{v\in X\uplusm Y\uplusm Z:\deg_F(v)=1\}\bigr|.
\]
Thus $\sat_G(F)$ counts vertices that occur exactly once in the selected hyperedges; the selected family itself need not be a matching.

\begin{theorem}[Petrank's gap-location theorem]\label{apx:thm:petrank-quality}
There is a constant $\gamma>0$ such that it is \NP-hard to distinguish three-partite hypergraphs $G=(X,Y,Z;E)$ of maximum vertex degree at most three satisfying
\begin{align*}
 \textnormal{\rmfamily(YES)}&\qquad \max_{F\subseteq E}\sat_G(F)=|X|+|Y|+|Z|,\\
 \textnormal{\rmfamily(NO)}&\qquad \max_{F\subseteq E}\sat_G(F)\le
 (1-\gamma)(|X|+|Y|+|Z|).
\end{align*}
\end{theorem}

\begin{proof}[Source]
This is the case $B=3$ of \citet[Theorem~4.4]{Petrank1994}.  The bounded optimization problem was earlier shown MAX SNP-complete by \citet{Kann1991}.
\end{proof}
Theorem~\ref{apx:thm:petrank-quality} is stated exactly as in the source:
Petrank's instances are not required to have equal part sizes, and his
objective counts exactly-once-covered elements over arbitrary subfamilies
rather than the cardinality of a matching.  The next lemma performs the
two conversions needed by the numerical construction: it equalizes the
three part sizes and restates the gap in terms of maximum matchings.

\begin{lemma}[Cyclic equalization]\label{apx:lem:cyclic-equalization}
For the same constant $\gamma$, it is \NP-hard to distinguish three-dimensional matching instances
\[
 H=(U,V,W;\mathcal E),
 \qquad |U|=|V|=|W|=q,
\]
of maximum degree at most three satisfying
\begin{align*}
 \textnormal{\rmfamily(YES)}&\qquad \OPT_{3\mathrm{DM}}(H)=q,\\
 \textnormal{\rmfamily(NO)}&\qquad \OPT_{3\mathrm{DM}}(H)\le(1-\gamma)q.
\end{align*}
\end{lemma}

\begin{proof}
Start with an instance $G=(X,Y,Z;E)$ from \cref{apx:thm:petrank-quality}, and put
\[
 q=|X|+|Y|+|Z|.
\]
Construct three vertex-disjoint cyclic copies of $G$.  In copy $0$, map $(x,y,z)$ to
\[
 (x^0,y^0,z^0)\in U\times V\times W.
\]
In copy $1$, map it to
\[
 (z^1,x^1,y^1)\in U\times V\times W,
\]
and in copy $2$, map it to
\[
 (y^2,z^2,x^2)\in U\times V\times W.
\]
Consequently,
\begin{align*}
 U&=X^0\uplusm Z^1\uplusm Y^2,\\
 V&=Y^0\uplusm X^1\uplusm Z^2,\\
 W&=Z^0\uplusm Y^1\uplusm X^2,
\end{align*}
so all three new parts have cardinality $q$.  The maximum degree remains at most three.

In the ``yes'' case, there is a family $F\subseteq E$ in which every original vertex occurs exactly once.  Such a family is a perfect matching in $G$.  Its three cyclic copies form a perfect matching of size $q$ in $H$.

In the ``no'' case, let $M_j$ be the restriction of any matching $M$ of $H$ to cyclic copy $j$.  Under the inverse copy map, $M_j$ is a matching in $G$, and therefore its selected edges make exactly $3|M_j|$ original vertices occur once.  Hence
\[
 3|M_j|\le\max_{F\subseteq E}\sat_G(F)\le(1-\gamma)q.
\]
Summing over the three disjoint copies gives
\[
 |M|=|M_0|+|M_1|+|M_2|\le(1-\gamma)q.
\]
The construction is polynomial and proves the lemma.
\end{proof}

\paragraph{How to read the equalization argument.}
The cyclic copies do not amplify the gap; they only balance the three part sizes.  Each original vertex appears once in each of the three new parts across the three copies, so every part has size $q$.  In the NO case, the matching restrictions $M_0,M_1,M_2$ are disjoint and may be evaluated separately by Petrank's saturated-vertex objective.  The factor three in $3|M_j|$ is canceled by the three cyclic copies, leaving the same relative gap $\gamma$.

\section{A unary perfect-completeness gap for \MaxNthreeDM}\label{apx:sec:n3dm-gap}

We now transform the equal-part bounded instance from \cref{apx:lem:cyclic-equalization} into a numerical instance.  The transformation is written out because ordinary \MaxThreeDM and \MaxNthreeDM are different optimization problems.

\subsection{From bounded \ThreeDM to numerical four-dimensional matching}

Let
\[
 H=(U,V,W;E),
 \qquad
 U=\{u_1,\ldots,u_q\},\quad
 V=\{v_1,\ldots,v_q\},\quad
 W=\{z_1,\ldots,z_q\},
\]
be an instance from \cref{apx:lem:cyclic-equalization}.  If $H$ has an isolated vertex, it cannot be a ``yes'' instance and may be mapped to a fixed ``no'' instance of the destination promise problem --- for concreteness, the \MaxNthreeDM{} instance with one item of weight $1$ in each class and target $4$, whose optimum is $0$ and which therefore satisfies the ``no'' promise.  We may therefore assume that every vertex has degree at least one.  Put
\[
 p=|E|,
 \qquad
 R=q+1.
\]
The degree bound gives
\begin{equation}\label{apx:eq:p-bounds}
 q\le p\le3q.
\end{equation}
For every vertex, choose one incident edge
\[
 \tau_U(u_i),\qquad \tau_V(v_j),\qquad \tau_W(z_k).
\]

For every edge $e=(u_i,v_j,z_k)\in E$, create items $A_e,B_e,C_e,D_e$ in four classes $\mathcal A,\mathcal B,\mathcal C,\mathcal D$, each of size $p$.  Define
\begin{align}
 \wt(A_e)&=iR^2+2R^3\mathbf 1[e=\tau_U(u_i)],\label{apx:eq:n4-a}\\
 \wt(B_e)&=jR+R^3\mathbf 1[e\ne\tau_V(v_j)],\label{apx:eq:n4-b}\\
 \wt(C_e)&=k+R^3\mathbf 1[e\ne\tau_W(z_k)],\label{apx:eq:n4-c}\\
 \wt(D_e)&=2R^3-iR^2-jR-k,\label{apx:eq:n4-d}
\end{align}
and use target
\begin{equation}\label{apx:eq:n4-target}
 T_4=4R^3.
\end{equation}
The item $A_{\tau_U(u_i)}$ is the \emph{actual} $A$-item representing $u_i$; all other $A$-items representing $u_i$ are \emph{dummy} items.  Actual and dummy $B$- and $C$-items are defined analogously.

All constructed weights are positive, as \cref{apx:def:max-n4dm} requires: the $A$-, $B$-, and $C$-weights are at least $iR^2\ge R^2$, $jR\ge R$, and $k\ge1$, respectively, and since $i,j,k\le q=R-1$,
\[
 iR^2+jR+k\le(R-1)(R^2+R+1)=R^3-1,
 \qquad\text{hence}\qquad
 \wt(D_e)\ge R^3+1>0.
\]

\paragraph{How to read the four-coordinate encoding.}
The powers $1,R,R^2$ record the three endpoint indices, while the $R^3$ digit distinguishes actual from dummy behavior.  The item $D_e$ contains the additive complement of the endpoint indices of edge $e$.  Since every index lies in $\{1,\ldots,q\}$ and $R=q+1$, the lower three digits have no carry.  A target quadruple therefore cannot splice together unrelated endpoints.

\begin{lemma}[Target-quadruple characterization]\label{apx:lem:n4-characterization}
Every target quadruple containing $D_e$, where $e=(u_i,v_j,z_k)$, uses an $A$-item representing $u_i$, a $B$-item representing $v_j$, and a $C$-item representing $z_k$.  Moreover, the quadruple is either all actual or all dummy in the following sense: it contains the three chosen actual endpoint items, or it contains three dummy endpoint items.
\end{lemma}
\begin{proof}
Consider a target quadruple containing $D_e$, and suppose its $A$-, $B$-,
and $C$-items represent the vertices $u_{i'}$, $v_{j'}$, and $z_{k'}$,
respectively.  Substituting
\eqref{apx:eq:n4-a}--\eqref{apx:eq:n4-d} into the target equation
$\wt(A)+\wt(B)+\wt(C)+\wt(D_e)=4R^3$ and collecting terms gives
\begin{equation}\label{apx:eq:n4-signed}
 (i'-i)R^2+(j'-j)R+(k'-k)
 \;=\;
 R^3\bigl(2-2\,\alpha_A-\beta_B-\beta_C\bigr),
\end{equation}
where $\alpha_A=\mathbf 1[\text{the $A$-item is actual}]$ and
$\beta_B,\beta_C$ are the corresponding dummy indicators of the $B$- and
$C$-items.  Note that $\wt(D_e)$ is not written in nonnegative base-$R$
digits, so the two sides of \eqref{apx:eq:n4-signed} are compared by
magnitude rather than digit by digit.  Since all indices lie in
$\{1,\ldots,q\}$ and $R=q+1$, each difference on the left has absolute
value at most $R-2$, so the left side has absolute value at most
$(R-2)(R^2+R+1)<R^3$.  The right side is an integer multiple of $R^3$.
Hence both sides equal zero.

The vanishing of the left side forces the indices.  Reducing
\eqref{apx:eq:n4-signed} modulo $R$ gives $k'=k$.  Dividing by $R$ and repeating the argument
gives $j'=j$ and then $i'=i$.  Thus the quadruple uses an $A$-item
representing $u_i$, a $B$-item representing $v_j$, and a $C$-item
representing $z_k$.

The vanishing of the right side gives $2\alpha_A+\beta_B+\beta_C=2$.
Since $\alpha_A,\beta_B,\beta_C\in\{0,1\}$, the only solutions are
$(\alpha_A,\beta_B,\beta_C)=(1,0,0)$ and $(0,1,1)$: either the $A$-item
is actual and the $B$- and $C$-items are actual, or the $A$-item is dummy
and the $B$- and $C$-items are dummy.  These are exactly the all-actual
and all-dummy cases.
\end{proof}

\begin{lemma}[Completeness and defect recovery in four dimensions]\label{apx:lem:n4-gap}
The following statements hold.
\begin{enumerate}[label=(\alph*)]
 \item If $H$ has a perfect matching of size $q$, the numerical four-dimensional instance has a perfect matching of size $p$.
 \item If the numerical four-dimensional instance has a matching of size $p-d$, then $H$ has a matching of size at least $\max\{0,q-d\}$.
\end{enumerate}
\end{lemma}

\begin{proof}
For part (a), let $M$ be a perfect matching of $H$.  For each vertex, assign its actual endpoint item to its unique incident edge in $M$.  Assign the remaining dummy items bijectively to the incident edges outside $M$.  Make the analogous independent assignments for the $V$- and $W$-vertices.  Every edge in $M$ receives three actual items, and every edge outside $M$ receives three dummy items.  With its $D_e$ item, each edge forms a target quadruple by \cref{apx:lem:n4-characterization}.  All items are used exactly once.

For part (b), a matching of size $p-d$ leaves exactly $d$ items unused in each class.  There are $q$ actual $A$-items, so at least $q-d$ of them are used when $d<q$.  By \cref{apx:lem:n4-characterization}, every quadruple containing an actual $A$-item is all actual.  The corresponding edges of $H$ are pairwise disjoint because distinct quadruples use distinct actual items for all three source parts.  They therefore form a matching of size at least $q-d$.  The assertion is trivial when $d\ge q$.
\end{proof}

\subsection{Compiling numerical four-dimensional matching into \MaxNthreeDM}

Let $\mathcal P$ be the set of admissible role-ordered pairs
\begin{equation}\label{apx:eq:admissible-pairs}
 \mathcal P=
 \left\{(C,D)\in\mathcal C\times\mathcal D:
 \begin{array}{l}
 \text{there exist }A\in\mathcal A,\ B\in\mathcal B\text{ such that}\\[-1mm]
 \wt(A)+\wt(B)+\wt(C)+\wt(D)=T_4
 \end{array}
 \right\}.
\end{equation}
For $P=(C,D)\in\mathcal P$, write
\[
 \sigma(P)=\wt(C)+\wt(D),
\]
and put $r=|\mathcal P|$.  For each $D_e$, pairing it with the actual $C$-item representing the $W$-endpoint of $e$ gives an admissible pair.  Conversely, \cref{apx:lem:n4-characterization} implies that every admissible $C$-item paired with $D_e$ represents that same endpoint.  There are at most three such $C$-items.  Therefore
\begin{equation}\label{apx:eq:r-bounds}
 p\le r\le3p.
\end{equation}

For every $P\in\mathcal P$, create two pair items $L_P$ and $R_P$, and create a filler set $F$ of cardinality $r-p$.  Define
\begin{align*}
 X_N&=\mathcal A\uplusm\{L_P:P\in\mathcal P\},\\
 Y_N&=\mathcal B\uplusm\mathcal C\uplusm F,\\
 Z_N&=\mathcal D\uplusm\{R_P:P\in\mathcal P\}.
\end{align*}
Each class has
\begin{equation}\label{apx:eq:N-size}
 N=p+r
\end{equation}
items.

Assign every item a role code $\rho$ and numerical coordinate $\eta$ as shown in \cref{apx:tab:compiler-role-data}.
\begin{table}[H]
\centering
\caption{Role and numerical coordinates in the four-to-three-dimensional compiler.}
\label{apx:tab:compiler-role-data}
\small
\begin{tabular}{@{}c|rrrrrrr@{}}
\toprule
Type & $A$ & $L$ & $B$ & $C$ & $F$ & $D$ & $R$\\
\midrule
$\rho$ & 20 & 1 & 70 & 30 & 89 & 69 & 10\\
$\eta$ & $\wt(A)$ & $T_4-\sigma(P)$ & $\wt(B)$ & $\wt(C)$ & 0 & $\wt(D)$ & $\sigma(P)$\\
\bottomrule
\end{tabular}
\end{table}

\paragraph{How to read \cref{apx:tab:compiler-role-data}.}
The role row is a small finite filter: among the twelve class-transversal type combinations, only three sum to $100$.  The numerical row then gives those three patterns complementary meanings.  The $(A,B,R)$ triple records the left half of a quadruple, the $(L,C,D)$ triple records the right half, and the $(L,F,R)$ triple pairs unused indices.  The scalar multiplier $300$ is larger than every possible role sum, so role equality and numerical equality cannot compensate for one another.

Encode the two coordinates by
\begin{equation}\label{apx:eq:n3-compile-weight}
 W_N(x)=\rho(x)+300\eta(x),
 \qquad
 T_N=100+300T_4.
\end{equation}
All item weights are positive.  In particular, admissibility of $P$ implies
\[
 T_4-\sigma(P)=\wt(A)+\wt(B)>0
\]
for some $A,B$.

\begin{lemma}[Allowed compiled triples]\label{apx:lem:compiled-patterns}
A class-transversal triple has weight $T_N$ if and only if it has one of the following type patterns and satisfies the displayed numerical condition:
\begin{align*}
 (A,B,R_P)&:& \wt(A)+\wt(B)+\sigma(P)&=T_4,\\
 (L_P,C,D)&:& \wt(C)+\wt(D)&=\sigma(P),\\
 (L_P,F,R_{P'})&:& \sigma(P)&=\sigma(P').
\end{align*}
\end{lemma}
\begin{proof}
Every class-transversal role sum lies strictly between $0$ and $300$.  Hence equality in \eqref{apx:eq:n3-compile-weight} forces the role sum to be exactly $100$ and the numerical-coordinate sum to be exactly $T_4$.  Checking the twelve possible combinations of an $X_N$-type, a $Y_N$-type, and a $Z_N$-type shows that the role sum is $100$ only for
\[
 (A,B,R),\qquad (L,C,D),\qquad (L,F,R).
\]
Substitution of the $\eta$-coordinates gives the stated conditions.
\end{proof}

\begin{lemma}[Completeness and defect recovery in three dimensions]\label{apx:lem:n3-compile-gap}
The following statements hold.
\begin{enumerate}[label=(\alph*)]
 \item If the numerical four-dimensional instance has a perfect matching of size $p$, the compiled \NthreeDM instance has a perfect matching of size $N$.
 \item If the compiled \NthreeDM instance has a matching of size $N-d$, the numerical four-dimensional instance has a matching of size at least $\max\{0,p-d\}$.
\end{enumerate}
\end{lemma}

\begin{proof}
For part (a), for every quadruple $(A,B,C,D)$ in a perfect four-dimensional matching, let $P=(C,D)$ and select
\[
 (A,B,R_P)
 \qquad\text{and}\qquad
 (L_P,C,D).
\]
The $p$ quadruples use $p$ distinct pair indices.  For every unused $P\in\mathcal P$, choose a distinct filler $f\in F$ and select $(L_P,f,R_P)$.  This gives
\[
 2p+(r-p)=p+r=N
\]
disjoint target triples.

For part (b), group the selected triples by their pair sum $s$.  Let $u_s,v_s,f_s$ be the numbers of selected triples of types $(A,B,R)$, $(L,C,D)$, and $(L,F,R)$, respectively, with pair sum $s$, and let $r_s$ be the number of pair indices $P$ satisfying $\sigma(P)=s$.  Capacity of the $R$-items and $L$-items gives
\[
 u_s+f_s\le r_s,
 \qquad
 v_s+f_s\le r_s.
\]
Consequently,
\begin{equation}\label{apx:eq:min-pair-bound}
 \min\{u_s,v_s\}\ge u_s+v_s+f_s-r_s.
\end{equation}
Pair arbitrarily $\min\{u_s,v_s\}$ triples of the first two types for every $s$.  Each pair yields a valid four-dimensional quadruple because
\[
 \wt(A)+\wt(B)=T_4-s,
 \qquad
 \wt(C)+\wt(D)=s.
\]
The recovered quadruples are item-disjoint.  Summing \eqref{apx:eq:min-pair-bound} over $s$ yields at least
\[
 (N-d)-r=p-d
\]
quadruples, unless this expression is negative, in which case the empty matching proves the stated maximum with zero.
\end{proof}

\paragraph{How to read the pairing inequality.}
For a fixed pair sum $s$, the $R$-items limit the combined number of left-half and filler triples, while the $L$-items limit the combined number of right-half and filler triples.  The same filler count appears in both inequalities.  Subtracting the available pair indices $r_s$ therefore leaves a guaranteed overlap between left and right halves.  Summed over all $s$, the compiler loses at most one reconstructed quadruple for every missing compiled triple.

\begin{theorem}[Unary perfect-completeness gap for \MaxNthreeDM]\label{apx:thm:n3dm-gap}
There is a constant $\varepsilon_{\mathrm N}>0$ such that it is \NP-hard to distinguish unary \MaxNthreeDM instances with $N$ items in each class satisfying
\begin{align*}
 \textnormal{\rmfamily(YES)}&\qquad \OPT_{\mathrm N}=N,\\
 \textnormal{\rmfamily(NO)}&\qquad \OPT_{\mathrm N}\le(1-\varepsilon_{\mathrm N})N.
\end{align*}
One may take
\begin{equation}\label{apx:eq:epsilon-N}
 \varepsilon_{\mathrm N}=\frac{\gamma}{12},
\end{equation}
where $\gamma$ is the constant in \cref{apx:thm:petrank-quality}.
\end{theorem}

\begin{proof}
In the ``yes'' case, \cref{apx:lem:n4-gap,apx:lem:n3-compile-gap} give a perfect compiled numerical matching.

Suppose the bounded source $H$ is a ``no'' instance and the compiled instance has a matching of size $N-d$.  If $d<q$, \cref{apx:lem:n3-compile-gap,apx:lem:n4-gap} recover a matching of $H$ of size at least $q-d$.  Therefore $d\ge\gamma q$, since otherwise the recovered matching would have size greater than $(1-\gamma)q$.  If $d\ge q$, the same inequality follows because $0<\gamma\le1$.  By \eqref{apx:eq:p-bounds}, \eqref{apx:eq:r-bounds}, and \eqref{apx:eq:N-size},
\[
 N=p+r\le4p\le12q.
\]
Thus
\[
 d\ge\gamma q\ge\frac{\gamma}{12}N,
\]
which proves the gap.

The bounded source has $O(q)$ edges.  The four-dimensional weights are $O(q^3)$, and the compiled three-dimensional weights remain $O(q^3)$.  The number of items is $O(q)$.  Hence the construction has polynomial length even when every number is written in unary.
\end{proof}

\paragraph{Reading the source gap.}
The constant $\varepsilon_{\mathrm N}=\gamma/12$ is obtained by comparing a source defect of at least $\gamma q$ with the destination upper bound $N\le12q$.  The number twelve is a size-expansion bound, not a property of numerical magnitudes.  This distinction is important: approximation loss is controlled combinatorially, while unary hardness is preserved by a separate polynomial bound on every encoded integer.

\section{The symmetric numerical construction}\label{apx:sec:symmetric-construction}

We restate the portion of the construction developed in Part~I that is needed for the quantitative analysis.  The construction itself is unchanged from Part~I.

Let the source instance have indexed weights
\[
 A=\mset{a_1,\ldots,a_n},\qquad
 B=\mset{b_1,\ldots,b_n},\qquad
 C=\mset{c_1,\ldots,c_n},
\]
and target $t$.  Equal numerical values may occur at different indices.  For each group $i$, create
\[
 S_i=\{P_{1,i},P_{2,i},P_{3,i},Q_{1,i},\ldots,Q_{10,i}\},
 \qquad
 S=\mathop{\uplusm}_{i=1}^n S_i.
\]
The physical output classes are the three labeled copies $S^X,S^Y,S^Z$, so each class has $13n$ occurrences and all three classes have the same weight multiset.

The role code $\rho_R$, group coefficient $h_R$, and source coordinate $\eta_{R,i}$ are given in \cref{apx:tab:symmetric-weights}.

\begin{table}[htbp]
\centering
\caption{Logical coordinates of the thirteen roles in source group $i$.}\label{apx:tab:symmetric-weights}
\begin{tabular}{@{}c r r l@{}}
\toprule
Role $R$ & $\rho_R$ & $h_R$ & $\eta_{R,i}$\\
\midrule
$P_1$&31&0&$3a_i$\\
$P_2$&8&0&$3b_i$\\
$P_3$&261&0&$3c_i$\\
$Q_1$&184&2&$-3b_i$\\
$Q_2$&85&$-2$&$3t-3a_i+3b_i$\\
$Q_3$&177&$-2$&$3t-2a_i+b_i+c_i$\\
$Q_4$&115&2&$2a_i-4b_i-c_i$\\
$Q_5$&23&2&$a_i-2b_i-2c_i$\\
$Q_6$&16&$-2$&$3t-a_i+2b_i-c_i$\\
$Q_7$&27&1&$2a_i-b_i-c_i$\\
$Q_8$&89&$-3$&$3t-2a_i+4b_i+c_i$\\
$Q_9$&188&1&$a_i-2b_i+c_i$\\
$Q_{10}$&96&1&$0$\\
\bottomrule
\end{tabular}
\end{table}

\paragraph{How to read \cref{apx:tab:symmetric-weights}.}
Read the columns from left to right.  The role code determines which types may meet.  The source coordinate makes $P_1,P_2,P_3$ carry $a_i,b_i,c_i$ and assigns each private role the residual needed by its filler identities.  The coefficient $h_R$ is reused twice, once with $i$ and once with $i^2$, to force local private items to share an index.  Negative source coordinates are harmless at this vector stage; the later shifts make every scalar weight positive without changing any triple equality.

For an occurrence $R_i\in S_i$, define
\begin{equation}\label{apx:eq:symmetric-vector}
 \omegavec(R_i)=(\rho_R,\eta_{R,i},h_Ri,h_Ri^2),
 \qquad
 \Kvec=(300,3t,0,0).
\end{equation}
The permitted role patterns are
\begin{align}
 F_0&=(P_1,P_2,P_3),\notag\\
 F_1&=(P_1,Q_1,Q_2), &
 F_2&=(P_2,Q_3,Q_4), &
 F_3&=(P_3,Q_5,Q_6),\notag\\
 F_4&=(Q_1,Q_7,Q_8), &
 F_5&=(Q_2,Q_7,Q_9), &
 F_6&=(Q_3,Q_7,Q_{10}),\notag\\
 F_7&=(Q_4,Q_8,Q_{10}), &
 F_8&=(Q_5,Q_8,Q_9), &
 F_9&=(Q_6,Q_9,Q_{10}).\label{apx:eq:F-patterns}
\end{align}
A triple of pattern $F_0$ is a \emph{main triple}; the other patterns are \emph{filler triples}.

\paragraph{Roles are not output classes.}
The symbols $P_1,P_2,P_3,Q_1,\ldots,Q_{10}$ name logical roles inside the common occurrence set $S$.  Each role occurrence has an $X$-copy, a $Y$-copy, and a $Z$-copy in the physical output.  Thus a main triple does not mean ``take $P_1$ from class $X$, $P_2$ from class $Y$, and $P_3$ from class $Z$'' in a fixed order.  Recoloring chooses which physical copy supplies each incidence.  This separation is what permits the three output weight multisets to be identical while the internal roles remain asymmetric.

\subsection{Mixed-radix packing and no carry}

Choose shifts $D_1,D_2,D_3$ so that every second, third, and fourth item coordinate becomes nonnegative.  Choose a base $L$ larger than every target digit and every sum of three shifted item digits.  Define
\begin{equation}\label{apx:eq:mixed-radix}
 W(R_i)=\rho_R+(\eta_{R,i}+D_1)L+(h_Ri+D_2)L^2+(h_Ri^2+D_3)L^3,
\end{equation}
and
\begin{equation}\label{apx:eq:encoded-target}
 K^\star=300+(3t+3D_1)L+3D_2L^2+3D_3L^3.
\end{equation}
All item weights and the target are positive integers.

\begin{lemma}[No-carry encoding]\label{apx:lem:no-carry}
For any three occurrences $u,v,z\in S$,
\begin{equation}\label{apx:eq:no-carry-equivalence}
 W(u)+W(v)+W(z)=K^\star
 \quad\Longleftrightarrow\quad
 \omegavec(u)+\omegavec(v)+\omegavec(z)=\Kvec.
\end{equation}
\end{lemma}
\begin{proof}
By the choice of $L$, every digit of the sum of three encoded item weights lies in $\{0,\ldots,L-1\}$, as does every target digit.  Hence no carry occurs in any coordinate, and equality of encoded integers is equivalent to coordinatewise equality of the shifted vectors.  Subtracting $3D_r$ in shifted coordinates $r=1,2,3$ gives the original vector equality.
\end{proof}

\paragraph{How to read the mixed-radix encoding.}
The scalar integer is a container for four already designed coordinates.  A target-sum proof should therefore be performed coordinate by coordinate and only then translated to the scalar equality.  The shifts add the same constant to every item digit and three times that constant to the target digit, so they preserve all three-item equalities.  The large base prevents an error in one coordinate from being hidden by a carry into the next.

\begin{lemma}[Two moments force a common index]\label{apx:lem:two-moments}
Let $\alpha,\beta,\chi$ be nonzero real numbers satisfying $\alpha+\beta+\chi=0$.  If
\[
 \alpha i+\beta j+\chi k=0,
 \qquad
 \alpha i^2+\beta j^2+\chi k^2=0,
\]
then $i=j=k$.
\end{lemma}

\begin{proof}
Since the coefficients are nonzero and sum to zero, two of them share one sign and the third has the opposite sign.    After permuting the variables and multiplying by $-1$ if necessary, assume $\alpha,\beta>0$ and $\chi=-(\alpha+\beta)$.  The first equality gives
\[
 k=\frac{\alpha i+\beta j}{\alpha+\beta}.
\]
Substitution into the second gives
\[
 0=\frac{\alpha\beta}{\alpha+\beta}(i-j)^2.
\]
Thus $i=j$, and the first equality then gives $k=i$.
\end{proof}

\begin{lemma}[Structure of encoded target triples]\label{apx:lem:target-structure}
Every triple of encoded weight $K^\star$ has one of the role patterns $F_0,\ldots,F_9$.  Moreover:
\begin{enumerate}[label=(\alph*)]
 \item a main triple has the form $(P_{1,i},P_{2,j},P_{3,k})$ and is legal if and only if $a_i+b_j+c_k=t$;
 \item in every filler triple, all private $Q$-items have one common group index;
 \item if an $F_1$-triple has private items from group $i$, its $P_1$-port represents the value $a_i$; analogously, an $F_2$- or $F_3$-triple localized at $i$ uses a port representing $b_i$ or $c_i$, respectively;
 \item for every $i$, each local triple $F_r(i)$, $r=1,\ldots,9$, has encoded weight $K^\star$.
\end{enumerate}
\end{lemma}

\begin{proof}
By \eqref{apx:eq:no-carry-equivalence}, a target equality holds coordinatewise.  The role-code sum is therefore $300$.  The exhaustive complementary-pair certificate in \cref{app:role-codes} shows that the only unordered role triples with sum $300$ are $F_0,\ldots,F_9$.

For $F_0$, all group coefficients vanish and the source-coordinate sum is $3a_i+3b_j+3c_k$, proving part (a).

For $F_1,F_2,F_3$, the two private group coefficients are $2,-2$ or $-2,2$, so the group coordinate alone forces their indices to agree.  For $F_4,\ldots,F_9$, the multiset of nonzero group coefficients is either $\{2,1,-3\}$ or $\{-2,1,1\}$; in both cases two coefficients share one sign, the third has the opposite sign, and the sum is zero.  The group and group-square equations, together with \cref{apx:lem:two-moments}, therefore force the three private indices to agree. This proves part (b).

If an $F_1$-triple has private group $i$, then
\[
 \eta_{Q_1,i}+\eta_{Q_2,i}=3t-3a_i.
\]
A port $P_{1,j}$ completes the source-coordinate target if and only if $a_j=a_i$.  Similarly,
\[
 \eta_{Q_3,i}+\eta_{Q_4,i}=3t-3b_i,
 \qquad
 \eta_{Q_5,i}+\eta_{Q_6,i}=3t-3c_i,
\]
which proves part (c).

Finally, direct substitution in \cref{apx:tab:symmetric-weights} shows that every local filler pattern has role sum $300$, source sum $3t$, and group-coefficient sum zero.  Since all occurrences use the same index, both group moments vanish.  Part (d) follows from \eqref{apx:eq:no-carry-equivalence}.
\end{proof}

\subsection{Forgetting and restoring the labeled copies}

\begin{lemma}[Partial copy-forgetting and recoloring]\label{apx:lem:partial-recoloring}
Let $S^X,S^Y,S^Z$ be three labeled copies of a weighted occurrence set $S$.
\begin{enumerate}[label=(\alph*)]
 \item A partial symmetric matching of $h$ target triples becomes, after copy labels are forgotten, a multiset $\mathcal H$ of $h$ three-incidence target hyperedges on $S$ with $\degH(s)\le3$ for every $s\in S$.
 \item Conversely, every multiset $\mathcal H$ of three-incidence target hyperedges on $S$ with maximum incidence degree at most three can be lifted in polynomial time to a partial symmetric matching of the same cardinality.
\end{enumerate}
\end{lemma}

\begin{proof}
Part (a) is immediate: an underlying occurrence has only the three copies $s^X,s^Y,s^Z$, each used at most once.

For part (b), form the incidence bipartite multigraph whose left vertices are occurrences of $S$, whose right vertices are hyperedge occurrences, and whose edges are incidences.  Its maximum degree is at most three, and every hyperedge vertex has degree exactly three.  By K\H{o}nig's line-coloring theorem for bipartite multigraphs \citep{LovaszPlummer1986}, its edges admit a proper coloring with three colors $X,Y,Z$.  At every hyperedge vertex the three incidences receive distinct colors, and at every item vertex no color is repeated.  Replacing an incidence of color $R$ at $s$ by the physical copy $s^R$ produces a class-transversal partial matching using every physical item at most once.
\end{proof}

\paragraph{Why partial recoloring matters here.}
The exact proof in Part~I needs the regular degree-three case.  Approximation hardness needs the nonregular case because omitted triples create vertices of degree zero, one, or two.  K\H{o}nig's line-coloring theorem applies to every bipartite multigraph of maximum degree three, so recoloring preserves the cardinality of an arbitrary partial target-sum hypergraph.  The deficit analysis may therefore be carried out after labels are forgotten and lifted back without additional loss.

\begin{lemma}[Forward extension of a partial source matching]\label{apx:lem:forward-partial}
If the source \MaxNthreeDM instance has a matching of size $k$, the constructed symmetric instance has a matching of size $12n+k$.  In particular, a perfect source matching yields a perfect symmetric matching of size $13n$.
\end{lemma}

\begin{proof}
For every group $i$, include two copies of each local pattern $F_1(i),F_2(i),F_3(i)$ and one copy of each $F_4(i),\ldots,F_9(i)$. As always, $F_r(i)$ takes all its roles, including the port for
$r=1,2,3$, from group $i$, so each port receives exactly two filler
incidences. These are $12n$ legal filler hyperedges.  Every private item then has degree three, and every port has degree two.  Add one main hyperedge for every triple in the source matching.  Source item-disjointness implies that no port receives more than one main incidence, so the maximum incidence degree is at most three.  Apply \cref{apx:lem:partial-recoloring}(b).
\end{proof}

\section{Defect stability of the symmetric construction}\label{apx:sec:stability}

The following lemma is the quantitative strengthening of the reverse implication.  Its proof is easiest to follow as a flow of accounting losses, shown in \cref{apx:fig:defect-flow}.

\begin{figure}[H]
\centering
\begin{tikzpicture}[
  box/.style={draw,rounded corners,align=center,minimum height=1.02cm,text width=2.55cm,fill=blue!3,font=\small},
  arr/.style={-{Latex[length=2.2mm]},thick},
  node distance=6mm
]
\node[box] (partial) {$13n-d$ symmetric\\target triples};
\node[box,right=of partial] (inc) {$3d$ missing\\incidences};
\node[box,right=of inc] (fill) {bounded filler\\multiplicity deviations};
\node[box,right=of fill] (main) {many main triples,\\possibly with value overuse};
\node[box,below=8mm of main] (repair) {delete at most the\\total overuse};
\node[box,left=of repair] (source) {at least $n-21d$\\source triples};
\draw[arr] (partial) -- (inc);
\draw[arr] (inc) -- (fill);
\draw[arr] (fill) -- (main);
\draw[arr] (main) -- (repair);
\draw[arr] (repair) -- (source);
\end{tikzpicture}
\caption{Defect flow in the quantitative reverse proof.  Each arrow is justified by a displayed inequality in the proof of \cref{apx:lem:defect-stability}.}
\label{apx:fig:defect-flow}
\end{figure}
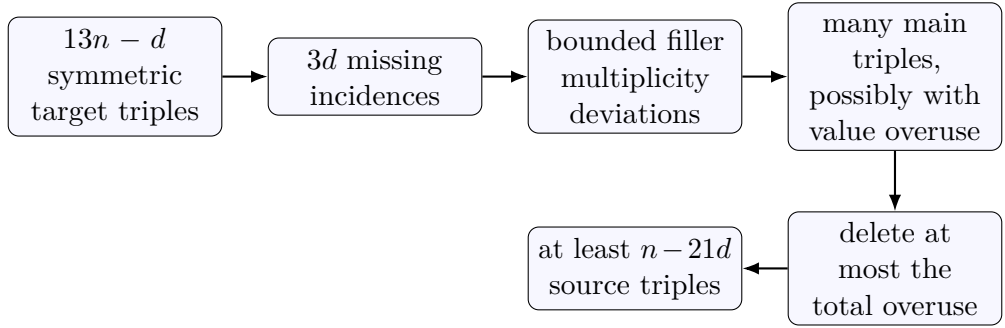

\paragraph{How to read \cref{apx:fig:defect-flow}.}
The first equality is exact: each missing triple accounts for three unused physical copies.  The middle steps are upper bounds.  They show that local private-item deficits can change the intended filler multiplicities only linearly, and that the resulting main-edge family can overuse source values only linearly.  The last step is a repair, not an assumption: main triples are deleted until their value multiplicities fit the labeled source occurrences.  The constant $21$ is the accumulated worst-case charge along this path.

\begin{lemma}[Defect-stability lemma]\label{apx:lem:defect-stability}
Let $I$ be a \MaxNthreeDM instance with $n$ items in each source class, and let $I^\star$ be its symmetric instance constructed in \cref{apx:sec:symmetric-construction}.  If $I^\star$ has a matching of size
\[
 13n-d,
\]
then $I$ has a matching of size at least
\[
 \max\{0,n-21d\}.
\]
Moreover, such a source matching can be recovered in polynomial time from the symmetric matching.
\end{lemma}

\begin{proof}
Forget the outer labels $X,Y,Z$ as in \cref{apx:lem:partial-recoloring}.  We obtain a multiset $\mathcal H$ of $13n-d$ target hyperedges on $S$, with incidence degree at most three at every occurrence.  Define
\begin{equation}\label{apx:eq:item-deficit}
 \delta(s)=3-\degH(s)
 \qquad(s\in S).
\end{equation}
The available incidence capacity is $3|S|=39n$, whereas the selected hyperedges use $3(13n-d)=39n-3d$ incidences.  Hence
\begin{equation}\label{apx:eq:total-deficit}
 \Delta:=\sum_{s\in S}\delta(s)=3d.
\end{equation}
By \cref{apx:lem:target-structure}, every selected hyperedge has a permitted main or filler pattern; every filler hyperedge is group-local on its private items; and the port-value restrictions of \cref{apx:lem:target-structure}(c) hold edge by edge.

\paragraph{Local filler equations with deficits.}
Fix a group $i$.  For $r=1,\ldots,9$, let $x_{r,i}$ be the number of filler hyperedges of pattern $F_r$ whose private items lie in group $i$.  Put
\[
 \delta_{\ell,i}=\delta(Q_{\ell,i}),
 \qquad
 \Delta_i=\sum_{\ell=1}^{10}\delta_{\ell,i}.
\]
The private-item incidence equations are
\begin{align}
 x_{1,i}+x_{4,i}&=3-\delta_{1,i}, &
 x_{1,i}+x_{5,i}&=3-\delta_{2,i},\label{apx:eq:def-q12}\\
 x_{2,i}+x_{6,i}&=3-\delta_{3,i}, &
 x_{2,i}+x_{7,i}&=3-\delta_{4,i},\label{apx:eq:def-q34}\\
 x_{3,i}+x_{8,i}&=3-\delta_{5,i}, &
 x_{3,i}+x_{9,i}&=3-\delta_{6,i},\label{apx:eq:def-q56}\\
 x_{4,i}+x_{5,i}+x_{6,i}&=3-\delta_{7,i},\label{apx:eq:def-q7}\\
 x_{4,i}+x_{7,i}+x_{8,i}&=3-\delta_{8,i},\label{apx:eq:def-q8}\\
 x_{5,i}+x_{8,i}+x_{9,i}&=3-\delta_{9,i},\label{apx:eq:def-q9}\\
 x_{6,i}+x_{7,i}+x_{9,i}&=3-\delta_{10,i}.\label{apx:eq:def-q10}
\end{align}

If $\Delta_i=0$, then since $\Delta_i=\sum_{\ell=1}^{10}\delta_{\ell,i}$ and each $\delta_{\ell,i}\ge 0$, we have $\delta_{1,i}=\cdots=\delta_{10,i}=0$. Thus equations \eqref{apx:eq:def-q12}--\eqref{apx:eq:def-q10} reduce to the exact zero-deficit system. The first six equations give
\[
 x_{4,i}=x_{5,i},
 \qquad
 x_{6,i}=x_{7,i},
 \qquad
 x_{8,i}=x_{9,i}.
\]
Writing these common values as $u,v,w$, the final four equations are
\[
 2u+v=3,
 \quad
 u+v+w=3,
 \quad
 u+2w=3,
 \quad
 2v+w=3,
\]
whose unique solution is $u=v=w=1$. Thus $x_{4,i}=\cdots=x_{9,i}=1$, and from the first six equations we obtain
\[
 x_{1,i}=x_{2,i}=x_{3,i}=2.
\]
Therefore, for $r=1,2,3$,
\[
 \pos{x_{r,i}-2}=0 \le \Delta_i
 \qquad\text{and}\qquad
 \pos{2-x_{r,i}}=0 \le 2\Delta_i.
\]

If $\Delta_i>0$, then $\Delta_i\ge1$, and the first six equations imply
\[
 0\le x_{r,i}\le3
 \qquad(r=1,2,3).
\]
Combining this elementary bound with the solution in the zero-deficit case, for every $r=1,2,3$ we have
\begin{equation}\label{apx:eq:local-positive-bound}
 \pos{x_{r,i}-2}\le\Delta_i,
 \qquad
 \pos{2-x_{r,i}}\le2\Delta_i.
\end{equation}
Thus \eqref{apx:eq:local-positive-bound} holds for all $\Delta_i\ge 0$.

Let
\begin{equation}\label{apx:eq:Delta-Q}
 \Delta_Q=\sum_{i=1}^n\Delta_i.
\end{equation}
Thus $\Delta_Q$ is the total deficit on private items.  Summing \eqref{apx:eq:local-positive-bound} gives, for $r=1,2,3$,
\begin{equation}\label{apx:eq:global-filler-bounds}
 \sum_i\pos{x_{r,i}-2}\le\Delta_Q,
 \qquad
 \sum_i\pos{2-x_{r,i}}\le2\Delta_Q.
\end{equation}

\paragraph{How to read the deficit equations.}
When $\Delta_i=0$, group $i$ has the unique perfect filler multiplicities $(2,2,2,1,\ldots,1)$.  When $\Delta_i>0$, the proof does not attempt to enumerate every perturbed integral solution.  It uses the simpler fact that $x_{1,i},x_{2,i},x_{3,i}$ remain between zero and three.  Since a defective group contributes at least one unit to $\Delta_i$, this bounded range yields the uniform positive-part estimates in \eqref{apx:eq:local-positive-bound}.  The argument is intentionally coarse but stable.

\paragraph{A lower bound on the number of main hyperedges.}
Let $h$ be the number of selected main hyperedges and put
\[
 \Delta_{P_1}=\sum_{i=1}^n\delta(P_{1,i}).
\]
Every main hyperedge contains one $P_1$-port, while the $F_1$ filler hyperedges account for all selected filler incidences at $P_1$-ports.  Therefore
\[
 h+\sum_i x_{1,i}=3n-\Delta_{P_1}.
\]
Rearranging and using \eqref{apx:eq:global-filler-bounds} gives
\begin{align}
 h
 &=n-\Delta_{P_1}+\sum_i(2-x_{1,i})\notag\\
 &\ge n-\Delta_{P_1}-\sum_i\pos{x_{1,i}-2}\notag\\
 &\ge n-\Delta_{P_1}-\Delta_Q.\label{apx:eq:h-lower}
\end{align}

\paragraph{Bounding repeated-value overuse.}
For a numerical value $\alpha$, let $r_A(\alpha)$ be its multiplicity in the source multiset $A$, and let $\mu_A(\alpha)$ be the number of selected main hyperedges whose $P_1$-incidence represents $\alpha$.  By the value-locality statement in \cref{apx:lem:target-structure}(c), every selected $F_1$ filler incidence attached to a private group with $a_i=\alpha$ uses a $P_1$-port of the same value.  Capacity over all $P_1$-ports of value $\alpha$ gives
\[
 \mu_A(\alpha)+\sum_{i:a_i=\alpha}x_{1,i}\le3r_A(\alpha).
\]
Subtracting $r_A(\alpha)$ and taking positive parts yields
\begin{equation}\label{apx:eq:A-overuse-one-value}
 \pos{\mu_A(\alpha)-r_A(\alpha)}
 \le\sum_{i:a_i=\alpha}\pos{2-x_{1,i}}.
\end{equation}
Therefore the total $A$-value overuse satisfies
\begin{align}
 E_A
 &:=\sum_\alpha\pos{\mu_A(\alpha)-r_A(\alpha)}\notag\\
 &\le\sum_i\pos{2-x_{1,i}}\notag\\
 &\le2\Delta_Q.\label{apx:eq:EA-bound}
\end{align}
The identical argument for the pairs $(P_2,F_2)$ and $(P_3,F_3)$ gives
\begin{equation}\label{apx:eq:EBC-bound}
 E_B\le2\Delta_Q,
 \qquad
 E_C\le2\Delta_Q.
\end{equation}
Thus
\begin{equation}\label{apx:eq:total-overuse}
 E_A+E_B+E_C\le6\Delta_Q.
\end{equation}

\paragraph{Why repeated values require a repair step.}
A main hyperedge records numerical values, whereas a source matching uses labeled occurrences.  If all values were distinct, incidence capacity at the ports would already give a direct injection.  With repeated values, defects may redistribute filler uses among equal-valued ports and allow too many main hyperedges to request the same value.  The quantities $E_A,E_B,E_C$ measure exactly this excess.  Bounding them is what turns a value-level decoding into an occurrence-level matching.

\paragraph{Repairing the main hyperedges.}
Starting with the selected main hyperedges, repeatedly choose a role-value pair whose current multiplicity exceeds its multiplicity in the corresponding source multiset, and delete one main hyperedge containing that role-value pair.  A deletion never increases any overuse and reduces the total overuse by at least one.  Hence at most $E_A+E_B+E_C$ main hyperedges are deleted.

By \eqref{apx:eq:h-lower} and \eqref{apx:eq:total-overuse}, the number of surviving main hyperedges is at least
\begin{align}
 h-(E_A+E_B+E_C)
 &\ge n-\Delta_{P_1}-\Delta_Q-6\Delta_Q\notag\\
 &=n-\Delta_{P_1}-7\Delta_Q\notag\\
 &\ge n-7\Delta.\label{apx:eq:surviving-main}
\end{align}
The final inequality holds because all deficits are nonnegative and $\Delta$ includes both $\Delta_{P_1}$ and $\Delta_Q$.  By \eqref{apx:eq:total-deficit},
\begin{equation}\label{apx:eq:n-21d}
 n-7\Delta=n-21d.
\end{equation}

In the surviving family, every numerical value is used no more often in each role than it occurs in the corresponding source multiset.  Equal-valued main incidences can therefore be assigned injectively to labeled source occurrences of that value.  Every surviving main hyperedge satisfies the source target equation by \cref{apx:lem:target-structure}(a).  The assigned triples are pairwise source-item-disjoint and form a feasible \NthreeDM matching.  If $n-21d<0$, the empty matching gives the stated maximum with zero.  Classification, counting, deletion, and assignment are all polynomial-time operations.
\end{proof}

\paragraph{Reading the constant $21$.}
The proof first obtains $n-7\Delta$ surviving main triples and then uses the exact identity $\Delta=3d$.  The factor seven covers one unit of port deficit or filler surplus and up to six units of repeated-value overuse; the factor three converts missing triples into missing incidences.  The bound is not claimed to be tight.  Its purpose is to be uniform, linear, and independent of the numerical values and the instance size.

\section{Gap transfer and approximation consequences}\label{apx:sec:gap-transfer}

\begin{proof}[Proof of \cref{apx:thm:main-gap}]
Take an instance from \cref{apx:thm:n3dm-gap} with $n$ items in each source class and apply the symmetric construction.  The output has $M=13n$ items in each class.

If the source is a ``yes'' instance, it has a perfect numerical matching, so \cref{apx:lem:forward-partial} gives a perfect symmetric matching and
\[
 \OPT_{\mathrm{SN}}=13n=M.
\]

Suppose the source is a ``no'' instance, so
\[
 \OPT_{\mathrm N}\le(1-\varepsilon_{\mathrm N})n.
\]
If the symmetric instance had a matching of size $13n-d$ with
\[
 d<\frac{\varepsilon_{\mathrm N}}{21}n,
\]
then \cref{apx:lem:defect-stability} would recover a source matching of size
\[
 n-21d>(1-\varepsilon_{\mathrm N})n,
\]
a contradiction.  Therefore
\[
 d\ge\frac{\varepsilon_{\mathrm N}}{21}n.
\]
Since $M=13n$,
\begin{align*}
 \OPT_{\mathrm{SN}}
 &\le13n-\frac{\varepsilon_{\mathrm N}}{21}n\\
 &=\left(1-\frac{\varepsilon_{\mathrm N}}{273}\right)M.
\end{align*}
Thus one may take
\begin{equation}\label{apx:eq:epsilon-SN}
 \varepsilon_{\mathrm{SN}}
 =\frac{\varepsilon_{\mathrm N}}{273}
 =\frac{\gamma}{3{,}276}.
\end{equation}

It remains to verify the numerical size.  The source instances from \cref{apx:thm:n3dm-gap} have $n=O(q)$ and largest number $O(q^3)$.  In the symmetric construction, the role coordinate is constant, the source coordinate is $O(q^3)$, and the two group coordinates are $O(n)$ and $O(n^2)$.  The shifts and the no-carry base $L$ in \eqref{apx:eq:mixed-radix} may therefore be chosen as $O(q^3)$.  Each encoded item weight and the target are $O(L^4)=O(q^{12})$, and there are $O(q)$ items.  Thus the unary output length and the construction time are polynomial in $q$.
\end{proof}

\paragraph{How to read the final gap calculation.}
A source ``no'' instance is missing at least an $\varepsilon_{\mathrm N}$ fraction of its $n$ triples.  If the symmetric defect were smaller than $\varepsilon_{\mathrm N}n/21$, the stability lemma would recover too many source triples.  Dividing that forced defect by the symmetric upper bound $13n$ produces the additional factor $13$, and hence the denominator $21\cdot13=273$.  The earlier numerical compilers contribute the factor twelve already contained in $\varepsilon_{\mathrm N}=\gamma/12$.

\begin{proof}[Proof of \cref{apx:cor:no-ptas}]
Suppose a PTAS exists.  Run it with an accuracy parameter $\delta$ satisfying $0<\delta<\varepsilon_{\mathrm{SN}}$.  On a ``yes'' instance of \cref{apx:thm:main-gap}, it returns a matching of size at least $(1-\delta)M>(1-\varepsilon_{\mathrm{SN}})M$.  On a ``no'' instance, no matching exceeds $(1-\varepsilon_{\mathrm{SN}})M$.  Comparing the returned value with this threshold solves the \NP-hard promise problem in polynomial time.  Hence no PTAS exists unless $\mathrm P=\mathrm{NP}$.
\end{proof}

\begin{proposition}[\APX membership]\label{apx:prop:apx-membership}
\MaxSNthreeDM has a polynomial-time $3$-approximation.
\end{proposition}

\begin{proof}
Enumerate all class-transversal triples whose weights sum to the target, and greedily construct any maximal family $\mathcal M_G$ of pairwise item-disjoint legal triples.  Let $\mathcal M^\star$ be an optimum matching.  By maximality, every triple in $\mathcal M^\star$ intersects at least one triple in $\mathcal M_G$; charge it to one such greedy triple.  A greedy triple has three items, and because $\mathcal M^\star$ is a matching, at most one optimum triple can contain each item.  Thus at most three optimum triples are charged to a greedy triple, giving
\[
 |\mathcal M^\star|\le3|\mathcal M_G|.
\]
There are at most $M^3$ candidate triples, so the algorithm is polynomial.
\end{proof}

\begin{remark}[Sharper approximation guarantee]
The preceding maximal-matching argument is included because it gives a
short, self-contained proof that \MaxSNthreeDM belongs to \APX.
A stronger quantitative guarantee follows from the approximation
literature for unweighted $3$-Set Packing.  Indeed, construct the
$3$-uniform hypergraph whose vertices are the labeled item occurrences
and whose hyperedges are precisely the class-transversal triples whose
weights sum to the target.  Feasible \MaxSNthreeDM solutions are exactly
matchings in this hypergraph.  Therefore, for every fixed
$\varepsilon>0$, the $(4/3+\varepsilon)$-approximation algorithms of
Cygan~\citep{Cygan2013} and F\"urer and Yu~\citep{FurerYu2014} apply
directly.  This sharpens the approximation factor but is not needed for
membership in \APX; the elementary factor-$3$ argument already proves
that classification.
\end{remark}

\begin{corollary}\label{apx:cor:gap-apx}
\MaxSNthreeDM belongs to \APX, has a unary perfect-completeness gap, and has no PTAS unless $\mathrm P=\mathrm{NP}$.
\end{corollary}

\begin{remark}[Relationship to Part III]\label{apx:rem:reduction-notion}
The proof in this part gives a perfect-completeness constant-gap reduction and therefore a fixed approximation threshold and a no-PTAS theorem.  It does not assert the actual-optimum error inequality required of an L-reduction.  Its safe standalone statement is:
\begin{quote}
There is a constant $\varepsilon>0$ such that it is \NP-hard to distinguish perfect \MaxSNthreeDM instances from instances with optimum at most $(1-\varepsilon)$ of perfect; consequently, \MaxSNthreeDM has no PTAS unless $\mathrm P=\mathrm{NP}$.
\end{quote}
Part~III later changes the construction by imposing one-live-port separation in order to obtain actual-optimum error transfer.  The present section should be cited for the broader grouped defect-stability theorem and the direct perfect-completeness promise gap; the exact division of labor is summarized in \cref{tab:partII-partIII}.
\end{remark}

\section{Fully worked examples}\label{apx:sec:worked-examples}

The purpose of the following examples is to make the new optimization layer visible with the smallest possible arithmetic.  They do not instantiate Petrank's asymptotic hard-gap family; a one-group example cannot exhibit an asymptotic constant gap.  Instead, they isolate the second transformation used in the proof: a source \MaxNthreeDM instance with $n=1$ is converted into a symmetric instance with $13$ occurrences in each labeled class.  The examples show exactly how a perfect source triple adds the thirteenth symmetric triple, and exactly what remains when that source triple is unavailable.

\subsection{Construction used in both examples}\label{apx:subsec:example-common}

For $n=1$, suppress the group index and write
\[
 S=\{P_1,P_2,P_3,Q_1,\ldots,Q_{10}\}.
\]
The symmetric output classes are $S^X,S^Y,S^Z$.  The role codes $\rho_R$ and group coefficients $h_R$ are those in \cref{apx:tab:symmetric-weights}.  Since the only group index is $i=1$, the last two vector coordinates of an occurrence $R$ are both $h_R$.  Thus
\[
 \omegavec(R)=(\rho_R,\eta_R,h_R,h_R),
 \qquad
 \Kvec=(300,3t,0,0).
\]

The same incidence coloring will be used in both examples.  At the unlabeled level, take two parallel copies of $F_1,F_2,F_3$, one copy of $F_4,\ldots,F_9$, and, in the ``yes'' instance only, one copy of $F_0$.  The following table gives an explicit coloring of the incidences by $X,Y,Z$.  Reading each row from left to right therefore gives one legal class-transversal triple in the labeled output instance.

\begin{table}[htbp]
\centering
\caption{An explicit incidence coloring for the one-group examples.  The row $F_0$ is used only in the ``yes'' instance.}
\label{apx:tab:example-coloring}
\small
\renewcommand{\arraystretch}{1.08}
\begin{tabular}{@{}c c c c@{}}
\toprule
Hyperedge occurrence & item from $S^X$ & item from $S^Y$ & item from $S^Z$\\
\midrule
$F_0$   & $P_1^X$    & $P_2^Y$    & $P_3^Z$\\
$F_1^a$ & $Q_1^X$    & $P_1^Y$    & $Q_2^Z$\\
$F_1^b$ & $Q_2^X$    & $Q_1^Y$    & $P_1^Z$\\
$F_2^a$ & $P_2^X$    & $Q_3^Y$    & $Q_4^Z$\\
$F_2^b$ & $Q_3^X$    & $Q_4^Y$    & $P_2^Z$\\
$F_3^a$ & $P_3^X$    & $Q_5^Y$    & $Q_6^Z$\\
$F_3^b$ & $Q_6^X$    & $P_3^Y$    & $Q_5^Z$\\
$F_4$   & $Q_8^X$    & $Q_7^Y$    & $Q_1^Z$\\
$F_5$   & $Q_9^X$    & $Q_2^Y$    & $Q_7^Z$\\
$F_6$   & $Q_7^X$    & $Q_{10}^Y$ & $Q_3^Z$\\
$F_7$   & $Q_4^X$    & $Q_8^Y$    & $Q_{10}^Z$\\
$F_8$   & $Q_5^X$    & $Q_9^Y$    & $Q_8^Z$\\
$F_9$   & $Q_{10}^X$ & $Q_6^Y$    & $Q_9^Z$\\
\bottomrule
\end{tabular}
\end{table}

\paragraph{How to read \cref{apx:tab:example-coloring}.}
The superscripts denote physical output classes, not new numerical roles.  For example, the row $F_1^a$ uses the three physical items $Q_1^X,P_1^Y,Q_2^Z$ and therefore has one item from each labeled class.  The row $F_1^b$ realizes a second occurrence of the same underlying role pattern with the other available copies of $P_1,Q_1,Q_2$.  Across all thirteen rows, every role appears exactly once in each of the three class columns.  If the $F_0$ row is deleted, every private $Q$-copy is still used, while precisely $P_1^X,P_2^Y,P_3^Z$ remain unused.

\subsection{A source-to-output ``yes'' instance}\label{apx:subsec:example-yes}

\subsubsection{Step 1: the source instance}

Take
\[
 A=\{1\},\qquad B=\{2\},\qquad C=\{3\},\qquad t=6.
\]
The unique source triple satisfies
\[
 1+2+3=6,
\]
so $\OPT_{\mathrm N}=1$.

The unshifted source-coordinate values are
\[
\begin{array}{c|rrrrrrrrrrrrr}
R&P_1&P_2&P_3&Q_1&Q_2&Q_3&Q_4&Q_5&Q_6&Q_7&Q_8&Q_9&Q_{10}\\
\hline
\eta_R&3&6&9&-6&21&21&-9&-9&18&-3&27&0&0.
\end{array}
\]
The smallest source coordinate is $-9$, and the smallest group coefficient is $-3$.  We therefore choose
\[
 D_1=9,\qquad D_2=D_3=3.
\]
The shifted target digits are
\[
 (\widehat K_0,\widehat K_1,\widehat K_2,\widehat K_3)
 =(300,45,9,9).
\]
The largest role digit is $261$, so $3\cdot261=783$.  Choosing
\[
 L=784
\]
therefore satisfies the no-carry requirement in every coordinate.

\subsubsection{Step 2: every encoded item weight}

With $L=784$, the encoding is
\[
 W(R)=\rho_R+(\eta_R+9)L+(h_R+3)L^2+(h_R+3)L^3.
\]
The target is
\[
 K^{\star}_{\mathrm{yes}}
 =300+45L+9L^2+9L^3
 =4{,}342{,}580{,}220.
\]
All thirteen encoded weights are shown in \cref{apx:tab:yes-weights}.

\begin{longtable}{@{}c r r r r@{}}
\caption{All encoded weights in the one-group ``yes'' instance.}\label{apx:tab:yes-weights}\\
\toprule
Role & $\rho_R$ & $\eta_R$ & shifted pair $(\eta_R+9,h_R+3)$ & $W(R)$\\
\midrule
\endfirsthead
\toprule
Role & $\rho_R$ & $\eta_R$ & shifted pair $(\eta_R+9,h_R+3)$ & $W(R)$\\
\midrule
\endhead
$P_1$&31&$3$&$(12,3)$&$1{,}447{,}524{,}319$\\
$P_2$&8&$6$&$(15,3)$&$1{,}447{,}526{,}648$\\
$P_3$&261&$9$&$(18,3)$&$1{,}447{,}529{,}253$\\
$Q_1$&184&$-6$&$(3,5)$&$2{,}412{,}527{,}336$\\
$Q_2$&85&$21$&$(30,1)$&$482{,}528{,}565$\\
$Q_3$&177&$21$&$(30,1)$&$482{,}528{,}657$\\
$Q_4$&115&$-9$&$(0,5)$&$2{,}412{,}524{,}915$\\
$Q_5$&23&$-9$&$(0,5)$&$2{,}412{,}524{,}823$\\
$Q_6$&16&$18$&$(27,1)$&$482{,}526{,}144$\\
$Q_7$&27&$-3$&$(6,4)$&$1{,}930{,}024{,}571$\\
$Q_8$&89&$27$&$(36,0)$&$28{,}313$\\
$Q_9$&188&$0$&$(9,4)$&$1{,}930{,}027{,}084$\\
$Q_{10}$&96&$0$&$(9,4)$&$1{,}930{,}026{,}992$\\
\bottomrule
\end{longtable}

\paragraph{How to read \cref{apx:tab:yes-weights}.}
The shifted pair records the second digit and the common value of the third and fourth digits.  For example,
\[
 W(Q_1)=184+3L+5L^2+5L^3=2{,}412{,}527{,}336.
\]
The small value of $W(Q_8)$ is not an error: its two group digits are zero after shifting, so only its role and source digits remain.  Positivity, rather than similarity of magnitudes, is the requirement.

\subsubsection{Step 3: all ten target identities}

The next two tables verify every role pattern at the vector and scalar levels.  Because $n=1$, the group and group-square sums are identical; the coordinate table records their common check, and the scalar table records the packed integer equality.

\begin{longtable}{@{}c l l l@{}}
\caption{Coordinate audit for the ten permitted patterns in the one-group ``yes'' instance.}\label{apx:tab:yes-pattern-audit}\\
\toprule
Pattern & role-code sum & source-coordinate sum & group-coefficient sum\\
\midrule
\endfirsthead
\toprule
Pattern & role-code sum & source-coordinate sum & group-coefficient sum\\
\midrule
\endhead
$F_0$&$31+8+261=300$&$3+6+9=18$&$0+0+0=0$\\
$F_1$&$31+184+85=300$&$3-6+21=18$&$0+2-2=0$\\
$F_2$&$8+177+115=300$&$6+21-9=18$&$0-2+2=0$\\
$F_3$&$261+23+16=300$&$9-9+18=18$&$0+2-2=0$\\
$F_4$&$184+27+89=300$&$-6-3+27=18$&$2+1-3=0$\\
$F_5$&$85+27+188=300$&$21-3+0=18$&$-2+1+1=0$\\
$F_6$&$177+27+96=300$&$21-3+0=18$&$-2+1+1=0$\\
$F_7$&$115+89+96=300$&$-9+27+0=18$&$2-3+1=0$\\
$F_8$&$23+89+188=300$&$-9+27+0=18$&$2-3+1=0$\\
$F_9$&$16+188+96=300$&$18+0+0=18$&$-2+1+1=0$\\
\bottomrule
\end{longtable}

Because $n=1$, the group-square coordinate has the same coefficient sums as the group coordinate in \cref{apx:tab:yes-pattern-audit}.  The corresponding ordinary encoded equalities are displayed separately in \cref{apx:tab:yes-scalar-audit}; separating the two tables prevents the large decimal representation from obscuring the coordinate logic.

\begingroup
\small
\begin{longtable}{@{}c >{\raggedright\arraybackslash}p{0.85\textwidth}@{}}
\caption{Ordinary encoded target-sum audit for the one-group ``yes'' instance.}\label{apx:tab:yes-scalar-audit}\\
\toprule
Pattern & encoded equality\\
\midrule
\endfirsthead
\toprule
Pattern & encoded equality\\
\midrule
\endhead
$F_0$&$1{,}447{,}524{,}319+1{,}447{,}526{,}648+1{,}447{,}529{,}253=K^\star_{\rm yes}$\\
$F_1$&$1{,}447{,}524{,}319+2{,}412{,}527{,}336+482{,}528{,}565=K^\star_{\rm yes}$\\
$F_2$&$1{,}447{,}526{,}648+482{,}528{,}657+2{,}412{,}524{,}915=K^\star_{\rm yes}$\\
$F_3$&$1{,}447{,}529{,}253+2{,}412{,}524{,}823+482{,}526{,}144=K^\star_{\rm yes}$\\
$F_4$&$2{,}412{,}527{,}336+1{,}930{,}024{,}571+28{,}313=K^\star_{\rm yes}$\\
$F_5$&$482{,}528{,}565+1{,}930{,}024{,}571+1{,}930{,}027{,}084=K^\star_{\rm yes}$\\
$F_6$&$482{,}528{,}657+1{,}930{,}024{,}571+1{,}930{,}026{,}992=K^\star_{\rm yes}$\\
$F_7$&$2{,}412{,}524{,}915+28{,}313+1{,}930{,}026{,}992=K^\star_{\rm yes}$\\
$F_8$&$2{,}412{,}524{,}823+28{,}313+1{,}930{,}027{,}084=K^\star_{\rm yes}$\\
$F_9$&$482{,}526{,}144+1{,}930{,}027{,}084+1{,}930{,}026{,}992=K^\star_{\rm yes}$\\
\bottomrule
\end{longtable}
\endgroup

\paragraph{How to read \cref{apx:tab:yes-pattern-audit,apx:tab:yes-scalar-audit}.}
Begin with the coordinate table.  Its second column verifies that the three roles are allowed to meet; its third verifies the source arithmetic; and its fourth verifies locality.  For this one-group example, the group-square check is identical.  Only after those logical coordinates have been checked should the scalar table be read.  The large decimal numbers do not create the identities; they merely package the already verified coordinates without carry.

\subsubsection{Step 4: the complete symmetric matching}

Use all thirteen rows of \cref{apx:tab:example-coloring}.  Each row has encoded sum $K^\star_{\mathrm{yes}}$ by \cref{apx:tab:yes-pattern-audit}.  In each output class, the thirteen entries are exactly
\[
 P_1,P_2,P_3,Q_1,\ldots,Q_{10}
\]
in some order.  Thus every physical item is used exactly once, and the thirteen rows form a perfect symmetric matching.  Consequently,
\[
 \OPT_{\mathrm{SN}}=13=12\cdot1+\OPT_{\mathrm N}.
\]

This affine equality is a fact about this particular one-group instance,
established by the exhaustive enumeration above; it is not an instance of
a general identity for the Part~II construction.  Part~II proves only the
forward inequality $\OPT_{\mathrm{SN}}\ge 12n+\OPT_{\mathrm N}$ together
with the defect-stability bound of
\cref{apx:lem:defect-stability}, which is measured from the perfect value
$13n$ rather than from $\OPT_{\mathrm{SN}}$.  Whether the equality holds
for every instance of the construction is precisely what remains open
here: with a single group there is no repeated-value interaction between
groups, so the reverse accounting loses nothing.  Establishing the exact
affine identity in general is the task of Part~III, which modifies the
construction by one-live-port separation.

The unlabeled degree ledger is equally transparent.  Each port lies in two filler edges and in $F_0$, so every port has degree three.  
Every private item has degree three by the following incidence ledger, which lists, for each private item, the filler hyperedges containing it:
\[
\begin{array}{c|l@{\qquad}c|l}
Q_1& F_1\ (\text{two copies}),\ F_4 & Q_2& F_1\ (\text{two copies}),\ F_5\\
Q_3& F_2\ (\text{two copies}),\ F_6 & Q_4& F_2\ (\text{two copies}),\ F_7\\
Q_5& F_3\ (\text{two copies}),\ F_8 & Q_6& F_3\ (\text{two copies}),\ F_9\\
Q_7& F_4,\ F_5,\ F_6 & Q_8& F_4,\ F_7,\ F_8\\
Q_9& F_5,\ F_8,\ F_9 & Q_{10}& F_6,\ F_7,\ F_9.
\end{array}
\]
The reverse proof therefore recovers the unique main edge $F_0=(P_1,P_2,P_3)$ and decodes it as the source equality $1+2+3=6$.

\subsection{A source-to-output ``no'' instance}\label{apx:subsec:example-no}

\subsubsection{Step 1: the source obstruction}

Now take
\[
 A=\{1\},\qquad B=\{2\},\qquad C=\{4\},\qquad t=6.
\]
The unique class-transversal source triple has sum
\[
 1+2+4=7\ne6,
\]
so $\OPT_{\mathrm N}=0$.

The source-coordinate values are
\[
\begin{array}{c|rrrrrrrrrrrrr}
R&P_1&P_2&P_3&Q_1&Q_2&Q_3&Q_4&Q_5&Q_6&Q_7&Q_8&Q_9&Q_{10}\\
\hline
\eta_R&3&6&12&-6&21&22&-10&-11&17&-4&28&1&0.
\end{array}
\]
Choose
\[
 D_1=11,\qquad D_2=D_3=3,
\]
so the shifted target digits are $(300,51,9,9)$.  The same base $L=784$ remains valid.

\subsubsection{Step 2: every encoded item weight}

The ordinary target is
\[
 K^{\star}_{\mathrm{no}}
 =300+51L+9L^2+9L^3
 =4{,}342{,}584{,}924.
\]
The thirteen item weights are listed in \cref{apx:tab:no-weights}.

\begin{longtable}{@{}c r r r r@{}}
\caption{All encoded weights in the one-group ``no'' instance.}\label{apx:tab:no-weights}\\
\toprule
Role & $\rho_R$ & $\eta_R$ & shifted pair $(\eta_R+11,h_R+3)$ & $W(R)$\\
\midrule
\endfirsthead
\toprule
Role & $\rho_R$ & $\eta_R$ & shifted pair $(\eta_R+11,h_R+3)$ & $W(R)$\\
\midrule
\endhead
$P_1$&31&$3$&$(14,3)$&$1{,}447{,}525{,}887$\\
$P_2$&8&$6$&$(17,3)$&$1{,}447{,}528{,}216$\\
$P_3$&261&$12$&$(23,3)$&$1{,}447{,}533{,}173$\\
$Q_1$&184&$-6$&$(5,5)$&$2{,}412{,}528{,}904$\\
$Q_2$&85&$21$&$(32,1)$&$482{,}530{,}133$\\
$Q_3$&177&$22$&$(33,1)$&$482{,}531{,}009$\\
$Q_4$&115&$-10$&$(1,5)$&$2{,}412{,}525{,}699$\\
$Q_5$&23&$-11$&$(0,5)$&$2{,}412{,}524{,}823$\\
$Q_6$&16&$17$&$(28,1)$&$482{,}526{,}928$\\
$Q_7$&27&$-4$&$(7,4)$&$1{,}930{,}025{,}355$\\
$Q_8$&89&$28$&$(39,0)$&$30{,}665$\\
$Q_9$&188&$1$&$(12,4)$&$1{,}930{,}029{,}436$\\
$Q_{10}$&96&$0$&$(11,4)$&$1{,}930{,}028{,}560$\\
\bottomrule
\end{longtable}

\paragraph{How to read \cref{apx:tab:no-weights}.}
The role and locality coordinates are identical to the ``yes'' instance; only the source values and the source-coordinate shift change.  The table should therefore be compared row by row with \cref{apx:tab:yes-weights}.  The purpose of the comparison is not to find similar decimal magnitudes, but to see that all filler identities survive while the main identity records the failed source sum.

\subsubsection{Step 3: every filler identity survives, but the main identity fails}

For the nine filler patterns, the source-coordinate sums are
\[
\begin{array}{c|c@{\qquad}c|c}
F_1&3-6+21=18&F_2&6+22-10=18\\
F_3&12-11+17=18&F_4&-6-4+28=18\\
F_5&21-4+1=18&F_6&22-4+0=18\\
F_7&-10+28+0=18&F_8&-11+28+1=18\\
F_9&17+1+0=18.&&
\end{array}
\]
Their role and group checks are unchanged, so every filler pattern still has encoded sum $K^\star_{\mathrm{no}}$.  Numerically,
\[
\begin{array}{c|l}
F_1&1{,}447{,}525{,}887+2{,}412{,}528{,}904+482{,}530{,}133=K^\star_{\rm no}\\
F_2&1{,}447{,}528{,}216+482{,}531{,}009+2{,}412{,}525{,}699=K^\star_{\rm no}\\
F_3&1{,}447{,}533{,}173+2{,}412{,}524{,}823+482{,}526{,}928=K^\star_{\rm no}\\
F_4&2{,}412{,}528{,}904+1{,}930{,}025{,}355+30{,}665=K^\star_{\rm no}\\
F_5&482{,}530{,}133+1{,}930{,}025{,}355+1{,}930{,}029{,}436=K^\star_{\rm no}\\
F_6&482{,}531{,}009+1{,}930{,}025{,}355+1{,}930{,}028{,}560=K^\star_{\rm no}\\
F_7&2{,}412{,}525{,}699+30{,}665+1{,}930{,}028{,}560=K^\star_{\rm no}\\
F_8&2{,}412{,}524{,}823+30{,}665+1{,}930{,}029{,}436=K^\star_{\rm no}\\
F_9&482{,}526{,}928+1{,}930{,}029{,}436+1{,}930{,}028{,}560=K^\star_{\rm no}\\
\end{array}
\]

By contrast, the main source-coordinate sum is
\[
 3+6+12=21=18+3.
\]
Since the three port roles and group digits are otherwise correct,
\begin{align*}
 W(P_1)+W(P_2)+W(P_3)
 &=K^\star_{\mathrm{no}}+3L\\
 &=4{,}342{,}584{,}924+2{,}352\\
 &=4{,}342{,}587{,}276.
\end{align*}
Thus $F_0$ misses the target in exactly the source digit and cannot be used.

\paragraph{How to read the failed main sum.}
The excess $2{,}352$ is not an opaque decimal discrepancy: it is $3L$.  The source-coordinate sum exceeds its target by three, and the no-carry encoding places one unit of source-coordinate error at value $L$.  The failure therefore records precisely the original source obstruction $1+2+4\ne6$.

\subsubsection{Step 4: a matching of size twelve}

Delete the $F_0$ row from \cref{apx:tab:example-coloring} and use the remaining twelve rows.  Every one of them is a target-sum triple by the calculations above.  All thirty copies of the private items $Q_1,\ldots,Q_{10}$ are used, and each port has two used copies.  The only unused physical items are
\[
 P_1^X,\qquad P_2^Y,\qquad P_3^Z.
\]
Hence
\[
 \OPT_{\mathrm{SN}}\ge12.
\]

\subsubsection{Step 5: why a thirteenth triple is impossible}

Suppose, for contradiction, that thirteen disjoint target triples existed.  There are thirteen underlying occurrences and each has only three physical copies.  Thirteen triples contain thirty-nine incidences, exactly the total capacity $13\cdot3$.  After forgetting labels, every underlying occurrence would therefore have degree three.

The role certificate permits only $F_0,\ldots,F_9$, and $F_0$ is arithmetically illegal in this instance.  Hence all thirteen hyperedges would have to be filler hyperedges.  Let $x_r$ be the multiplicity of $F_r$.  Degree three at the ten private items gives the exact system
\[
\begin{aligned}
 x_1+x_4&=3,&x_1+x_5&=3,&x_2+x_6&=3,&x_2+x_7&=3,\\
 x_3+x_8&=3,&x_3+x_9&=3,&x_4+x_5+x_6&=3,\\
 x_4+x_7+x_8&=3,&x_5+x_8+x_9&=3,&x_6+x_7+x_9&=3.
\end{aligned}
\]
Its unique nonnegative integral solution is
\[
 (x_1,x_2,x_3,x_4,x_5,x_6,x_7,x_8,x_9)
 =(2,2,2,1,1,1,1,1,1),
\]
whose total is only
\[
 2+2+2+1+1+1+1+1+1=12.
\]
This contradicts the assumed thirteen filler hyperedges.  Therefore
\[
 \OPT_{\mathrm{SN}}=12=12\cdot1+\OPT_{\mathrm N}.
\]

As before, this equality is verified by enumeration for this instance and is not claimed as a general identity.

\subsubsection{Step 6: the defect ledger in the ``no'' instance}

Here $n=1$, $v=12$, and
\[
 d=13n-v=1.
\]
After forgetting labels, the three ports have degree two and the ten private items have degree three.  Thus
\[
 \delta(P_1)=\delta(P_2)=\delta(P_3)=1,
 \qquad
 \delta(Q_\ell)=0\quad(\ell=1,\ldots,10),
\]
and
\[
 \Delta=3=3d,
 \qquad
 \Delta_Q=0.
\]
The defect-stability lemma gives the guaranteed source matching size
\[
 \max\{0,n-21d\}=\max\{0,1-21\}=0,
\]
which is correct.  The bound is deliberately conservative on this tiny instance; its role in the theorem is to remain linear and uniform over arbitrarily large instances.

\subsection{What the two examples establish}\label{apx:subsec:example-lessons}

The examples make three points visible.  First, the twelve filler triples are a fixed local scaffold: they exist in both instances and consume all private copies.  Second, the thirteenth triple is present if and only if the source target equation holds.  Third, a missing symmetric triple creates three missing incidences, which is the identity $\Delta=3d$ at the start of the quantitative reverse proof.  In large instances the stability lemma shows that these local incidence defects can spoil only linearly many recovered source triples, allowing a constant source gap to survive the symmetric construction.

\section{Lessons learned}\label{apx:sec:lessons-learned}

Strong \NP-hardness and approximation hardness answer different questions.  An exact reduction may preserve the difference between perfection and one missing triple while allowing that relative difference to vanish.  The general lesson is to identify a defect variable early and prove that every transformation maps small destination defect to small source defect by a constant independent of instance size.

The symmetric construction also illustrates how an exact saturation proof can be made stable.  Replace each degree equality by an equality with a nonnegative deficit, solve the zero-defect system first, and then bound deviations using positive parts rather than attempting a complete case analysis.  If the decoded objects carry values rather than labeled occurrences, add an explicit repair argument for multiplicity overuse.  These steps turn a qualitative reverse implication into a quantitative one while leaving the original gadget unchanged.

A final methodological lesson concerns exposition.  Large mixed-radix integers should be treated as a packaging device, not as the source of the logic.  The proof is more transparent when role admissibility, source arithmetic, locality, incidence counting, and recoloring are explained separately and then composed.  The one-group examples demonstrate this separation: the twelve filler triples are identical in the yes and no cases, and the source equation controls only the presence or absence of the thirteenth main triple.

\section{Conclusions}\label{apx:sec:conclusions}

\subsection*{Summary of results}
Part~II proves a unary perfect-completeness gap for maximum-cardinality Symmetric Numerical Three-Dimensional Matching.  Petrank's bounded three-dimensional matching gap is transferred through two polynomially bounded numerical compilers to \MaxNthreeDM.  The Part~I symmetric construction maps an $n$-item source instance to three identical classes of size $13n$, and the new defect-stability analysis shows that $13n-d$ symmetric triples yield at least $n-21d$ source triples.  Hence, for $\varepsilon_{\mathrm{SN}}=\gamma/3{,}276>0$, it is \NP-hard to distinguish a perfect symmetric matching from one of value at most $(1-\varepsilon_{\mathrm{SN}})$ times perfect, even under unary encoding.  A maximal legal triple matching is a $3$-approximation, so the problem belongs to \APX; independently, the gap excludes a PTAS unless $\mathrm P=\mathrm{NP}$.  The worked examples verify every encoded weight and isolate the one main triple that distinguishes the yes and no cases.

\subsection*{Modeling and computational implications}
Identical marginal weight lists and one common target do not merely preserve strong \NP-hardness; they also exclude a polynomial-time approximation scheme unless $\mathrm P=\mathrm{NP}$.  Near-complete solutions therefore cannot be guaranteed arbitrarily closely in general.  The proof identifies the structural obstacles: filler obligations consume most copies, missing assignments disturb incidence balances, and repeated values create occurrence-level multiplicity conflicts.  Symmetric formulations should expose roles, locality, degree balances, and interchangeable copies explicitly.  Nevertheless, the elementary $3$-approximation shows that constant-factor performance remains available through the underlying three-set-packing interpretation.

\subsection*{Future research possibilities}
The constants $21$ and $\gamma/3{,}276$ are conservative, and sharper analysis of the local deficit system may improve them.  Part~III resolves the standard L-reduction question on a deliberately one-live-port-separated image family; an important remaining question is whether the broader Part~II construction itself admits equally exact actual-optimum error transfer without sacrificing its direct gap formulation.  Stronger approximation or parameterized algorithms may be possible for bounded value multiplicity, few distinct weights, restricted compatibility hypergraphs, few defective groups, or instances close to the canonical filler scaffold.  Computational comparisons of greedy, local-search, integer-programming, constraint-programming, and exact-cover methods would help determine which structural restrictions make symmetric instances practically tractable.

\clearpage

\clearpage
\part{A Standard L-Reduction for Maximum Symmetric Numerical Three-Dimensional Matching}\label{part:lred}
\setcounter{theorem}{0}
\renewcommand{\thesection}{\arabic{section}}
\renewcommand{\theHsection}{lred.sec.\arabic{section}}

\noindent\hyperlink{maincontents}{\small\textit{Return to main contents}}\par\medskip
\section*{Proof navigation for Part III}
\addcontentsline{toc}{section}{Proof navigation for Part III}

\textbf{What this part adds.}
Part~II proves a direct perfect-completeness gap from a broader grouped construction.  Part~III changes the reduction so that the destination objective has an exact fixed offset and every destination error unit can be charged once.  This is the additional structure required for a standard L-reduction and for the APX-hardness claim made in this part.

\textbf{Forward direction.}
A source matching chooses disjoint three-dimensional edges.  The pair compiler places each represented $(u_i,v_j)$ pair either in fallback mode or, when its source edge is selected, in a two-triple complete mode.  Padding and one-live-port separation then distribute the three numerical classes among disjoint group families.  The 13-role symmetric gadget contributes twelve local filler triples per group and one main triple per selected numerical triple.  Recoloring restores the three identical labeled output classes.

\textbf{Reverse direction.}
A symmetric solution is first decoded to a separated numerical matching.  The new local certificate bounds useful port capacity by missing filler triples, and one-live separation prevents the same local deficit from being charged in several roles.  Complete pairs then decode to disjoint source edges.  The two decoders compose to give the actual-optimum error inequality with $\beta=1$.

\textbf{Suggested reading paths.}
Readers interested mainly in the theorem may read \cref{lred:sec:definitions,lred:sec:roadmap,lred:sec:final-lred}.  Readers auditing the new idea should focus on one-live-port separation in \cref{lred:sec:separation}, the local filler certificate in \cref{lred:sec:local-certificate}, and the disjoint charging argument in \cref{lred:sec:affine-symmetric}.  The detailed ``yes'' and ``no'' instances are in \cref{lred:sec:examples}; the independent calculation audit is in \cref{lred:sec:audit}.

\section{Optimization problems and L-reductions}\label{lred:sec:definitions}

Repeated values are allowed throughout.  An item is a labeled occurrence, not merely its numerical value.

\begin{definition}[Bounded Maximum Three-Dimensional Matching]\label{lred:def:max3dm}
An instance is a three-partite $3$-uniform hypergraph
\[
  H=(U,V,W;E),\qquad E\subseteq U\times V\times W,
\]
in which every vertex belongs to at most three edges.  A feasible solution is a set of pairwise vertex-disjoint edges.  Its maximum cardinality is denoted $\OPT_{3\mathrm{DM}}(H)$.  We write $p=|E|$.
\end{definition}

This bounded problem is a standard APX-hard source.  Kann proved bounded Maximum Three-Dimensional Matching MAX SNP-complete, and later explicit results establish approximation hardness even when every element has only two occurrences \citep{Kann1991,ChlebikChlebikova2006}.  Thus the degree-three restriction used here is more than sufficient.

\begin{definition}[Maximum Numerical Three-Dimensional Matching]\label{lred:def:maxn3dm}
An instance consists of three pairwise disjoint item sets $X_N,Y_N,Z_N$, one positive integer weight for each item, and a positive integer target $T_N$.  A feasible solution is a collection of pairwise item-disjoint triples in $X_N\times Y_N\times Z_N$ whose weights sum to $T_N$.  The objective is to maximize the number of selected triples; its optimum is $\OPT_{\mathrm N}$.
\end{definition}

\begin{definition}[Maximum Symmetric Numerical Three-Dimensional Matching]\label{lred:def:maxsn3dm}
An instance of \MaxSNthreeDM has three disjoint labeled classes $S^X,S^Y,S^Z$ that are copies of one weighted occurrence set $S$.  Corresponding copies have equal weights.  A feasible solution is a collection of pairwise item-disjoint triples, each containing one item from each class and having the common target sum.  The optimum is denoted $\OPT_{\mathrm{SN}}$.
\end{definition}

\begin{definition}[L-reduction]\label{lred:def:lred}
Let $A$ and $B$ be maximization problems.  An L-reduction consists of polynomial-time maps $f$ and $g$ and constants $\alpha,\beta>0$ such that, for every instance $I$ of $A$ and every feasible solution $Y$ of $f(I)$,
\begin{align}
  \OPT_B(f(I))&\le \alpha\,\OPT_A(I),\label{lred:eq:lred1}\\
  \OPT_A(I)-\val_A(g(I,Y))
  &\le \beta\bigl(\OPT_B(f(I))-\val_B(Y)\bigr).\label{lred:eq:lred2}
\end{align}
The second map may use both the source instance and the destination solution.  For maximization problems, all displayed differences are nonnegative, so absolute values are unnecessary.
\end{definition}

L-reductions were introduced by Papadimitriou and Yannakakis as approximation-preserving transformations \citep{PapadimitriouYannakakis1991}.  The second inequality is the important distinction from a mere perfect-completeness gap: if the destination error is zero, the decoded source error must also be zero.

\begin{theorem}[Main theorem]\label{lred:thm:main}
There is an L-reduction from \MaxThreeDMThree to unary \MaxSNthreeDM with
\[
  \boxed{\alpha=764,\qquad \beta=1.}
\]
Consequently, \MaxSNthreeDM is APX-hard under a standard L-reduction and admits a PTAS-reduction from \MaxThreeDMThree.  Since \MaxSNthreeDM has a polynomial-time $3$-approximation, it is APX-complete under this reduction convention.
\end{theorem}

\begin{remark}[Why the source's approximation status matters]\label{lred:rem:source-no-ptas}
An L-reduction alone is a PTAS-transfer mechanism; it is not by itself a proof that the destination lacks a PTAS.  From \eqref{lred:eq:lred1}--\eqref{lred:eq:lred2}, a $(1-\delta)$-approximate destination solution decodes to a source solution of value at least
\[
  (1-\alpha\beta\delta)\OPT_A.
\]
Thus a PTAS for the destination would give a PTAS for the source by choosing $\delta=\varepsilon/(\alpha\beta)$.  The no-PTAS consequence here is valid because \MaxThreeDMThree is a MAX SNP-complete bounded matching problem and has no PTAS unless $\mathrm P=\mathrm{NP}$ \citep{Kann1991,ChlebikChlebikova2006}.  With the displayed constants, $\alpha\beta=764$.
\end{remark}

\begin{remark}[Why Parts II and III are both retained]\label{lred:rem:relation-parts}
This part supplies a standard L-reduction, an exact affine optimum identity, and an error inequality for every destination solution on the one-live-port-separated image family.  Part~II supplies a direct perfect-completeness gap and a defect-stability lemma for the broader grouped inputs accepted by the original symmetric construction.  Their distinct scopes and guarantees are summarized in \cref{tab:partII-partIII}.  The comparison is dimensional rather than a total ordering, and neither theorem includes the other.
\end{remark}

\section{All artifacts of the proof}\label{lred:sec:artifacts}

\Cref{lred:tab:artifacts} is a catalogue of every object introduced by the reduction.  It is intended both as a reading guide and as a checklist against missing definitions.

\begingroup
\scriptsize
\setlength{\tabcolsep}{3pt}
\begin{longtable}{@{}p{0.07\textwidth}p{0.16\textwidth}p{0.17\textwidth}p{0.25\textwidth}p{0.27\textwidth}@{}}
\caption{Artifacts of the L-reduction.}\label{lred:tab:artifacts}\\
\toprule
Stage & Artifact & Number & Description & Job in the proof\\
\midrule
\endfirsthead
\toprule
Stage & Artifact & Number & Description & Job in the proof\\
\midrule
\endhead
1 & $H=(U,V,W;E)$ & source & Degree-at-most-three three-partite hypergraph & APX-hard optimization source\\
1 & $p,k^\star$ & two scalars & $p=|E|$ and $k^\star=\OPT_{3\mathrm{DM}}(H)$ & Size-to-optimum bound $p\le7k^\star$\\
2 & $\Pairset$ & $r$ pairs & Pairs $(u_i,v_j)$ occurring in at least one source edge & One local compiler gadget per represented pair\\
2 & $A_i,B_j,C_k$ & one per source vertex & Vertex items in the numerical compiler & Enforce disjoint use of $U,V,W$ vertices\\
2 & $L_{ij},F_{ij},R_{ij}$ & three per represented pair & Left, fallback, and right pair items & Permit either one fallback triple or two edge-representing triples\\
2 & $D_{ijk}$ & one per source edge & Edge item & Identifies the represented source edge\\
3 & Six coordinates & fixed dimension & Role, $u,u^2,v,v^2,w$ & Exclude unintended compiler triples and synchronize indices\\
3 & $D_0,\ldots,D_5,L_N$ & seven scalars & Digit shifts and no-carry base & Convert signed vectors to positive unary-polynomial integers\\
4 & Complete pair & at most $r$ & A pair using one left and one right triple & Decodes to one source edge\\
4 & $r+k^\star$ & objective identity & Exact numerical optimum & First affine layer of the reduction\\
5 & $m$ & class size & Maximum of the three compiler class sizes after padding & Common size before separation\\
5 & $\barrier=T_N+1$ & barrier value & Weight too large for a target triple & Creates isolated padding and separation occurrences\\
5 & One-live groups & $n=3m$ & Groups $(x,\barrier,\barrier)$, $(\barrier,y,\barrier)$, $(\barrier,\barrier,z)$ & At most one useful port per symmetric block\\
6 & $P_1,P_2,P_3$ & three roles per group & Ports carrying the separated numerical items & Interface between source triples and symmetric gadget\\
6 & $Q_1,\ldots,Q_{10}$ & ten roles per group & Private filler items & Regulate local incidence capacity\\
6 & $F_0,\ldots,F_9$ & ten role patterns & One main pattern and nine filler patterns & Exhaustive target-triple classification\\
6 & $\rho_R,\eta_{R,i},h_R$ & role data & Role code, source expression, locality coefficient & Enforce type, arithmetic, and group locality\\
6 & $x_{r,i},f_i$ & local counts & Multiplicity of filler pattern $F_r$ and total fillers in group $i$ & Measure useful port capacity against filler loss\\
7 & $E_X,E_Y,E_Z$ & three overuse totals & Excess demand for source values in the three roles & Bound repair deletions\\
7 & $12n-f$ & global filler deficit & Missing fillers relative to the canonical scaffold & Pays for all useful value overuse once\\
7 & $36m+r+k^\star$ & objective identity & Exact symmetric optimum & Affine relation needed by the L-reduction\\
8 & $\Gamma=36m+r$ & fixed offset & Destination value not carrying source objective & Cancels in the second L-inequality\\
8 & $f,g$ & two algorithms & Instance construction and solution decoder & Complete the L-reduction\\
8 & $764,1$ & constants & Conservative $\alpha$ and exact $\beta$ & Standard APX-hardness and PTAS transfer\\
\bottomrule
\end{longtable}
\endgroup

\paragraph{How to read the table.}
The construction contains two affine layers.  The pair compiler contributes the fixed offset $r$, and the symmetric scaffold contributes the fixed offset $36m$.  Neither offset carries source objective value.  The decoder subtracts these offsets in the only safe way: it never subtracts them algebraically from a possibly small solution and declares the result feasible; instead, it constructs feasible intermediate matchings and uses lower bounds with a positive-part truncation.

\section{Eight-step proof roadmap}\label{lred:sec:roadmap}

The proof is organized into the following eight explicit stages.

\begin{enumerate}[label=\textbf{Step \arabic*.},leftmargin=3.2em]
\item \textbf{Normalize and bound the source.}  Delete isolated vertices, handle the empty-edge case separately, and prove $p\le7k^\star$.
\item \textbf{Build the exact pair compiler.}  Introduce the vertex, pair, and edge items and specify the three intended numerical triple types.
\item \textbf{Audit and scalarize the compiler.}  Check all twelve class-respecting role combinations, prove the two-moment synchronization, and apply a six-digit no-carry encoding with polynomially bounded positive integers.
\item \textbf{Prove the compiler objective identity.}  Show $\OPT_{\mathrm N}=r+k^\star$ and give the polynomial complete-pair decoder.
\item \textbf{Equalize and separate the numerical classes.}  Pad to common size $m$, create $n=3m$ one-live groups, and prove that the numerical optimum is unchanged.
\item \textbf{Apply and strengthen the symmetric gadget.}  Restate the 13-role construction, partial recoloring, and the four explicit linear-combination certificates that imply $\pos{2-x_r}\le12-f_i$.
\item \textbf{Prove exact affine behavior and decode.}  Charge repeated-value overuse to disjoint one-live groups, repair the main triples, and obtain $\OPT_{\mathrm{SN}}=36m+r+k^\star$.
\item \textbf{Verify the L-reduction and unary bounds.}  Prove the two inequalities with $\alpha=764$, $\beta=1$, combine the decoders, handle the zero-optimum source, and derive APX and PTAS consequences.
\end{enumerate}

Each stage has an independently checkable output.  A failure in one layer cannot be concealed by a later mixed-radix encoding or by the outer copy recoloring.

\section{Step 1: source normalization and the size bound}\label{lred:sec:source}

Let $H=(U,V,W;E)$ be a nonempty instance of \MaxThreeDMThree.  Delete every isolated vertex.  This does not change the feasible edge sets or the optimum, and it gives
\[
  |U|\le p,\qquad |V|\le p,\qquad |W|\le p.
\]
Write
\[
  k^\star=\OPT_{3\mathrm{DM}}(H).
\]

\begin{lemma}[Edges are linearly bounded by the optimum]\label{lred:lem:pbound}
For every nonempty degree-at-most-three source instance,
\[
  p\le7k^\star.
\]
\end{lemma}

\begin{proof}
Take any maximal source matching $M$.  Every unselected edge intersects at least one edge of $M$, or else it could be added.  A selected edge $e=(u,v,w)$ intersects at most
\[
  (\deg(u)-1)+(\deg(v)-1)+(\deg(w)-1)\le 2+2+2=6
\]
other edges.  Thus $e$ accounts for itself and at most six unselected edges.  The $|M|$ selected edges account for all $p$ edges, so
\[
  p\le7|M|\le7k^\star.
\]
\end{proof}

\paragraph{Why the constant seven is conservative.}
Conflicting edges may be counted more than once, so the proof only overestimates $p$.  Tightness is irrelevant; the L-reduction needs a constant independent of the instance.  If one starts from the known degree-two hard subclass, the same argument gives $p\le4k^\star$ and improves the displayed $\alpha$, but we retain the degree-three statement and the previously advertised constant $764$.

\section{Steps 2--4: an exact compiler from bounded \ThreeDM to \MaxNthreeDM}\label{lred:sec:compiler}

\subsection{The pair gadgets and intended triples}\label{lred:subsec:pair-gadgets}

Index the source vertices as
\[
  U=\{u_1,\ldots,u_a\},\qquad
  V=\{v_1,\ldots,v_b\},\qquad
  W=\{w_1,\ldots,w_c\}.
\]
Define the represented pair set
\[
  \Pairset=\{(i,j):\text{there exists }k\text{ with }(u_i,v_j,w_k)\in E\},
  \qquad r=|\Pairset|.
\]
Clearly $r\le p$.

Create the three numerical classes
\begin{align}
 X_N&=\{A_i:i\in[a]\}\uplusm\{L_{ij}:(i,j)\in\Pairset\},\label{lred:eq:compiler-X}\\
 Y_N&=\{B_j:j\in[b]\}\uplusm\{C_k:k\in[c]\}
       \uplusm\{F_{ij}:(i,j)\in\Pairset\},\label{lred:eq:compiler-Y}\\
 Z_N&=\{R_{ij}:(i,j)\in\Pairset\}
       \uplusm\{D_{ijk}:(u_i,v_j,w_k)\in E\}.
       \label{lred:eq:compiler-Z}
\end{align}
The intended target triples are
\begin{align}
 T^{L}_{ij}&=(A_i,B_j,R_{ij}) &&((i,j)\in\Pairset),\label{lred:eq:left-triple}\\
 T^{R}_{ijk}&=(L_{ij},C_k,D_{ijk}) &&((u_i,v_j,w_k)\in E),\label{lred:eq:right-triple}\\
 T^{F}_{ij}&=(L_{ij},F_{ij},R_{ij}) &&((i,j)\in\Pairset).\label{lred:eq:fallback-triple}
\end{align}
A selected source edge is represented by the pair $T^L_{ij},T^R_{ijk}$.  An unused pair is disposed of by the single fallback triple $T^F_{ij}$.

\begin{figure}[H]
\centering
\begin{tikzpicture}[
 item/.style={draw,circle,minimum size=8mm,inner sep=1pt,font=\small},
 edge/.style={thick},
 lab/.style={font=\small,align=center}
]
\node[item] (A) at (0,1.2) {$A_i$};
\node[item] (B) at (1.6,1.2) {$B_j$};
\node[item] (R) at (3.2,1.2) {$R_{ij}$};
\node[item] (L) at (0,-1.0) {$L_{ij}$};
\node[item] (C) at (1.6,-1.0) {$C_k$};
\node[item] (D) at (3.2,-1.0) {$D_{ijk}$};
\node[item] (F) at (1.6,-2.5) {$F_{ij}$};
\draw[edge,blue!55!black] (A)--(B)--(R)--cycle;
\draw[edge,blue!55!black] (L)--(C)--(D)--cycle;
\draw[edge,gray!70] (L)--(F)--(R)--cycle;
\node[lab] at (1.6,2.0) {left triple $T^L_{ij}$};
\node[lab] at (1.6,-0.15) {right triple $T^R_{ijk}$};
\node[lab] at (3.95,-2.1) {fallback $T^F_{ij}$};
\end{tikzpicture}
\caption{One pair gadget.  A represented edge uses the blue left and right triples.  If the pair is not completed, the gray fallback triple contributes one baseline unit.  The shared items $L_{ij}$ and $R_{ij}$ prevent combining the fallback with either side of a completed pair.}
\label{lred:fig:pair-gadget}
\end{figure}
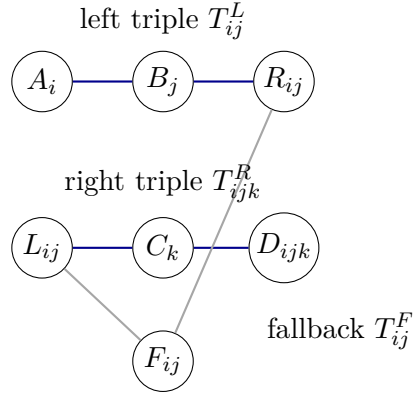

\paragraph{How to read \cref{lred:fig:pair-gadget}.}
The figure depicts item sharing, not a graph instance.  A numerical matching may select both blue triples because they are item-disjoint; together they identify the same pair $(i,j)$ and one edge $(i,j,k)$.  The fallback uses both $L_{ij}$ and $R_{ij}$, so selecting it blocks both blue alternatives.  This is why every pair contributes either one baseline triple or, when completed, two triples and one unit of source objective.

\subsection{Six signed coordinates}\label{lred:subsec:six-coordinates}

Use the coordinate order
\[
  (\lambda,u,u^2,v,v^2,w)
\]
and vector target $\mathbf 0$.  Assign the following signed vectors.

\begin{table}[H]
\centering
\caption{Six-coordinate pair compiler.}
\label{lred:tab:compiler-vectors}
\small
\renewcommand{\arraystretch}{1.15}
\begin{tabular}{@{}c c c@{}}
\toprule
Item & Class & Vector\\
\midrule
$A_i$ & $X_N$ & $(0,-i,-i^2,0,0,0)$\\
$L_{ij}$ & $X_N$ & $(1,i,i^2,j,j^2,0)$\\
$B_j$ & $Y_N$ & $(0,0,0,-j,-j^2,0)$\\
$C_k$ & $Y_N$ & $(1,0,0,0,0,k)$\\
$F_{ij}$ & $Y_N$ & $(-1,-2i,-2i^2,-2j,-2j^2,0)$\\
$R_{ij}$ & $Z_N$ & $(0,i,i^2,j,j^2,0)$\\
$D_{ijk}$ & $Z_N$ & $(-2,-i,-i^2,-j,-j^2,-k)$\\
\bottomrule
\end{tabular}
\end{table}

The first coordinate is a role filter.  The remaining coordinates synchronize the indices and identify the chosen source edge.

\begin{table}[H]
\centering
\caption{Audit of all twelve class-respecting role combinations.  A zero role sum occurs exactly for the three intended types.}
\label{lred:tab:compiler-role-audit}
\small
\begin{tabular}{@{}ccc r c@{}}
\toprule
$X_N$ role & $Y_N$ role & $Z_N$ role & Role sum & Status\\
\midrule
$A$&$B$&$R$&$0$&intended left\\
$A$&$B$&$D$&$-2$&excluded\\
$A$&$C$&$R$&$1$&excluded\\
$A$&$C$&$D$&$-1$&excluded\\
$A$&$F$&$R$&$-1$&excluded\\
$A$&$F$&$D$&$-3$&excluded\\
$L$&$B$&$R$&$1$&excluded\\
$L$&$B$&$D$&$-1$&excluded\\
$L$&$C$&$R$&$2$&excluded\\
$L$&$C$&$D$&$0$&intended right\\
$L$&$F$&$R$&$0$&intended fallback\\
$L$&$F$&$D$&$-2$&excluded\\
\bottomrule
\end{tabular}
\end{table}

\paragraph{How to read \cref{lred:tab:compiler-vectors,lred:tab:compiler-role-audit}.}
Use the two tables as a two-stage filter.  First add only the $\lambda$-coordinates of a class-respecting role combination; the audit table shows that exactly the left, right, and fallback role types survive.  Then return to the remaining five coordinates in the vector table.  Their linear and squared moments synchronize the pair indices, and the final coordinate identifies the selected third-role edge.  Thus the first table supplies the equations, while the second prevents unintended role combinations from reaching those equations.

\begin{lemma}[Exact vector triple characterization]\label{lred:lem:compiler-characterization}
A class-transversal compiler triple has vector sum $\mathbf 0$ if and only if it is one of the triples in \eqref{lred:eq:left-triple}--\eqref{lred:eq:fallback-triple}.
\end{lemma}

\begin{proof}
By \cref{lred:tab:compiler-role-audit}, only the role types $(A,B,R)$, $(L,C,D)$, and $(L,F,R)$ can have zero first coordinate.

For $(A_i,B_j,R_{i'j'})$, the next four coordinates give
\[
  -i+i'=0,\quad -i^2+(i')^2=0,
  \qquad -j+j'=0,\quad -j^2+(j')^2=0.
\]
Thus $i=i'$ and $j=j'$, giving exactly $T^L_{ij}$.

For $(L_{ij},C_k,D_{i'j'k'})$, the coordinate equations are
\[
  i=i',\qquad j=j',\qquad k=k',
\]
with the square coordinates providing redundant confirmation.  This is exactly $T^R_{ijk}$ for an existing edge item $D_{ijk}$.

Finally, consider $(L_{ij},F_{i'j'},R_{i''j''})$.  The two $u$-coordinates give
\[
 i-2i'+i''=0,
 \qquad
 i^2-2(i')^2+(i'')^2=0.
\]
The first equality says $i+i''=2i'$.  Hence
\begin{align*}
 (i-i')^2+(i''-i')^2
 &=i^2+(i'')^2-2i'(i+i'')+2(i')^2\\
 &=2(i')^2-4(i')^2+2(i')^2=0.
\end{align*}
Both squares vanish, so $i=i'=i''$.  The same argument in the $v$-coordinates gives $j=j'=j''$.  Thus the triple is exactly $T^F_{ij}$.  Direct addition verifies that every intended triple has vector sum zero.
\end{proof}

\paragraph{Why both first and second moments are used.}
The linear equation $i+i''=2i'$ alone permits different indices in arithmetic progression.  The square equation rules out that splicing by forcing zero variance.  This is the same two-moment principle used later to localize private items in the symmetric gadget.

\subsection{No-carry scalarization and unary size}\label{lred:sublred:sec:compiler-scalar}

For each coordinate $q\in\{0,1,\ldots,5\}$, choose
\[
  D_q\ge \max\{0,-\min_s \wvec_q(s)\},
\]
where $s$ ranges over all compiler items.  The shifted digit of item $s$ is $\overline\wvec_q(s)=\wvec_q(s)+D_q$, and the shifted target digit is $3D_q$.  Choose an integer base $L_N$ satisfying
\begin{equation}\label{lred:eq:compiler-base}
 L_N>
 \max_q\left\{3D_q,\ 3\max_s\overline\wvec_q(s)\right\}.
\end{equation}
Define
\begin{equation}\label{lred:eq:compiler-scalar}
 W_N(s)=\sum_{q=0}^{5}\overline\wvec_q(s)L_N^q,
 \qquad
 T_N=\sum_{q=0}^{5}3D_qL_N^q.
\end{equation}
Choose $D_0$ one unit larger than the lower bound displayed above, so that every least-significant digit is positive.  All item weights are then positive.

\begin{lemma}[Compiler no-carry equivalence]\label{lred:lem:compiler-no-carry}
For every class-transversal triple $(x,y,z)$,
\[
  W_N(x)+W_N(y)+W_N(z)=T_N
  \quad\Longleftrightarrow\quad
  \wvec(x)+\wvec(y)+\wvec(z)=\mathbf 0.
\]
\end{lemma}

\begin{proof}
Every sum of three shifted item digits and every target digit is strictly smaller than $L_N$ by \eqref{lred:eq:compiler-base}.  Therefore no carry occurs.  Equality of the scalar integers is equivalent to equality in all six shifted coordinates, and subtracting $3D_q$ recovers the signed vector equations.
\end{proof}

Let $q_0=\max\{a,b,c,p\}$.  Every index is at most $q_0$, so every signed coordinate in \cref{lred:tab:compiler-vectors} has magnitude $O(q_0^2)$.  The shifts and $L_N$ are $O(q_0^2)$, and every encoded weight and the target are $O(L_N^6)=O(q_0^{12})$.  Because isolated vertices were deleted, $q_0\le p$ for a nonempty source.  The compiler is therefore polynomial under unary encoding.

\subsection{Exact objective relation and decoder}\label{lred:sublred:sec:compiler-opt}

\begin{lemma}[Forward pair completion]\label{lred:lem:compiler-forward}
A source matching of size $k$ gives a numerical matching of size $r+k$.
\end{lemma}

\begin{proof}
For every selected edge $(u_i,v_j,w_k)$, select $T^L_{ij}$ and $T^R_{ijk}$.  For every represented pair $(i,j)$ not used by a selected source edge, select $T^F_{ij}$.  A source matching never uses the same $u_i$, $v_j$, or $w_k$ twice, and it cannot contain two edges with the same pair $(i,j)$.  Hence all selected compiler items are disjoint.  There are $2k$ triples for completed pairs and $r-k$ fallback triples, for a total of $r+k$.
\end{proof}

\begin{definition}[Complete pair]\label{lred:def:complete-pair}
In a numerical compiler solution, a represented pair $(i,j)$ is \emph{complete} if the solution contains both one left triple $T^L_{ij}$ and one right triple $T^R_{ijk}$ for some edge $(i,j,k)$.
\end{definition}

\begin{lemma}[Complete-pair decoder]\label{lred:lem:complete-pair-decoder}
Every numerical compiler matching of size $s$ contains at least $s-r$ complete pairs and yields in polynomial time a source matching of size at least $\max\{0,s-r\}$.
\end{lemma}

\begin{proof}
By \cref{lred:lem:compiler-characterization}, every selected triple is associated with a unique represented pair.  For a fixed pair, the shared item $R_{ij}$ prevents selecting two left triples, the shared item $L_{ij}$ prevents selecting two right triples, and the fallback triple uses both $L_{ij}$ and $R_{ij}$ and therefore cannot coexist with either side.  Thus a pair contributes at most two selected triples; it contributes two exactly when it is complete.

Let $c$ be the number of complete pairs and let $a'$ be the number of pairs contributing at least one triple.  Then
\[
  s=a'+c\le r+c,
\]
so $c\ge s-r$.  For each complete pair, output the source edge identified by its right triple.  Distinct complete pairs use distinct $A_i$ and $B_j$ items in their left triples and distinct $C_k$ items in their right triples.  The decoded source edges are therefore pairwise vertex-disjoint.
\end{proof}

\begin{theorem}[Exact pair-compiler identity]\label{lred:thm:compiler-identity}
The unpadded compiler instance satisfies
\[
  \boxed{\OPT_{\mathrm N}=r+k^\star.}
\]
The forward and reverse maps are polynomial-time algorithms.
\end{theorem}

\begin{proof}
\Cref{lred:lem:compiler-forward} gives $\OPT_{\mathrm N}\ge r+k^\star$.  If a numerical solution has size $s$, \cref{lred:lem:complete-pair-decoder} gives a source matching of size at least $s-r$, so $s-r\le k^\star$ and $s\le r+k^\star$.  Taking the maximum over $s$ proves equality.
\end{proof}

\section{Step 5: padding and one-live-port separation}\label{lred:sec:separation}

The three compiler classes need not have equal sizes.  Pad each class with isolated items of weight
\[
  \barrier=T_N+1
\]
until all have common size
\begin{equation}\label{lred:eq:m-def}
  m=\max\{|X_N|,|Y_N|,|Z_N|\}.
\end{equation}
Because all weights are positive, no target-sum triple can contain a padding item.  Padding does not change the optimum or the decoder.

Write the padded classes as
\[
  X_N=\{x_1,\ldots,x_m\},\quad
  Y_N=\{y_1,\ldots,y_m\},\quad
  Z_N=\{z_1,\ldots,z_m\},
\]
where some displayed items may be padding occurrences of weight $\barrier$.

Create a new numerical instance with $n=3m$ indexed groups:
\begin{align}
 (a_i,b_i,c_i)&=(W_N(x_i),\barrier,\barrier) &&(1\le i\le m),\label{lred:eq:sep-X}\\
 (a_{m+i},b_{m+i},c_{m+i})&=(\barrier,W_N(y_i),\barrier) &&(1\le i\le m),\label{lred:eq:sep-Y}\\
 (a_{2m+i},b_{2m+i},c_{2m+i})&=(\barrier,\barrier,W_N(z_i)) &&(1\le i\le m).\label{lred:eq:sep-Z}
\end{align}
The target remains $T_N$.

\begin{lemma}[One-live-port separation]\label{lred:lem:separation}
The separated numerical instance has the same feasible matchings and the same optimum as the padded compiler instance, under the natural identification of nonbarrier occurrences.  Every separated group has at most one coordinate whose item can occur in a target triple.
\end{lemma}

\begin{proof}
A class-transversal triple containing an occurrence of weight $\barrier=T_N+1$ has total weight greater than $T_N$, because all other weights are positive.  Therefore every legal separated triple uses one original nonbarrier item from each numerical class.  Those triples are exactly the legal triples of the padded compiler instance.  Conversely, every compiler triple appears unchanged among the separated nonbarrier items.  The optimum is preserved.  Equations \eqref{lred:eq:sep-X}--\eqref{lred:eq:sep-Z} place nonbarrier items in only one coordinate of each group; a group generated from a padding item may have no live coordinate at all.
\end{proof}

\paragraph{Why separation is mathematically essential.}
In the unrestricted symmetric construction, a defective local filler block can create useful capacity at more than one port.  Charging all such shortages to the same missing fillers can double-count the local deficit.  Separation ensures that only one of the three port shortages in a group can be useful to a main triple.  The local inequality in \cref{lred:sec:local-certificate} can therefore be charged once per group and then summed globally.

The class-size bound needed later follows immediately.  Since $r\le p$ and isolated source vertices were deleted,
\begin{align}
 |X_N|&=|U|+r\le2p,\notag\\
 |Y_N|&=|V|+|W|+r\le3p,\notag\\
 |Z_N|&=r+p\le2p.
\end{align}
Hence
\begin{equation}\label{lred:eq:m-bound}
  m\le3p,
  \qquad
  n=3m\le9p.
\end{equation}

\section{Step 6: the symmetric 13-role construction}\label{lred:sec:symmetric}

This section restates the complete portion of the symmetric construction needed for the L-reduction.  The construction is applied to the separated numerical instance with lists $A=\mset{a_1,\ldots,a_n}$, $B=\mset{b_1,\ldots,b_n}$, $C=\mset{c_1,\ldots,c_n}$ and target $t=T_N$.

For every group $i$, create
\[
 S_i=\{P_{1,i},P_{2,i},P_{3,i},Q_{1,i},\ldots,Q_{10,i}\},
 \qquad
 S=\mathop{\uplusm}_{i=1}^{n}S_i.
\]
The physical output classes are three labeled copies $S^X,S^Y,S^Z$.  Thus each class contains $13n$ items and the three weight multisets are identical.

\subsection{Role data and target patterns}\label{lred:subsec:sym-role-data}

For role $R$ in group $i$, define
\begin{equation}\label{lred:eq:sym-vector}
  \wvec(R_i)=(\rho_R,\eta_{R,i},h_Ri,h_Ri^2),
  \qquad
  \Kvec=(300,3t,0,0),
\end{equation}
using \cref{lred:tab:sym-role-data}.

\begin{table}[htbp]
\centering
\caption{Logical coordinates of the thirteen symmetric roles.}
\label{lred:tab:sym-role-data}
\begin{tabular}{@{}c r r l@{}}
\toprule
Role $R$ & $\rho_R$ & $h_R$ & $\eta_{R,i}$\\
\midrule
$P_1$&31&0&$3a_i$\\
$P_2$&8&0&$3b_i$\\
$P_3$&261&0&$3c_i$\\
$Q_1$&184&2&$-3b_i$\\
$Q_2$&85&$-2$&$3t-3a_i+3b_i$\\
$Q_3$&177&$-2$&$3t-2a_i+b_i+c_i$\\
$Q_4$&115&2&$2a_i-4b_i-c_i$\\
$Q_5$&23&2&$a_i-2b_i-2c_i$\\
$Q_6$&16&$-2$&$3t-a_i+2b_i-c_i$\\
$Q_7$&27&1&$2a_i-b_i-c_i$\\
$Q_8$&89&$-3$&$3t-2a_i+4b_i+c_i$\\
$Q_9$&188&1&$a_i-2b_i+c_i$\\
$Q_{10}$&96&1&$0$\\
\bottomrule
\end{tabular}
\end{table}

\paragraph{How to read \cref{lred:tab:sym-role-data}.}
Read a row as a recipe for one role, not as a complete target triple.  The fixed code $\rho_R$ selects admissible role multisets, the expression $\eta_{R,i}$ carries the source arithmetic, and the two coefficients generated from $h_R$ enforce localization through the group and group-square coordinates.  The pattern list immediately below tells which three recipes are added together.  Separating these three jobs makes it easier to audit a target identity without confusing role admissibility with numerical feasibility.

The permitted role patterns are
\begin{align}
 F_0&=(P_1,P_2,P_3),\notag\\
 F_1&=(P_1,Q_1,Q_2), &
 F_2&=(P_2,Q_3,Q_4), &
 F_3&=(P_3,Q_5,Q_6),\notag\\
 F_4&=(Q_1,Q_7,Q_8), &
 F_5&=(Q_2,Q_7,Q_9), &
 F_6&=(Q_3,Q_7,Q_{10}),\notag\\
 F_7&=(Q_4,Q_8,Q_{10}), &
 F_8&=(Q_5,Q_8,Q_9), &
 F_9&=(Q_6,Q_9,Q_{10}).\label{lred:eq:sym-patterns}
\end{align}
Pattern $F_0$ is main; the other nine patterns are fillers.

\begin{lemma}[Structure of symmetric target triples]\label{lred:lem:sym-structure}
After a no-carry scalarization of \eqref{lred:eq:sym-vector}, every target triple has one of the patterns $F_0,\ldots,F_9$.  Moreover:
\begin{enumerate}[label=(\alph*)]
\item a main triple $(P_{1,i},P_{2,j},P_{3,k})$ is legal exactly when $a_i+b_j+c_k=t$;
\item all private items of a filler triple have one common group index;
\item an $F_1$ filler localized at group $i$ uses a $P_1$ port whose represented value is $a_i$, and analogously for $F_2$ and $b_i$ and for $F_3$ and $c_i$;
\item every local filler $F_r(i)$, $r=1,\ldots,9$, is legal.
\end{enumerate}
\end{lemma}

\begin{proof}
The fixed role codes have sum $300$ exactly for the ten patterns in \eqref{lred:eq:sym-patterns}; an exhaustive complementary-pair certificate is reproduced in \cref{lred:tab:role-complements}.  For $F_0$, the last two coordinates vanish and the source coordinate is $3(a_i+b_j+c_k)$.

For $F_1,F_2,F_3$, the two private coefficients are $2,-2$ or $-2,2$, so their group indices agree.  For $F_4,\ldots,F_9$, the multiset of nonzero coefficients is either $\{2,1,-3\}$ or $\{-2,1,1\}$.  The equations in $i$ and $i^2$ force a common index by the same zero-variance calculation used in \cref{lred:lem:compiler-characterization}.

For an $F_1$ filler whose private items are from group $i$,
\[
  \eta_{Q_1,i}+\eta_{Q_2,i}=3t-3a_i.
\]
A port $P_{1,j}$ completes the target exactly when $a_j=a_i$.  The analogous identities are
\[
  \eta_{Q_3,i}+\eta_{Q_4,i}=3t-3b_i,
  \qquad
  \eta_{Q_5,i}+\eta_{Q_6,i}=3t-3c_i.
\]
Direct substitution verifies all local filler source identities and the vanishing group coefficients.
\end{proof}

\subsection{Scalar encoding and partial recoloring}\label{lred:subsec:sym-scalar-recolor}

Choose shifts $D_1,D_2,D_3$ that make the last three coordinates nonnegative and a base $L_S$ larger than every target digit and every sum of three shifted item digits.  Define
\begin{align}
 W_S(R_i)&=\rho_R+(\eta_{R,i}+D_1)L_S+(h_Ri+D_2)L_S^2+(h_Ri^2+D_3)L_S^3,\label{lred:eq:sym-scalar}\\
 K_S&=300+(3t+3D_1)L_S+3D_2L_S^2+3D_3L_S^3.\label{lred:eq:sym-target}
\end{align}
The no-carry argument is identical to \cref{lred:lem:compiler-no-carry}: scalar equality is equivalent to the four vector equations.

\begin{lemma}[Partial copy-forgetting and recoloring]\label{lred:lem:partial-recolor}
A partial symmetric matching of $h$ triples becomes, after the $X,Y,Z$ copy labels are forgotten, a multiset of $h$ target hyperedges on $S$ with incidence degree at most three at every occurrence.  Conversely, every such degree-at-most-three target hypergraph can be lifted in polynomial time to a partial matching of the three labeled copies without changing its cardinality.
\end{lemma}

\begin{proof}
Forgetting labels leaves at most the three physical copies of any underlying occurrence.  Conversely, form the incidence bipartite multigraph between underlying occurrences and hyperedge occurrences.  Its maximum degree is at most three, and every hyperedge vertex has degree three.  By K\H{o}nig's line-coloring theorem for bipartite multigraphs, its edges have a proper $3$-edge-coloring \citep{LovaszPlummer1986,Alon2003}.  Interpret the three colors as $X,Y,Z$.  The three incidences of every hyperedge receive distinct colors, and no item vertex repeats a color, producing a legal class-transversal partial matching.
\end{proof}

\begin{lemma}[Forward symmetric extension]\label{lred:lem:sym-forward}
A separated numerical matching of size $s$ yields a symmetric matching of size $12n+s$.
\end{lemma}

\begin{proof}
For every group $i$, take two copies each of $F_1(i),F_2(i),F_3(i)$ and
one copy each of $F_4(i),\ldots,F_9(i)$; here $F_r(i)$ takes all its
roles from group $i$, so in particular the ports used by $F_1(i)$,
$F_2(i)$, and $F_3(i)$ are $P_{1,i}$, $P_{2,i}$, and $P_{3,i}$
themselves.
These are twelve legal fillers.  Every private item then has degree three and every port has degree two.  Add one main hyperedge for every numerical source triple.  Source item-disjointness makes the main incidence degree at every live port at most one.  Apply \cref{lred:lem:partial-recolor}.
\end{proof}

\subsection{The new local filler certificate}\label{lred:sec:local-certificate}

Fix a group $i$ in an arbitrary partial symmetric solution after copy labels are forgotten.  Let $x_{r,i}$ be the number of fillers of pattern $F_r$ whose private items are localized at group $i$, and define
\[
  f_i=\sum_{r=1}^{9}x_{r,i}.
\]
Every private occurrence has incidence degree at most three.  Therefore
\begin{align}
 C_1:&\quad x_1+x_4\le3, &
 C_2:&\quad x_1+x_5\le3,\notag\\
 C_3:&\quad x_2+x_6\le3, &
 C_4:&\quad x_2+x_7\le3,\notag\\
 C_5:&\quad x_3+x_8\le3, &
 C_6:&\quad x_3+x_9\le3,\notag\\
 C_7:&\quad x_4+x_5+x_6\le3, &
 C_8:&\quad x_4+x_7+x_8\le3,\notag\\
 C_9:&\quad x_5+x_8+x_9\le3, &
 C_{10}:&\quad x_6+x_7+x_9\le3,
 \label{lred:eq:capacity-system}
\end{align}
where the group subscript $i$ is suppressed.

\begin{lemma}[Four linear-combination certificates]\label{lred:lem:four-certificates}
Every nonnegative solution of \eqref{lred:eq:capacity-system} satisfies
\begin{align}
 f_i&\le12,\label{lred:eq:fi12}\\
 f_i-x_{1,i}&\le10,\label{lred:eq:fi-x1}\\
 f_i-x_{2,i}&\le10,\label{lred:eq:fi-x2}\\
 f_i-x_{3,i}&\le10.\label{lred:eq:fi-x3}
\end{align}
\end{lemma}

\begin{proof}
Multiply $C_1,\ldots,C_{10}$ by the nonnegative coefficient vectors in \cref{lred:tab:multiplier-certificates} and sum.  The table also records the resulting coefficient of every $x_j$ and the resulting right-hand side.  Thus every displayed inequality is a direct nonnegative linear combination of valid capacity inequalities.

\begin{table}[H]
\centering
\caption{Exact multiplier certificates for the local filler inequalities.}
\label{lred:tab:multiplier-certificates}
\scriptsize
\renewcommand{\arraystretch}{1.16}
\begin{tabular}{@{}l l l r@{}}
\toprule
Derived inequality & Multipliers on $(C_1,\ldots,C_{10})$ & Resulting coefficients on $(x_1,\ldots,x_9)$ & RHS\\
\midrule
$f_i\le12$
& $\frac13(2,1,1,2,2,1,1,0,1,1)$
& $(1,1,1,1,1,1,1,1,1)$ & 12\\
$f_i-x_1\le10$
& $\frac13(0,0,1,2,1,2,2,1,1,0)$
& $(0,1,1,1,1,1,1,1,1)$ & 10\\
$f_i-x_2\le10$
& $\frac13(1,2,0,0,2,1,1,1,0,2)$
& $(1,0,1,1,1,1,1,1,1)$ & 10\\
$f_i-x_3\le10$
& $\frac13(2,1,2,1,0,0,0,1,2,1)$
& $(1,1,0,1,1,1,1,1,1)$ & 10\\
\bottomrule
\end{tabular}
\end{table}

\paragraph{How to read \cref{lred:tab:multiplier-certificates}.}
Each row is a complete, checkable proof certificate.  Multiply the ten capacity inequalities by the nonnegative multipliers in the second column and add them.  The third column records the resulting coefficient of every $x_j$, and the last column records the resulting right-hand side.  No optimization result is being quoted: verifying the displayed products directly proves the four inequalities used by the decoder.
\end{proof}

\begin{corollary}[One-port deficit inequality]\label{lred:cor:one-port-deficit}
For each $r\in\{1,2,3\}$,
\begin{equation}\label{lred:eq:key-local}
  \boxed{\pos{2-x_{r,i}}\le12-f_i.}
\end{equation}
\end{corollary}

\begin{proof}
By \eqref{lred:eq:fi12}, the right-hand side is nonnegative.  If $x_{r,i}\ge2$, the left-hand side is zero.  If $x_{r,i}<2$, the appropriate inequality among \eqref{lred:eq:fi-x1}--\eqref{lred:eq:fi-x3} gives
\[
  12-f_i\ge2-x_{r,i}.
\]
\end{proof}

\paragraph{Intuition.}
The canonical local scaffold has $f_i=12$ and $x_{1,i}=x_{2,i}=x_{3,i}=2$.  Reducing $x_{r,i}$ below two frees capacity at port $r$.  Inequality \eqref{lred:eq:key-local} says that each unit of freed capacity costs at least one missing filler triple.  It is an individual-port statement.  One-live separation is what permits the proof to use only one such statement in each group.

\section{Step 7: exact affine behavior of the symmetric construction}\label{lred:sec:affine-symmetric}

Let $Y$ be any symmetric solution and forget its copy labels.  Let
\[
  v=\val_{\mathrm{SN}}(Y)=f+h,
\]
where $f=\sum_i f_i$ is the number of filler hyperedges and $h$ is the number of main hyperedges.

For every nonbarrier weight $\alpha$ in the separated first class, let $r_X(\alpha)$ be its source multiplicity and let $\mu_X(\alpha)$ be the number of selected main hyperedges whose $P_1$ incidence represents $\alpha$.  Define $r_Y,\mu_Y$ and $r_Z,\mu_Z$ analogously.

Let $G_X$ be the groups whose only potentially live coordinate is the first coordinate and whose first value is nonbarrier.  Define $G_Y$ and $G_Z$ analogously.  These three group sets are pairwise disjoint.  Groups created from padding items may belong to none of them.

\begin{lemma}[Overuse is paid by missing fillers]\label{lred:lem:overuse-bound}
Define
\begin{align*}
 E_X&=\sum_\alpha\pos{\mu_X(\alpha)-r_X(\alpha)},\\
 E_Y&=\sum_\alpha\pos{\mu_Y(\alpha)-r_Y(\alpha)},\\
 E_Z&=\sum_\alpha\pos{\mu_Z(\alpha)-r_Z(\alpha)}.
\end{align*}
Then
\begin{equation}\label{lred:eq:total-overuse-new}
  E_X+E_Y+E_Z\le12n-f.
\end{equation}
\end{lemma}

\begin{proof}
Fix a first-class value $\alpha$.  By the value-locality part of
\cref{lred:lem:sym-structure}, every $F_1$ filler localized at a group in
$G_X$ with first value $\alpha$ uses a $P_1$ port of value $\alpha$, and
every main hyperedge counted by $\mu_X(\alpha)$ does the same through its
single $P_1$ incidence.  Each of these hyperedges contains exactly one
port occurrence, and distinct selected hyperedges consume distinct
physical port incidences: hyperedges in a partial solution are pairwise
item-disjoint, so even when several of them meet the same underlying port
occurrence they occupy different labeled copies of it.  The
$\mu_X(\alpha)+\sum_{i\in G_X:a_i=\alpha}x_{1,i}$ consumed incidences are
therefore pairwise distinct, while the $r_X(\alpha)$ ports of value
$\alpha$ supply at most $3r_X(\alpha)$ incidences in total.  Hence
\[
  \mu_X(\alpha)+\sum_{i\in G_X:a_i=\alpha}x_{1,i}
  \le3r_X(\alpha).
\]
Subtracting $r_X(\alpha)$ gives
\[
  \mu_X(\alpha)-r_X(\alpha)
  \le
  \sum_{i\in G_X:a_i=\alpha}(2-x_{1,i}).
\]
Taking positive parts and using the elementary inequality
$\pos{\sum_i z_i}\le\sum_i\pos{z_i}$ therefore gives
\[
  \pos{\mu_X(\alpha)-r_X(\alpha)}
  \le
  \pos{\sum_{i\in G_X:a_i=\alpha}(2-x_{1,i})}
  \le
  \sum_{i\in G_X:a_i=\alpha}\pos{2-x_{1,i}}.
\]
Summing over $\alpha$ and applying \cref{lred:cor:one-port-deficit},
\[
 E_X\le\sum_{i\in G_X}\pos{2-x_{1,i}}
 \le\sum_{i\in G_X}(12-f_i).
\]
The same argument gives
\[
 E_Y\le\sum_{i\in G_Y}(12-f_i),
 \qquad
 E_Z\le\sum_{i\in G_Z}(12-f_i).
\]
The group sets are disjoint, and every term $12-f_i$ is nonnegative by \eqref{lred:eq:fi12}.  Therefore
\[
 E_X+E_Y+E_Z
 \le\sum_{i\in G_X\cup G_Y\cup G_Z}(12-f_i)
 \le\sum_{i=1}^{n}(12-f_i)=12n-f.
\]
\end{proof}

\paragraph{Where the Part~II obstruction disappears.}
In the more general grouped construction of Part~II, the same group could appear in all three sums and its quantity $12-f_i$ could be charged two or three times.  Here a group has at most one useful port, so the sets $G_X,G_Y,G_Z$ are disjoint.  This is the only point in the proof where one-live separation is used, but it is decisive.

\begin{lemma}[Polynomial overuse repair]\label{lred:lem:repair}
From the $h$ selected main hyperedges one can delete at most $E_X+E_Y+E_Z$ hyperedges and assign the survivors injectively to labeled separated numerical items.  The survivors form a feasible separated \MaxNthreeDM matching.
\end{lemma}

\begin{proof}
While some role-value pair is used more often than its source multiplicity, delete one main hyperedge containing that role-value pair.  Deletion cannot increase any positive overuse and decreases the chosen positive overuse by one.  Thus no more than the initial total $E_X+E_Y+E_Z$ deletions occur.

After termination, every value is requested no more often in each role than it occurs in the corresponding source class.  Assign equal-valued main incidences injectively to the labeled source occurrences.  Every surviving main hyperedge satisfies the numerical target equation by \cref{lred:lem:sym-structure}(a).  No main hyperedge can contain the barrier value $T_N+1$, since all source values are positive and the target is $T_N$.  Therefore the assigned triples use only genuine padded-compiler items and form a feasible matching.
\end{proof}

\begin{theorem}[Exact symmetric affine identity]\label{lred:thm:symmetric-affine}
For the separated numerical instance with $n=3m$ groups,
\begin{equation}\label{lred:eq:sym-affine}
  \boxed{\OPT_{\mathrm{SN}}=12n+\OPT_{\mathrm N}=36m+\OPT_{\mathrm N}.}
\end{equation}
Moreover, every symmetric solution of value $v$ decodes in polynomial time to a numerical matching of value at least
\begin{equation}\label{lred:eq:sym-decoder}
  \max\{0,v-12n\}=\max\{0,v-36m\}.
\end{equation}
\end{theorem}

\begin{proof}
By \cref{lred:lem:overuse-bound,lred:lem:repair}, the number of surviving main hyperedges is at least
\[
  h-(E_X+E_Y+E_Z)
  \ge h-(12n-f)=v-12n.
\]
If this bound is negative, return the empty matching.  This proves \eqref{lred:eq:sym-decoder}, and it implies
\[
  v\le12n+\OPT_{\mathrm N}
\]
for every symmetric solution.  Conversely, \cref{lred:lem:sym-forward} extends every numerical matching of size $s$ to a symmetric matching of size $12n+s$.  Applying this to an optimum numerical matching proves equality in \eqref{lred:eq:sym-affine}.
\end{proof}

Combining \cref{lred:thm:compiler-identity,lred:thm:symmetric-affine} gives the exact destination optimum
\begin{equation}\label{lred:eq:full-affine}
  \boxed{\OPT_{\mathrm{SN}}=36m+r+k^\star.}
\end{equation}

\section{Step 8: the L-reduction, unary bound, and consequences}\label{lred:sec:final-lred}

Define the instance map $f$ by the complete pipeline:
\begin{align*}
 \MaxThreeDMThree
 &\longrightarrow \text{pair compiler}
 \longrightarrow \text{padding}\\
 &\longrightarrow \text{one-live separation}
 \longrightarrow \text{symmetric 13-role construction}.
\end{align*}
For a nonempty source, define
\[
  \Gamma=36m+r.
\]
Then \eqref{lred:eq:full-affine} is
\[
  \OPT_{\mathrm{SN}}=\Gamma+k^\star.
\]

\subsection{First L-inequality}\label{lred:subsec:first-l}

Using $m\le3p$, $r\le p$, and $p\le7k^\star$,
\begin{align}
 \OPT_{\mathrm{SN}}
 &=36m+r+k^\star\notag\\
 &\le108p+p+k^\star\notag\\
 &=109p+k^\star\notag\\
 &\le109(7k^\star)+k^\star\notag\\
 &=764k^\star.
 \label{lred:eq:alpha764}
\end{align}
Thus \eqref{lred:eq:lred1} holds with $\alpha=764$.

\subsection{Second L-inequality and the combined decoder}\label{lred:subsec:second-l}

Given a symmetric solution $Y$ of value $v$, first apply \cref{lred:thm:symmetric-affine} to obtain a padded compiler matching of size
\[
  s\ge\max\{0,v-36m\}.
\]
Then apply the complete-pair decoder.  The final source matching has value $k_Y$ satisfying
\begin{equation}\label{lred:eq:combined-decoder}
  k_Y\ge\max\{0,v-36m-r\}=\max\{0,v-\Gamma\}.
\end{equation}

If $v\ge\Gamma$, then
\begin{align*}
 k^\star-k_Y
 &\le k^\star-(v-\Gamma)\\
 &=\Gamma+k^\star-v\\
 &=\OPT_{\mathrm{SN}}-v.
\end{align*}
If $v<\Gamma$, then $k_Y\ge0$ and
\[
  k^\star-k_Y\le k^\star<\Gamma+k^\star-v
  =\OPT_{\mathrm{SN}}-v.
\]
Therefore
\begin{equation}\label{lred:eq:beta1}
  \boxed{k^\star-k_Y\le\OPT_{\mathrm{SN}}-v,}
\end{equation}
so the second L-inequality holds with $\beta=1$.

\paragraph{The zero-optimum source case.}
If $E=\varnothing$, then $k^\star=0$.  Map the instance to a fixed symmetric instance with one item of weight $1$ in each class and target $4$.  Its optimum is zero.  The only feasible destination solution is empty, and the decoder returns the empty source matching.  Both L-inequalities hold.  This separate branch is necessary because an affine construction with a positive fixed scaffold could not satisfy \eqref{lred:eq:lred1} when the source optimum is zero.

\subsection{Polynomial running time and unary size}\label{lred:subsec:unary-size}

The pair compiler uses $O(p)$ items and six coordinates of magnitude $O(p^2)$.  Its scalar weights and target are $O(p^{12})$.  Padding and separation create $n=3m=O(p)$ numerical groups and use only the barrier value $T_N+1=O(p^{12})$.

In the symmetric construction, every source-coordinate expression is a fixed linear combination of $a_i,b_i,c_i,t$, so its magnitude is $O(p^{12})$.  The group and group-square coordinates are $O(p)$ and $O(p^2)$.  Hence the symmetric no-carry base can be chosen as $O(p^{12})$.  Four digits then give final item weights and target of size
\[
  O(p^{48}).
\]
There are $13n=O(p)$ items in each class.  The unary output length is therefore polynomial in $p$, and all arithmetic operations use polynomially many bits.  Both the construction and the two-stage decoder are polynomial-time algorithms.

\begin{proof}[Proof of \cref{lred:thm:main}]
The instance map and combined decoder are polynomial by the preceding construction.  Equation \eqref{lred:eq:alpha764} proves the first L-inequality with $\alpha=764$, and \eqref{lred:eq:beta1} proves the second with $\beta=1$.  The empty source is handled separately.  The numerical size analysis proves unary-polynomiality.  Thus the displayed maps form the claimed L-reduction.
\end{proof}

\begin{corollary}[Explicit PTAS transfer]\label{lred:cor:ptas-transfer}
If a polynomial-time algorithm returns, for every $\delta>0$, a symmetric solution of value at least $(1-\delta)\OPT_{\mathrm{SN}}$, then the decoder returns a source solution of value at least
\[
  (1-764\delta)k^\star.
\]
Consequently, choosing $\delta=\varepsilon/764$ gives a PTAS-reduction.
\end{corollary}

\begin{proof}
By the second L-inequality and then the first,
\begin{align*}
 k^\star-k_Y
 &\le \OPT_{\mathrm{SN}}-v\\
 &\le\delta\OPT_{\mathrm{SN}}\\
 &\le764\delta k^\star.
\end{align*}
\end{proof}

\begin{proposition}[APX membership]\label{lred:prop:apx-membership}
\MaxSNthreeDM has a polynomial-time $3$-approximation.
\end{proposition}

\begin{proof}
Enumerate all class-transversal target-sum triples and greedily take any maximal item-disjoint family $M_G$.  Every optimum triple intersects some triple in $M_G$.  A chosen triple has three items, and an optimum matching contains at most one triple through each item, so at most three optimum triples can be charged to one greedy triple.  Hence $\OPT_{\mathrm{SN}}\le3|M_G|$.
\end{proof}

The sharper $(4/3+\varepsilon)$ guarantee discussed after
\cref{apx:prop:apx-membership} also applies, although the elementary
factor-$3$ argument is sufficient here.

\begin{corollary}[Standard APX-completeness]\label{lred:cor:apx-complete}
\MaxSNthreeDM is APX-hard under an explicit L-reduction and belongs to APX.  It is therefore APX-complete under this standard convention.  In particular, it has no PTAS unless $\mathrm P=\mathrm{NP}$.
\end{corollary}

\section{Fully worked  examples}\label{lred:sec:examples}

The two examples use the same vertex sets, the same number of source edges, and the same encoding parameters.  The first has a perfect matching of size two; the second has no matching of size two but still has a nonzero optimum.  They therefore form a genuinely nontrivial ``yes/no'' pair for the decision threshold two while illustrating the optimization values carried by the L-reduction.

The examples are intentionally complete rather than schematic.  Each begins with a source optimum proof, lists every item created by the exact pair compiler, audits every intended target-sum triple, displays all 21 one-live groups, and then follows the matching through the symmetric scaffold.  The final error ledger shows how the fixed offset cancels and why one unit of destination error pays for at most one unit of source error.  Readers meeting the construction for the first time may read the explanatory paragraphs and objective calculations first, then return to the large tables as an arithmetic audit.

\subsection{Numerical parameters common to both examples}\label{lred:subsec:example-common}

For the six-coordinate compiler, the indices are only $1$ and $2$.  Use shifts
\begin{equation}\label{lred:eq:example-compiler-shifts}
  (D_0,D_1,D_2,D_3,D_4,D_5)=(3,4,8,4,8,2)
\end{equation}
and base
\[
  L_N=37.
\]
The shifted target digits are
\[
  (9,12,24,12,24,6),
\]
so the scalar target is
\begin{align}
 T_N
 &=9+12(37)+24(37^2)+12(37^3)+24(37^4)+6(37^5)\notag\\
 &=\boxed{461{,}684{,}751}.
 \label{lred:eq:example-TN}
\end{align}
The barrier value is therefore
\[
  \barrier=T_N+1=461{,}684{,}752.
\]

\paragraph{Audit of the base.}
After shifting, the largest possible item digit in the six coordinates is respectively $4,6,12,6,12,4$.  The corresponding sums of three item digits are at most
\[
  12,18,36,18,36,12,
\]
all strictly below $37$.  Every target digit is also below $37$.  Thus the scalar calculations below are genuinely coordinatewise; no carry is possible.

After padding, both examples have compiler class size $m=7$.  One-live separation therefore creates
\[
  n=3m=21
\]
groups and a canonical symmetric filler scaffold of
\[
  12n=252
\]
triples.  Each symmetric output class has $13n=273$ item occurrences.

For both examples, an exact calculation over the 21 separated groups gives
\begin{align*}
 \min_{R,i}\eta_{R,i}&=-2{,}000{,}645{,}026,\\
 \max_{R,i}\eta_{R,i}&= 3{,}385{,}699{,}279.
\end{align*}
Use
\begin{equation}\label{lred:eq:example-sym-shifts}
  D_1=2{,}000{,}645{,}026,
  \qquad D_2=63,
  \qquad D_3=1{,}323.
\end{equation}
The largest shifted source digit is $5{,}386{,}344{,}305$, so the largest sum of three source digits is $16{,}159{,}032{,}915$.  A valid symmetric base is therefore
\begin{equation}\label{lred:eq:example-LS}
  L_S=16{,}159{,}032{,}916.
\end{equation}
The four target digits are
\[
  (300,7{,}386{,}989{,}331,189,3{,}969),
\]
and the final symmetric target is
\begin{equation}\label{lred:eq:example-KS}
  \boxed{K_S=16{,}746{,}621{,}154{,}576{,}131{,}373{,}554{,}069{,}341{,}030{,}904}.
\end{equation}
The large decimal representation is not a source of complexity: it contains only four base-$L_S$ digits and is polynomially bounded in the source size.

\subsection{A source-to-output "yes" instance}\label{lred:subsec:yes-example}

\subsubsection{Source instance and optimum}

Let
\[
 U=\{u_1,u_2\},\qquad
 V=\{v_1,v_2\},\qquad
 W=\{w_1,w_2\},
\]
and let
\begin{equation}\label{lred:eq:yes-edges}
 E_Y=\{(u_1,v_1,w_1),(u_2,v_2,w_2),(u_1,v_2,w_2)\}.
\end{equation}
The first two edges are disjoint, so
\[
 M_Y=\{(u_1,v_1,w_1),(u_2,v_2,w_2)\}
\]
is a matching of size two.  No matching can contain more than $|U|=2$ edges.  Therefore
\begin{equation}\label{lred:eq:yes-source-opt}
  k_Y^\star=2.
\end{equation}
The represented pair set is
\[
  \Pairset_Y=\{(1,1),(1,2),(2,2)\},
  \qquad r=3.
\]

\subsubsection{Every compiler item and scalar weight}

\Cref{lred:tab:yes-compiler-items} lists every nonpadding compiler item.  The fourth column gives the six shifted base-$37$ digits; the last column is the resulting scalar weight.

{\small
\begin{longtable}{@{}c c p{0.26\textwidth} p{0.22\textwidth} r@{}}
\caption{All compiler items for the ``yes'' instance.}\label{lred:tab:yes-compiler-items}\\
\toprule
Class & Item & Signed vector & Shifted digits & Scalar weight\\
\midrule
\endfirsthead
\toprule
Class & Item & Signed vector & Shifted digits & Scalar weight\\
\midrule
\endhead
$X_N$ & $A_{1}$ & $(0,-1,-1,0,0,0)$ & $(3,3,7,4,8,2)$ & 153,893,511 \\
$X_N$ & $A_{2}$ & $(0,-2,-4,0,0,0)$ & $(3,2,4,4,8,2)$ & 153,889,367 \\
$X_N$ & $L_{11}$ & $(1,1,1,1,1,0)$ & $(4,5,9,5,9,2)$ & 155,821,138 \\
$X_N$ & $L_{12}$ & $(1,1,1,2,4,0)$ & $(4,5,9,6,12,2)$ & 161,494,274 \\
$X_N$ & $L_{22}$ & $(1,2,4,2,4,0)$ & $(4,6,12,6,12,2)$ & 161,498,418 \\
$Y_N$ & $B_{1}$ & $(0,0,0,-1,-1,0)$ & $(3,4,8,3,7,2)$ & 151,970,103 \\
$Y_N$ & $B_{2}$ & $(0,0,0,-2,-4,0)$ & $(3,4,8,2,4,2)$ & 146,296,967 \\
$Y_N$ & $C_{1}$ & $(1,0,0,0,0,1)$ & $(4,4,8,4,8,3)$ & 223,238,875 \\
$Y_N$ & $C_{2}$ & $(1,0,0,0,0,2)$ & $(4,4,8,4,8,4)$ & 292,582,832 \\
$Y_N$ & $F_{11}$ & $(-1,-2,-2,-2,-2,0)$ & $(2,2,6,2,6,2)$ & 150,042,476 \\
$Y_N$ & $F_{12}$ & $(-1,-2,-2,-4,-8,0)$ & $(2,2,6,0,0,2)$ & 138,696,204 \\
$Y_N$ & $F_{22}$ & $(-1,-4,-8,-4,-8,0)$ & $(2,0,0,0,0,2)$ & 138,687,916 \\
$Z_N$ & $R_{11}$ & $(0,1,1,1,1,0)$ & $(3,5,9,5,9,2)$ & 155,821,137 \\
$Z_N$ & $R_{12}$ & $(0,1,1,2,4,0)$ & $(3,5,9,6,12,2)$ & 161,494,273 \\
$Z_N$ & $R_{22}$ & $(0,2,4,2,4,0)$ & $(3,6,12,6,12,2)$ & 161,498,417 \\
$Z_N$ & $D_{111}$ & $(-2,-1,-1,-1,-1,-1)$ & $(1,3,7,3,7,1)$ & 82,624,738 \\
$Z_N$ & $D_{222}$ & $(-2,-2,-4,-2,-4,-2)$ & $(1,2,4,2,4,0)$ & 7,603,501 \\
$Z_N$ & $D_{122}$ & $(-2,-1,-1,-2,-4,-2)$ & $(1,3,7,2,4,0)$ & 7,607,645 \\\bottomrule
\end{longtable}
}

The three classes initially have sizes $5,7,6$.  Add two $X_N$ padding items and one $Z_N$ padding item, all of weight $461{,}684{,}752$, to obtain $m=7$.

\subsubsection{All intended compiler triples and their sums}

For each represented pair, the left, right, and fallback sums are as follows:
\begin{longtable}{@{}c l r@{}}
\caption{Target-sum audit for every intended compiler triple in the ``yes'' instance.}\label{lred:tab:yes-compiler-triples}\\
\toprule
Pair or edge & Triple & Sum\\
\midrule
\endfirsthead
\toprule
Pair or edge & Triple & Sum\\
\midrule
\endhead
$(1,1)$ & $(A_1,B_1,R_{11})$ & $153{,}893{,}511+151{,}970{,}103+155{,}821{,}137=461{,}684{,}751$\\
$(1,1,1)$ & $(L_{11},C_1,D_{111})$ & $155{,}821{,}138+223{,}238{,}875+82{,}624{,}738=461{,}684{,}751$\\
$(1,1)$ & $(L_{11},F_{11},R_{11})$ & $155{,}821{,}138+150{,}042{,}476+155{,}821{,}137=461{,}684{,}751$\\
$(1,2)$ & $(A_1,B_2,R_{12})$ & $153{,}893{,}511+146{,}296{,}967+161{,}494{,}273=461{,}684{,}751$\\
$(1,2,2)$ & $(L_{12},C_2,D_{122})$ & $161{,}494{,}274+292{,}582{,}832+7{,}607{,}645=461{,}684{,}751$\\
$(1,2)$ & $(L_{12},F_{12},R_{12})$ & $161{,}494{,}274+138{,}696{,}204+161{,}494{,}273=461{,}684{,}751$\\
$(2,2)$ & $(A_2,B_2,R_{22})$ & $153{,}889{,}367+146{,}296{,}967+161{,}498{,}417=461{,}684{,}751$\\
$(2,2,2)$ & $(L_{22},C_2,D_{222})$ & $161{,}498{,}418+292{,}582{,}832+7{,}603{,}501=461{,}684{,}751$\\
$(2,2)$ & $(L_{22},F_{22},R_{22})$ & $161{,}498{,}418+138{,}687{,}916+161{,}498{,}417=461{,}684{,}751$\\
\bottomrule
\end{longtable}

The source matching $M_Y$ completes pairs $(1,1)$ and $(2,2)$ and leaves pair $(1,2)$ in fallback mode.  The resulting numerical matching is
\begin{align*}
 \mathcal N_Y=\{&
 (A_1,B_1,R_{11}),
 (L_{11},C_1,D_{111}),\\
 & (A_2,B_2,R_{22}),
 (L_{22},C_2,D_{222}),\\
 & (L_{12},F_{12},R_{12})\}.
\end{align*}
The five triples are item-disjoint.  By the exact compiler identity,
\begin{equation}\label{lred:eq:yes-N-opt}
  \OPT_{\mathrm N}=r+k_Y^\star=3+2=5.
\end{equation}

\subsubsection{All one-live groups}

\Cref{lred:tab:yes-groups} displays all 21 separated groups.  The ``nominal live role'' records where the padded compiler item was placed.  A nominal item equal to $\barrier$ is still unusable, so groups 6, 7, and 21 actually have no live port.

\begin{longtable}{@{}r c c r r r@{}}
\caption{All separated groups for the ``yes'' instance.}\label{lred:tab:yes-groups}\\
\toprule
Group & Nominal live role & Compiler item & $a_i$ & $b_i$ & $c_i$\\
\midrule
\endfirsthead
\toprule
Group & Nominal live role & Compiler item & $a_i$ & $b_i$ & $c_i$\\
\midrule
\endhead
1 & $P_1$ & $A_{1}$ & 153,893,511 & 461,684,752 & 461,684,752 \\
2 & $P_1$ & $A_{2}$ & 153,889,367 & 461,684,752 & 461,684,752 \\
3 & $P_1$ & $L_{11}$ & 155,821,138 & 461,684,752 & 461,684,752 \\
4 & $P_1$ & $L_{12}$ & 161,494,274 & 461,684,752 & 461,684,752 \\
5 & $P_1$ & $L_{22}$ & 161,498,418 & 461,684,752 & 461,684,752 \\
6 & $P_1$ & $\barrier^X_{6}$ & 461,684,752 & 461,684,752 & 461,684,752 \\
7 & $P_1$ & $\barrier^X_{7}$ & 461,684,752 & 461,684,752 & 461,684,752 \\
8 & $P_2$ & $B_{1}$ & 461,684,752 & 151,970,103 & 461,684,752 \\
9 & $P_2$ & $B_{2}$ & 461,684,752 & 146,296,967 & 461,684,752 \\
10 & $P_2$ & $C_{1}$ & 461,684,752 & 223,238,875 & 461,684,752 \\
11 & $P_2$ & $C_{2}$ & 461,684,752 & 292,582,832 & 461,684,752 \\
12 & $P_2$ & $F_{11}$ & 461,684,752 & 150,042,476 & 461,684,752 \\
13 & $P_2$ & $F_{12}$ & 461,684,752 & 138,696,204 & 461,684,752 \\
14 & $P_2$ & $F_{22}$ & 461,684,752 & 138,687,916 & 461,684,752 \\
15 & $P_3$ & $R_{11}$ & 461,684,752 & 461,684,752 & 155,821,137 \\
16 & $P_3$ & $R_{12}$ & 461,684,752 & 461,684,752 & 161,494,273 \\
17 & $P_3$ & $R_{22}$ & 461,684,752 & 461,684,752 & 161,498,417 \\
18 & $P_3$ & $D_{111}$ & 461,684,752 & 461,684,752 & 82,624,738 \\
19 & $P_3$ & $D_{222}$ & 461,684,752 & 461,684,752 & 7,603,501 \\
20 & $P_3$ & $D_{122}$ & 461,684,752 & 461,684,752 & 7,607,645 \\
21 & $P_3$ & $\barrier^Z_{7}$ & 461,684,752 & 461,684,752 & 461,684,752 \\\bottomrule
\end{longtable}

This table, together with \cref{lred:tab:sym-role-data,lred:eq:example-sym-shifts,lred:eq:example-LS}, specifies every one of the $13\cdot21=273$ common symmetric item weights.  For example, group 1 has
\[
 (a_1,b_1,c_1)=(153{,}893{,}511,461{,}684{,}752,461{,}684{,}752).
\]
The complete item-level audit for that group appears in \cref{lred:tab:yes-group1-symweights}.

\begin{landscape}
\begin{table}[p]
\centering
\caption{All thirteen symmetric item weights in group 1 of the ``yes'' instance.}
\label{lred:tab:yes-group1-symweights}
\small
\setlength{\tabcolsep}{4pt}
\renewcommand{\arraystretch}{1.14}
\resizebox{\linewidth}{!}{%
\begin{tabular}{@{}c c c r@{}}
\toprule
Role & Signed vector $(\rho,\eta,hi,hi^2)$ & Shifted digits & Scalar weight\\
\midrule
$P_1$ & $(31,461{,}680{,}533,0,0)$ & $(31,2{,}462{,}325{,}559,63,1{,}323)$ & $5{,}582{,}207{,}051{,}525{,}377{,}124{,}449{,}864{,}312{,}837{,}211$\\
$P_2$ & $(8,1{,}385{,}054{,}256,0,0)$ & $(8,3{,}385{,}699{,}282,63,1{,}323)$ & $5{,}582{,}207{,}051{,}525{,}392{,}045{,}276{,}248{,}039{,}303{,}456$\\
$P_3$ & $(261,1{,}385{,}054{,}256,0,0)$ & $(261,3{,}385{,}699{,}282,63,1{,}323)$ & $5{,}582{,}207{,}051{,}525{,}392{,}045{,}276{,}248{,}039{,}303{,}709$\\
$Q_1$ & $(184,-1{,}385{,}054{,}256,2,2)$ & $(184,615{,}590{,}770,65,1{,}325)$ & $5{,}590{,}645{,}762{,}110{,}161{,}102{,}215{,}640{,}446{,}351{,}344$\\
$Q_2$ & $(85,2{,}308{,}427{,}976,-2,-2)$ & $(85,4{,}309{,}073{,}002,61,1{,}321)$ & $5{,}573{,}768{,}340{,}940{,}593{,}146{,}888{,}564{,}581{,}842{,}349$\\
$Q_3$ & $(177,2{,}000{,}636{,}735,-2,-2)$ & $(177,4{,}001{,}281{,}761,61,1{,}321)$ & $5{,}573{,}768{,}340{,}940{,}588{,}173{,}279{,}770{,}006{,}353{,}685$\\
$Q_4$ & $(115,-2{,}000{,}636{,}738,2,2)$ & $(115,8{,}288,65,1{,}325)$ & $5{,}590{,}645{,}762{,}110{,}151{,}154{,}998{,}051{,}295{,}373{,}763$\\
$Q_5$ & $(23,-1{,}692{,}845{,}497,2,2)$ & $(23,307{,}799{,}529,65,1{,}325)$ & $5{,}590{,}645{,}762{,}110{,}156{,}128{,}606{,}845{,}870{,}862{,}427$\\
$Q_6$ & $(16,1{,}692{,}845{,}494,-2,-2)$ & $(16,3{,}693{,}490{,}520,61,1{,}321)$ & $5{,}573{,}768{,}340{,}940{,}583{,}199{,}670{,}975{,}430{,}864{,}768$\\
$Q_7$ & $(27,-615{,}582{,}482,1,1)$ & $(27,1{,}385{,}062{,}544,64,1{,}324)$ & $5{,}586{,}426{,}406{,}817{,}766{,}626{,}528{,}355{,}091{,}849{,}819$\\
$Q_8$ & $(89,3{,}385{,}690{,}991,-3,-3)$ & $(89,5{,}386{,}336{,}017,60,1{,}320)$ & $5{,}569{,}548{,}985{,}648{,}203{,}644{,}810{,}073{,}802{,}829{,}741$\\
$Q_9$ & $(188,-307{,}791{,}241,1,1)$ & $(188,1{,}692{,}853{,}785,64,1{,}324)$ & $5{,}586{,}426{,}406{,}817{,}771{,}600{,}137{,}149{,}667{,}338{,}736$\\
$Q_{10}$ & $(96,0,1,1)$ & $(96,2{,}000{,}645{,}026,64,1{,}324)$ & $5{,}586{,}426{,}406{,}817{,}776{,}573{,}745{,}944{,}242{,}827{,}400$\\\bottomrule
\end{tabular}%
}
\end{table}
\end{landscape}

Direct addition gives $K_S$ for every local pattern.  For example,
\begin{align*}
 W_S(P_{1,1})+W_S(Q_{1,1})+W_S(Q_{2,1})
 &=K_S,\\
 W_S(Q_{1,1})+W_S(Q_{7,1})+W_S(Q_{8,1})
 &=K_S.
\end{align*}
The same exact sum is obtained for all nine patterns $F_1(1),\ldots,F_9(1)$; the vector identities in \cref{lred:tab:sym-role-data} prove the statement for every group.

\subsubsection{The complete symmetric matching and its value}

The five numerical triples in $\mathcal N_Y$ become the following five main hyperedges, where the group indices are read from \cref{lred:tab:yes-groups}:
\begin{align*}
 &(P_{1,1},P_{2,8},P_{3,15}) &&\text{for }(A_1,B_1,R_{11}),\\
 &(P_{1,3},P_{2,10},P_{3,18})&&\text{for }(L_{11},C_1,D_{111}),\\
 &(P_{1,2},P_{2,9},P_{3,17})&&\text{for }(A_2,B_2,R_{22}),\\
 &(P_{1,5},P_{2,11},P_{3,19})&&\text{for }(L_{22},C_2,D_{222}),\\
 &(P_{1,4},P_{2,13},P_{3,16})&&\text{for }(L_{12},F_{12},R_{12}).
\end{align*}
Each displayed main hyperedge has scalar sum exactly $K_S$.  For example,
\begin{align*}
 &W_S(P_{1,1})+W_S(P_{2,8})+W_S(P_{3,15})\\
 &\qquad=16{,}746{,}621{,}154{,}576{,}131{,}373{,}554{,}069{,}341{,}030{,}904.
\end{align*}

Add the canonical twelve fillers in every one of the 21 groups.  The degree-at-most-three hypergraph has
\[
  252+5=257
\]
hyperedges and lifts by recoloring to a legal symmetric matching.  The exact affine theorem gives
\begin{equation}\label{lred:eq:yes-S-opt}
  \OPT_{\mathrm{SN}}=36m+r+k_Y^\star=252+3+2=\boxed{257}.
\end{equation}

\paragraph{Decoder check.}
Here $\Gamma=36m+r=255$.  A symmetric solution of value $257$ decodes to at least $257-255=2$ source edges, so it is source-optimal.  A solution of value $256$ decodes to at least one source edge, and the source error is at most the one-unit destination error.  A solution of value $255$ is guaranteed only the empty source matching; both source and destination errors are two.  These three values illustrate the exact $\beta=1$ inequality.

\subsection{A source-to-output ``no'' instance}\label{lred:subsec:no-example}

Here ``no instance" means that the two-by-two source has no matching of size two.  It still has nonempty feasible matchings, so it exercises the optimization reduction rather than collapsing to the separately handled zero-optimum case.

\subsubsection{Source instance and exact obstruction}

Let
\[
 U=\{u_1,u_2\},\qquad V=\{v_1,v_2\},\qquad W=\{w_1,w_2\},
\]
and let
\begin{equation}\label{lred:eq:no-edges}
 E_N=\{e_{111},e_{122},e_{212}\}
 =\{(u_1,v_1,w_1),(u_1,v_2,w_2),(u_2,v_1,w_2)\}.
\end{equation}
Every pair of source edges intersects:
\begin{itemize}
\item $e_{111}$ and $e_{122}$ share $u_1$;
\item $e_{111}$ and $e_{212}$ share $v_1$;
\item $e_{122}$ and $e_{212}$ share $w_2$.
\end{itemize}
Thus no two source edges form a matching.  Since every single edge is feasible,
\begin{equation}\label{lred:eq:no-source-opt}
 k_N^\star=\OPT_{3\mathrm{DM}}(H_N)=\boxed{1}.
\end{equation}
The represented pair set is
\[
 \Pairset_N=\{(1,1),(1,2),(2,1)\},\qquad r=3.
\]
The instance is therefore a genuine ``no'' instance for a perfect matching of the two vertices in each class, while remaining nontrivial as a maximization instance.

\subsubsection{Every compiler item and scalar weight}

The common shifts and base are those in \cref{lred:subsec:example-common}.  The following table lists every compiler item.  As in the ``yes'' instance, the signed vector, shifted base-$37$ digits, and scalar weight together provide a complete arithmetic specification.

\begin{longtable}{@{}c c p{0.27\textwidth} p{0.27\textwidth} r@{}}
\caption{Every compiler item in the ``no'' instance.}\label{lred:tab:no-items}\\
\toprule
Class & Item & Signed vector & Shifted digits & Scalar weight\\
\midrule
\endfirsthead
\toprule
Class & Item & Signed vector & Shifted digits & Scalar weight\\
\midrule
\endhead
$X_N$ & $A_{1}$ & $(0,-1,-1,0,0,0)$ & $(3,3,7,4,8,2)$ & 153,893,511 \\
$X_N$ & $A_{2}$ & $(0,-2,-4,0,0,0)$ & $(3,2,4,4,8,2)$ & 153,889,367 \\
$X_N$ & $L_{11}$ & $(1,1,1,1,1,0)$ & $(4,5,9,5,9,2)$ & 155,821,138 \\
$X_N$ & $L_{12}$ & $(1,1,1,2,4,0)$ & $(4,5,9,6,12,2)$ & 161,494,274 \\
$X_N$ & $L_{21}$ & $(1,2,4,1,1,0)$ & $(4,6,12,5,9,2)$ & 155,825,282 \\
$Y_N$ & $B_{1}$ & $(0,0,0,-1,-1,0)$ & $(3,4,8,3,7,2)$ & 151,970,103 \\
$Y_N$ & $B_{2}$ & $(0,0,0,-2,-4,0)$ & $(3,4,8,2,4,2)$ & 146,296,967 \\
$Y_N$ & $C_{1}$ & $(1,0,0,0,0,1)$ & $(4,4,8,4,8,3)$ & 223,238,875 \\
$Y_N$ & $C_{2}$ & $(1,0,0,0,0,2)$ & $(4,4,8,4,8,4)$ & 292,582,832 \\
$Y_N$ & $F_{11}$ & $(-1,-2,-2,-2,-2,0)$ & $(2,2,6,2,6,2)$ & 150,042,476 \\
$Y_N$ & $F_{12}$ & $(-1,-2,-2,-4,-8,0)$ & $(2,2,6,0,0,2)$ & 138,696,204 \\
$Y_N$ & $F_{21}$ & $(-1,-4,-8,-2,-2,0)$ & $(2,0,0,2,6,2)$ & 150,034,188 \\
$Z_N$ & $R_{11}$ & $(0,1,1,1,1,0)$ & $(3,5,9,5,9,2)$ & 155,821,137 \\
$Z_N$ & $R_{12}$ & $(0,1,1,2,4,0)$ & $(3,5,9,6,12,2)$ & 161,494,273 \\
$Z_N$ & $R_{21}$ & $(0,2,4,1,1,0)$ & $(3,6,12,5,9,2)$ & 155,825,281 \\
$Z_N$ & $D_{111}$ & $(-2,-1,-1,-1,-1,-1)$ & $(1,3,7,3,7,1)$ & 82,624,738 \\
$Z_N$ & $D_{122}$ & $(-2,-1,-1,-2,-4,-2)$ & $(1,3,7,2,4,0)$ & 7,607,645 \\
$Z_N$ & $D_{212}$ & $(-2,-2,-4,-1,-1,-2)$ & $(1,2,4,3,7,0)$ & 13,276,637 \\
\bottomrule
\end{longtable}

All nine intended compiler triples have sum $T_N=461{,}684{,}751$:
\begin{align*}
 A_1+B_1+R_{11}&=T_N,
 &L_{11}+C_1+D_{111}&=T_N,
 &L_{11}+F_{11}+R_{11}&=T_N,\\
 A_1+B_2+R_{12}&=T_N,
 &L_{12}+C_2+D_{122}&=T_N,
 &L_{12}+F_{12}+R_{12}&=T_N,\\
 A_2+B_1+R_{21}&=T_N,
 &L_{21}+C_2+D_{212}&=T_N,
 &L_{21}+F_{21}+R_{21}&=T_N.
\end{align*}
For instance,
\begin{align*}
 155{,}821{,}138+223{,}238{,}875+82{,}624{,}738
 &=461{,}684{,}751,\\
 161{,}494{,}274+138{,}696{,}204+161{,}494{,}273
 &=461{,}684{,}751.
\end{align*}
The first calculation is a right triple for $e_{111}$; the second is the fallback for pair $(1,2)$.

An optimal compiler matching is
\begin{align}\label{lred:eq:no-N-matching}
 \mathcal N_N=\{&
 (A_1,B_1,R_{11}),
 (L_{11},C_1,D_{111}),\notag\\
 & (L_{12},F_{12},R_{12}),
 (L_{21},F_{21},R_{21})\}.
\end{align}
It completes pair $(1,1)$ and uses fallbacks for the other two represented pairs.  Hence it has four triples.  By \cref{lred:thm:compiler-identity},
\begin{equation}\label{lred:eq:no-N-opt}
 \OPT_{\mathrm N}=r+k_N^\star=3+1=\boxed{4}.
\end{equation}
A five-triple compiler matching would contain at least $5-r=2$ complete pairs and would decode to two disjoint source edges, contradicting \eqref{lred:eq:no-source-opt}.  Thus the numerical obstruction is not asserted merely by counting; it follows from the explicit complete-pair decoder.

\subsubsection{All one-live groups}

The compiler class sizes are again $5$, $7$, and $6$, so $m=7$, $n=21$, and the same barrier and symmetric encoding parameters apply.  Every separated group appears in \cref{lred:tab:no-groups}.

\begin{longtable}{@{}r c c r r r@{}}
\caption{All separated groups for the ``no'' instance.}\label{lred:tab:no-groups}\\
\toprule
Group & Nominal live role & Compiler item & $a_i$ & $b_i$ & $c_i$\\
\midrule
\endfirsthead
\toprule
Group & Nominal live role & Compiler item & $a_i$ & $b_i$ & $c_i$\\
\midrule
\endhead
1 & $P_1$ & $A_{1}$ & 153,893,511 & 461,684,752 & 461,684,752 \\
2 & $P_1$ & $A_{2}$ & 153,889,367 & 461,684,752 & 461,684,752 \\
3 & $P_1$ & $L_{11}$ & 155,821,138 & 461,684,752 & 461,684,752 \\
4 & $P_1$ & $L_{12}$ & 161,494,274 & 461,684,752 & 461,684,752 \\
5 & $P_1$ & $L_{21}$ & 155,825,282 & 461,684,752 & 461,684,752 \\
6 & $P_1$ & $\barrier^X_{6}$ & 461,684,752 & 461,684,752 & 461,684,752 \\
7 & $P_1$ & $\barrier^X_{7}$ & 461,684,752 & 461,684,752 & 461,684,752 \\
8 & $P_2$ & $B_{1}$ & 461,684,752 & 151,970,103 & 461,684,752 \\
9 & $P_2$ & $B_{2}$ & 461,684,752 & 146,296,967 & 461,684,752 \\
10 & $P_2$ & $C_{1}$ & 461,684,752 & 223,238,875 & 461,684,752 \\
11 & $P_2$ & $C_{2}$ & 461,684,752 & 292,582,832 & 461,684,752 \\
12 & $P_2$ & $F_{11}$ & 461,684,752 & 150,042,476 & 461,684,752 \\
13 & $P_2$ & $F_{12}$ & 461,684,752 & 138,696,204 & 461,684,752 \\
14 & $P_2$ & $F_{21}$ & 461,684,752 & 150,034,188 & 461,684,752 \\
15 & $P_3$ & $R_{11}$ & 461,684,752 & 461,684,752 & 155,821,137 \\
16 & $P_3$ & $R_{12}$ & 461,684,752 & 461,684,752 & 161,494,273 \\
17 & $P_3$ & $R_{21}$ & 461,684,752 & 461,684,752 & 155,825,281 \\
18 & $P_3$ & $D_{111}$ & 461,684,752 & 461,684,752 & 82,624,738 \\
19 & $P_3$ & $D_{122}$ & 461,684,752 & 461,684,752 & 7,607,645 \\
20 & $P_3$ & $D_{212}$ & 461,684,752 & 461,684,752 & 13,276,637 \\
21 & $P_3$ & $\barrier^Z_{7}$ & 461,684,752 & 461,684,752 & 461,684,752 \\
\bottomrule
\end{longtable}

The table and \cref{lred:tab:sym-role-data,lred:eq:example-sym-shifts,lred:eq:example-LS} determine all $273$ common symmetric weights.  For instance, group 1 is identical to group 1 in the ``yes'' instance and therefore has the thirteen weights in \cref{lred:tab:yes-group1-symweights}.  Group 5 differs only in its live $P_1$ value: substituting $a_5=155{,}825{,}282$ and $b_5=c_5=\barrier$ into the displayed role formulas determines its complete private block.

\subsubsection{An optimal symmetric solution and the impossibility of one more triple}

The four compiler triples in \eqref{lred:eq:no-N-matching} become the main hyperedges
\begin{align*}
 &(P_{1,1},P_{2,8},P_{3,15})&&\text{for }(A_1,B_1,R_{11}),\\
 &(P_{1,3},P_{2,10},P_{3,18})&&\text{for }(L_{11},C_1,D_{111}),\\
 &(P_{1,4},P_{2,13},P_{3,16})&&\text{for }(L_{12},F_{12},R_{12}),\\
 &(P_{1,5},P_{2,14},P_{3,17})&&\text{for }(L_{21},F_{21},R_{21}).
\end{align*}
Adding twelve canonical fillers in each of the 21 groups produces
\[
 252+4=256
\]
hyperedges, which lift by recoloring to a legal symmetric matching.  Therefore
\begin{equation}\label{lred:eq:no-S-lower}
 \OPT_{\mathrm{SN}}\ge256.
\end{equation}
The exact affine theorem gives the matching upper bound
\begin{equation}\label{lred:eq:no-S-opt}
 \OPT_{\mathrm{SN}}
 =36m+r+k_N^\star
 =252+3+1
 =\boxed{256}.
\end{equation}
There is also a direct decoder explanation.  If a symmetric solution had $257$ triples, then the combined decoder would return at least
\[
 257-(36m+r)=257-255=2
\]
source edges.  No such source matching exists.  Thus the missing $257$th destination triple is exactly the encoded source obstruction, not an artifact of the padding or filler scaffold.

\paragraph{Error ledger.}
Here $\Gamma=36m+r=255$ and the source optimum is one.  A destination solution of value $256$ decodes to at least one source edge and therefore has zero error in both problems.  A destination solution of value $255$ is guaranteed only the empty source matching: its destination error is one and its source error is one.  For every lower destination value, the source error remains at most one while the destination error is larger.  This is the numerical form of $\beta=1$.

\subsection{What the two examples demonstrate}\label{lred:subsec:example-lessons}

The examples use the same dimensions, shifts, bases, barrier value, and symmetric scaffold.  Only one source edge changes.  In the ``yes'' instance, the source edges $e_{111}$ and $e_{222}$ are disjoint, so two pairs can be completed and the symmetric optimum is $257$.  In the ``no'' instance, each pair of source edges conflicts, so only one pair can be completed and the symmetric optimum is $256$.

The fixed destination offset is identical in both cases:
\[
 \Gamma=36m+r=252+3=255.
\]
Thus the one-unit source optimum difference appears as an exactly one-unit destination optimum difference.  This is the intended intuition behind an affine optimization reduction.  The large numerical weights are enforcement devices; they do not distort the optimization gap.

\section{Independent audit of the proof}\label{lred:sec:audit}

This section separates logical proof obligations from arithmetic checks.  The theorem is proved symbolically in the preceding sections.  The finite calculations below were also recomputed independently to detect transcription errors.

\subsection{Audit of the compiler role and moment coordinates}\label{lred:sublred:sec:audit-compiler}

There are two $X_N$ role types, three $Y_N$ role types, and two $Z_N$ role types, hence exactly
\[
 2\cdot3\cdot2=12
\]
class-respecting role combinations.  \Cref{lred:tab:compiler-role-audit} lists all twelve.  Exactly three have zero role sum, and those are precisely the intended left, right, and fallback types.

For fallback triples, the equal-first-moment and equal-second-moment equations can also be written as
\begin{align*}
 i+i''&=2i',\\
 i^2+(i'')^2&=2(i')^2.
\end{align*}
Subtracting half the square of the first equation from the second gives
\[
 \frac{(i-i'')^2}{2}=0,
\]
so $i=i''$ and then $i=i'$.  This alternative calculation confirms the zero-variance argument in \cref{lred:lem:compiler-characterization}.  The same calculation applies to the $j$-indices.

\subsection{Audit of the symmetric role system}\label{lred:sublred:sec:audit-symroles}

The role coordinate has target $300$.  The following complement table gives, for each possible first role, the only unordered pair of remaining roles whose codes complete the target.  It is a compact independent certificate that the main pattern and nine filler patterns are the only role-code solutions.

\begin{longtable}{@{}c p{0.68\textwidth}@{}}
\caption{Role-code complement certificate for target $300$.}\label{lred:tab:role-complements}\\
\toprule
First role & Unordered completing pairs\\
\midrule
\endfirsthead
\toprule
First role & Unordered completing pairs\\
\midrule
\endhead
$P_1$ & $\{P_2,P_3\}$, $\{Q_1,Q_2\}$\\
$P_2$ & $\{P_1,P_3\}$, $\{Q_3,Q_4\}$\\
$P_3$ & $\{P_1,P_2\}$, $\{Q_5,Q_6\}$\\
$Q_1$ & $\{P_1,Q_2\}$, $\{Q_7,Q_8\}$\\
$Q_2$ & $\{P_1,Q_1\}$, $\{Q_7,Q_9\}$\\
$Q_3$ & $\{P_2,Q_4\}$, $\{Q_7,Q_{10}\}$\\
$Q_4$ & $\{P_2,Q_3\}$, $\{Q_8,Q_{10}\}$\\
$Q_5$ & $\{P_3,Q_6\}$, $\{Q_8,Q_9\}$\\
$Q_6$ & $\{P_3,Q_5\}$, $\{Q_9,Q_{10}\}$\\
$Q_7$ & $\{Q_1,Q_8\}$, $\{Q_2,Q_9\}$, $\{Q_3,Q_{10}\}$\\
$Q_8$ & $\{Q_1,Q_7\}$, $\{Q_4,Q_{10}\}$, $\{Q_5,Q_9\}$\\
$Q_9$ & $\{Q_2,Q_7\}$, $\{Q_5,Q_8\}$, $\{Q_6,Q_{10}\}$\\
$Q_{10}$ & $\{Q_3,Q_7\}$, $\{Q_4,Q_8\}$, $\{Q_6,Q_9\}$\\
\bottomrule
\end{longtable}

After duplicate descriptions are removed, the complement table yields exactly ten unordered role multisets: the main pattern and $F_1,\ldots,F_9$.  Since the 13 roles permit
\[
 \binom{13+3-1}{3}=\binom{15}{3}=455
\]
unordered triples with repetition, the certificate excludes the other $445$ role multisets.  The higher coordinates then enforce the numerical target and, for fillers, localization to one group.

\subsection{Audit of the local multiplier certificates}\label{lred:sublred:sec:audit-multipliers}

Let the left sides of the ten capacity inequalities in \eqref{lred:eq:capacity-system} be denoted $C_1,\ldots,C_{10}$.  The four multiplier rows in \cref{lred:tab:multiplier-certificates} have the following exact coefficient products.

\begin{table}[H]
\centering
\caption{Coefficient audit for the local filler certificates.}\label{lred:tab:multiplier-audit}
\small
\begin{tabular}{@{}c c c c@{}}
\toprule
Claim & Coefficients of $(x_1,\ldots,x_9)$ & RHS & Result\\
\midrule
$f_i\le12$ & $(1,1,1,1,1,1,1,1,1)$ & $12$ & exact\\
$f_i-x_1\le10$ & $(0,1,1,1,1,1,1,1,1)$ & $10$ & exact\\
$f_i-x_2\le10$ & $(1,0,1,1,1,1,1,1,1)$ & $10$ & exact\\
$f_i-x_3\le10$ & $(1,1,0,1,1,1,1,1,1)$ & $10$ & exact\\
\bottomrule
\end{tabular}
\end{table}

All multipliers are nonnegative, so every derived inequality is valid.  An exhaustive enumeration of nonnegative integer vectors satisfying \eqref{lred:eq:capacity-system} finds $5{,}149$ feasible vectors.  Among them, the maximum $f_i$ is $12$, and none violates any of
\[
 \pos{2-x_1}\le12-f_i,
 \qquad
 \pos{2-x_2}\le12-f_i,
 \qquad
 \pos{2-x_3}\le12-f_i.
\]
The enumeration is not needed for the proof; it is a finite independent check of the symbolic certificates.

\subsection{Audit of the affine and L-reduction constants}\label{lred:sublred:sec:audit-constants}

\begin{table}[H]
\centering
\caption{Constant ledger.}\label{lred:tab:constant-ledger}
\renewcommand{\arraystretch}{1.15}
\begin{tabular}{@{}p{0.28\textwidth}p{0.22\textwidth}p{0.42\textwidth}@{}}
\toprule
Quantity & Bound or identity & Reason\\
\midrule
Represented pairs & $r\le p$ & Every represented pair occurs in at least one edge\\
Compiler class size & $m\le3p$ & $|X_N|\le2p$, $|Y_N|\le3p$, $|Z_N|\le2p$\\
Separated groups & $n=3m$ & One first-, one second-, and one third-role group per padded index\\
Symmetric scaffold & $12n=36m$ & Twelve canonical fillers per separated group\\
Source size & $p\le7k^\star$ & Maximal matching in degree at most three\\
Destination optimum & $36m+r+k^\star$ & Two exact affine identities\\
First L-bound & $\le108p+p+k^\star$ & Substitute $m\le3p$ and $r\le p$\\
Final first L-bound & $\le764k^\star$ & $109p+k^\star\le109(7k^\star)+k^\star$\\
Second L-bound & $k^\star-k_Y\le\OPT_{\mathrm{SN}}-v$ & Offset $\Gamma=36m+r$ cancels\\
\bottomrule
\end{tabular}
\end{table}

\paragraph{How to read \cref{lred:tab:constant-ledger}.}
Read from top to bottom as a substitution chain.  The first four rows bound the sizes of the compiler and the symmetric scaffold; the source-size row converts the edge count into a multiple of the source optimum; the next rows substitute those bounds into the destination optimum.  The value $764$ is therefore a conservative transparent constant obtained from earlier inequalities, not an optimized feature of the construction.

The arithmetic in the last line of the first-bound calculation is
\[
 109\cdot7+1=763+1=764.
\]
No hidden instance-dependent coefficient occurs.  The constant is intentionally conservative: it is chosen for a transparent proof, not for optimization.

\subsection{Audit of polynomial and unary size}\label{lred:sublred:sec:audit-size}

The size ledger is as follows.
\begin{align*}
 \text{compiler coordinate magnitude}&=O(p^2),\\
 L_N&=O(p^2),\\
 \max\{W_N,T_N\}&=O(p^{12}),\\
 n&=O(p),\\
 \text{symmetric coordinate magnitude}&=O(p^{12}),\\
 L_S&=O(p^{12}),\\
 \max\{W_S,K_S\}&=O(p^{48}),\\
 |S|&=13n=O(p).
\end{align*}
Thus even a literal unary listing has length $O(p^{49})$, up to a constant factor.  This loose exponent is polynomial and is all the strong/unary result requires.  More economical bases would improve the exponent but not the theorem.

\subsection{Algorithmic audit of the decoder}\label{lred:sublred:sec:audit-decoder}

For clarity, the solution map $g$ can be implemented by the following finite sequence.
\begin{enumerate}[label=\arabic*.,leftmargin=2.2em]
\item Forget the three outer symmetric copy labels and classify every selected hyperedge as main or filler using the role table.
\item Compute value multiplicities of selected main incidences in each port role.
\item While some role-value pair exceeds its separated-source multiplicity, delete one main hyperedge containing that excess pair.
\item Assign the surviving equal-valued incidences injectively to labeled padded compiler items.
\item Classify each surviving numerical triple as left, right, or fallback using \cref{lred:lem:compiler-characterization}.
\item For every pair having both a left and a right triple, output the edge named by its right triple.
\end{enumerate}
Sorting or hashing the weights and pair indices implements every step in polynomial time.  The output edges are pairwise disjoint by the item-disjointness of the left and right triples.  The two quantitative guarantees used by the proof are
\[
 s\ge\max\{0,v-36m\}
 \quad\text{and}\quad
 k_Y\ge\max\{0,s-r\},
\]
which compose to \eqref{lred:eq:combined-decoder}.

\section{Lessons learned}\label{lred:sec:lessons}

The first lesson is that perfect-completeness stability and L-reduction stability are different requirements.  A bound measured from the theoretical maximum $13n$ may prove a constant gap and exclude a PTAS, yet still fail to show that an optimal destination solution decodes optimally.  The exact offset identity is the correct target when seeking an L-reduction.

The second lesson is that the obstruction in the original grouped construction is local but genuine.  A block with several useful ports may use one filler deficit to create excess capacity in more than one source role.  Strengthening the global counting argument cannot repair that double charge.  One-live-port separation changes the input arrangement so that the local certificate is charged only once.

The third lesson is that the numerical compiler and the symmetric gadget should be audited as independent modules.  The six-coordinate compiler supplies the first exact offset $r$; the 13-role gadget supplies the second exact offset $36m$.  Their decoders compose because each has a solution-by-solution lower bound, and the two fixed offsets then cancel in the L-reduction error inequality.

\section{Conclusions}\label{lred:sec:conclusion}

Part~III establishes a standard L-reduction from degree-three Maximum Three-Dimensional Matching to unary Maximum Symmetric Numerical Three-Dimensional Matching.  The reduction uses an exact pair compiler, one-live-port separation, and a strengthened local filler certificate.  The resulting exact objective identity is
\[
 \OPT_{\mathrm{SN}}=36m+r+\OPT_{3\mathrm{DM}},
\]
and the explicit L-reduction constants are $\alpha=764$ and $\beta=1$.

The proof clarifies why the stricter reduction format was not an automatic consequence of the Part~II perfect-completeness gap.  The difficulty was not merely a loose constant.  It was the possibility that one defective local block could subsidize excess capacity at several useful ports.  Separating the source roles makes the useful group sets disjoint, after which the certificate $\pos{2-x_r}\le12-f_i$ produces one-for-one error transfer.

The reduction remains polynomial under unary encoding, and the maximal legal-triple algorithm places the destination problem in APX.  Therefore \MaxSNthreeDM is APX-hard under a standard L-reduction, is APX-complete under that convention, and has no PTAS unless $\mathrm P=\mathrm{NP}$.  The last consequence uses the fact that the bounded matching source has no PTAS; it does not follow from the definition of an L-reduction alone.  The two worked examples show concretely that one unit of source optimum changes the destination optimum by exactly one unit while the large numerical scaffold remains fixed.

Taken together, the three parts trace three increasingly demanding robustness questions under numerical symmetry.  Part~I preserves exact feasibility.  Part~II preserves a constant relative gap by proving that a small destination defect can damage only a linearly bounded number of decoded source triples.  Part~III preserves error from the actual optimum, solution by solution, after one-live-port separation.  The principal tutorial lesson is that stronger approximation conclusions require stronger reverse invariants, not merely larger numerical encodings.  Part~II remains the result to cite for the direct endpoint-gap theorem, while Part~III is the result to cite for standard APX-hardness and actual-optimum error transfer.

\clearpage

\newpage
\appendix

\section{Glossary}
\small
\begin{description}[
  style=nextline,
  leftmargin=!,
  labelwidth=\widthof{\bfseries Perfect-completeness gap},
  labelsep=0.5em,
  itemsep=0.7em
]
  \item[Actual and dummy items]
  Items in the first numerical compiler of Part~II that distinguish the selected incidence of a bounded-\ThreeDM vertex from its remaining incidences.  An actual target quadruple records a selected source edge; dummy quadruples dispose of unselected edge incidences.

  \item[APX]
  The class of optimization problems admitting a polynomial-time constant-factor approximation algorithm.  \MaxSNthreeDM belongs to \APX because any maximal legal triple matching is a $3$-approximation.

  \item[Defect]
  The difference between a natural perfect upper bound and the value of a partial solution.  In Part~II, a symmetric solution of value $13n-d$ has defect $d$ and exactly $3d$ missing incidences after labels are forgotten.

  \item[Degree budget]
  The number of incidences available to an underlying occurrence in the intermediate triple system.  The exact construction has degree budget three at every port and private item.

  \item[Filler edge]
  A permitted triple used to satisfy local incidence requirements rather than to encode a selected source triple.  Filler edges consume two incidences of every port and all incidences of every private item in a perfect solution.

  \item[Filler multiplicity]
  The number of hyperedge occurrences of one filler pattern.  The private-item degree equations force these values exactly in Part~I and bound their deviations in Part~II.

  \item[Incidence]
  A membership relation between an item occurrence and a triple.  The incidence ledger records how often each port or private item is used.

  \item[L-reduction]
  A polynomial-time pair of instance and solution maps satisfying one inequality that bounds the destination optimum by a constant multiple of the source optimum and a second inequality that transfers solution error from the destination back to the source.  Part~III proves these inequalities with $\alpha=764$ and $\beta=1$.

  \item[Main edge]
  A permitted triple of the three port roles.  In Part~I it represents one source N3DM triple; in Part~II the surviving main edges are repaired and decoded into a partial source matching; in Part~III it participates in the exact affine decoder after one-live separation.

  \item[Max-SN3DM]
  The maximum-cardinality version of \SNthreeDM.  A feasible solution is any family of pairwise item-disjoint class-transversal triples having the common target sum.

  \item[Occurrence]
  A separately usable item even when another item has the same numerical value.  Reverse arguments must respect multiplicities of occurrences, not merely the set of distinct weights.

  \item[One-live-port separation]
  The Part~III preprocessing that places each useful numerical item in a group where the other two port values are barriers.  Consequently each group can contribute useful capacity in at most one source role, so one missing filler triple cannot be charged several times.

  \item[Pair compiler]
  The exact Part~III transformation that associates a local gadget with each represented $(u_i,v_j)$ pair.  A fallback contributes one numerical triple, whereas a complete pair contributes two and identifies one source edge, giving the affine identity $\OPT_{\mathrm N}=r+\OPT_{3\mathrm{DM}}$.

  \item[Perfect-completeness gap]
  An NP-hard promise distinction between instances attaining a natural upper bound $N(I)$ and instances whose optimum is at most $(1-\varepsilon)N(I)$ for a fixed constant $\varepsilon>0$.

  \item[Port]
  An occurrence representing one possible use of a source item.  Each port has a fixed role type and, in the perfect construction, one incidence not consumed by filler edges.

  \item[Private item]
  An auxiliary occurrence attached to one source group.  Private items close the local filler network and do not participate in main triples.

  \item[Recoloring]
  The assignment of output-class labels to the incidences of an unlabeled target-sum hypergraph by edge-coloring its bipartite incidence multigraph.  The colors determine which physical copy supplies each incidence.

  \item[Role pattern]
  A permitted combination of item types, determined by the role-code coordinate.  The fixed catalogue contains one main pattern and nine filler patterns.

  \item[Target identity]
  A coordinatewise equality showing that a permitted role pattern reaches the vector target.  Mixed-radix packing converts it into an ordinary integer equality.

  \item[Unary encoding]
  Representation in which a numerical value contributes length proportional to its magnitude.  Polynomial bounds on every constructed value are therefore required for unary hardness.
\end{description}
\normalsize

\section{Acknowledgments}
The authors thank the sponsors of their Named Professorship positions, the Orkand Corporation and the Berry Family Fund, respectively.  Generative AI systems were used for drafting assistance, editorial revision, and selected arithmetic or consistency checks.  The authors reviewed the complete manuscript and take responsibility for all definitions, proofs, calculations, and claims.

\section{Citing this work}\label{sec:cite}
This work should be cited as ``Symmetric Numerical Three-Dimensional Matching: Intractability and Inapproximability,'' working paper, Fisher College of Business, The Ohio State University, Columbus, Ohio, 2026.  A citation to the distributed edition should also identify \emph{Version \VersionNumber\ : \VersionDate}.



\end{document}